\documentclass{elsarticle}

\usepackage{caption}
\usepackage{subcaption}
\usepackage{algorithm}
\usepackage{algpseudocode}
\usepackage{textcomp}

\usepackage{amsthm}
\usepackage{amsfonts}
\usepackage{amsmath}
\usepackage{amssymb}
\usepackage{tikz-cd} 
\usepackage{tikz}
\usepackage{pgfplots}

\pgfplotsset{compat=newest} 
\pgfplotsset{plot coordinates/math parser=false} 
\newlength\figureheight 
\newlength\figurewidth

\providecommand{\U}[1]{\protect\rule{.1in}{.1in}}
\newtheorem{theorem}{Theorem}

\newtheorem{definition}[theorem]{Definition}
\newtheorem{example}[theorem]{Example}

\newtheorem{lemma}[theorem]{Lemma}

\newtheorem{remark}[theorem]{Remark}

\input{tcilatex}

\journal{ArXiv}

\begin{document}
%
\begin{frontmatter}%

\title{A Polynomial-Time Algorithm for Unconstrained Binary Quadratic Optimization}%

\author{Juan Ignacio Mulero-Mart\'{i}nez}
\address{Department of Automatic Control, Electrical Engineering and Electronic Technology, Technical University of Cartagena, Campus Muralla del Mar 30203, Spain. \\
E-mail:  juan.mulero@upct.es}

%
%

\begin{abstract}
In this paper, an exact algorithm in polynomial time is developed to solve
unrestricted binary quadratic programs. The computational complexity is
$O\left(  n^{\frac{15}{2}}\right)  $, although very conservative, it is
sufficient to prove that this minimization problem is in the complexity class
$P$. The implementation aspects are also described in detail with a special
emphasis on the transformation of the quadratic program into a linear program
that can be solved in polynomial time. The algorithm was implemented in MATLAB
and checked by generating five million matrices of arbitrary dimensions up
to 30 with random entries in the range $\left[  -50,50\right]  $. All the
experiments carried out have revealed that the method works correctly.
\end{abstract}

\begin{keyword}
Unconstrained binary quadratic programming, global optimization, complexity measures, and classes
\end{keyword}

\end{frontmatter}%

\section{Introduction}

The unconstrained binary quadratic programming (UBQP) problem occurs in many
computer vision, image processing, and pattern recognition applications,
including but not limited to image segmentation/pixel labeling, image
registration/matching, image denoising/restoration, partitioning of graphs,
data clustering, and data classification. Much of the algorithmic progress at
UBQP has been due to the computer vision research community, \cite{969114},
\cite{4204169}, \cite{6619019}, \cite{1262177}. For example, the objective
functions in the UBQP problem are a class of energy functions that are widely
useful and have had very striking success in computer vision (see
\cite{Felzenszwalb2011} for a recent survey).

The UBQP problem dates back to the 1960s where pseudo-boolean functions and
binary quadratic optimization were introduced by Hammer and Rudeanu,
\cite{Hammer1968}. Since then, it has become an active research area in
Discrete Mathematics and Complexity Theory (surveys in \cite{Boros2002} and in
\cite{Hansen1993}, give a good account of this topic).

Currently, this problem has become a major problem in recent years due to the
discovery that UBQP represents a unifying framework for a very wide variety of
combinatorial optimization problems. In particular, as pointed out in
\cite{Kochenberger2006} the UBQP model includes the following important
combinatorial optimization problems: maximum cut problems, maximum, click
problems, maximum independent set problems, graph coloring problems,
satisfiability problems, quadratic knapsack problems, etc.

The UBQP problem is generally NP-Hard, \cite{Pardalos1992} (you can use the
UBQP problem to optimize the number of constraints satisfied on a 0/1 integer
Programming instance, one of the Karp's 21 NP-complete problems). Only a few
special cases are solvable in polynomial time. In fact, the problem of
determining local minima of pseudo-boolean functions is found in the
PLS-complete class (the class of hardest polynomial local search problems),
\cite{Pardalos1992}, \cite{Schaffer1991}, and in general, local search
problems are found in the EXP class, \cite{Tovey1985}, \cite{Tovey1986},
\cite{Hammer1988}, \cite{WilliamsonHoke1988}, \cite{EmamyK1989},
\cite{Tovey2003}. Global optimization methods are NP-complete. To obtain a
global optimal solution by exact methods (generally based on branch and bound
strategies), the following techniques should be highlighted: the combinatorial
variable elimination algorithm\footnote{\bigskip This algorithm is in the
class EXP and only runs in polynomial time for pseudo-Boolean functions
associated with graphs of bounded tree-width.}, \cite{Hammer1968},
\cite{Hammer1963}, \cite{Crama1990}; the continuous relaxation with
linearization (where the requirement of binary variables is replaced by a
weaker restriction of membership to the closed interval $\left[  0,1\right]
$), \cite{Balas1984}, \cite{Balas1978}; the posiform transformations,
\cite{Hammer1984}, \cite{Bourjolly1992}; the conflict graphs (the connection
between posiform minimization problem and the maximum weighted stability),
\cite{Hammer1978}, \cite{Hamor1980}, \cite{Ebenegger1984}, \cite{Alexe2003},
\cite{Hammer1985}, \cite{Hertz1997}; the linearization strategies such as
standard linearization (consisting in transforming the minimization problem
into an equivalent linear 0--1 programming problem), \cite{Dantzig1960},
\cite{Fortet1960}, \cite{Glover1974}, \cite{Hansen1979}, Glover method,
\cite{Glover1975}, improved linearization strategy, \cite{Sherali2006}, (the
reader is referred to \cite{Forrester2020} for a recent comparison of these
methods); semidefinite-based solvers, \cite{Krislock2017} (and the references
therein), et cetera. The reader is referred to the survey
\cite{Kochenberger2014} for a detailed description of these techniques until 2014.

Many researchers have extensively studied the UBQP problem, however, up to 
date nobody has succeeded in developing an algorithm running in polynomial
time. We claim, and this is the main contribution of this work, that UBQP is
in the complexity class $P$. The main idea is to transform the UBQP problem
into a linear programming (LP) problem, that is solved in polynomial time. We
guarantee that the minimum of the LP problem is also the minimum of the UBQP
problem. We also provide the implementation details of the algorithm,
motivated by the following aspects that any work on discrete optimization
should present:

(i) Describe in detail the algorithms to be able to reproduce the experiments
and even improve them in the future.

(ii) Providing the source code so that it is openly available to the
scientific community: Interestingly, a recent study by Dunning has revealed
that only $4\%$ of papers on heuristic methods provide the source code,
\cite{Dunning2018}.

(iii) Establish random test problems with an arbitrary input size. Here it is
important to indicate the ranges of the parameters in the UBQP problem.

This procedure has been implemented in MATLAB (source code is provided as
supplementary material) and checked with five million random matrices up to
dimension $30$, with entries in the range $\left[  -50,50\right]  $.

An advantage of this algorithm is its modularity concerning the dimension of the problem: the set of linear constraints of the equivalent linear programming problem is fixed for a constant arbitrary dimension regardless of the objective function of the quadratic problem. Finally, we highlight that the objective of the work is not the speed of resolution of the problem but simply to show that the UBQP problem can be solved in polynomial time. Future works will analyze large-scale UBQP problems as well as the design of more efficient polynomial-time algorithms.

The paper is organized as follows: Section 2 describes the relaxation process
for the UBQP problem. Next in section 3, the main result about the equivalence
of the UBQP problem with a linear programming problem is addressed. For
simplicity in the exposition, the case $n=3$ is presented first and then it is
generalized for $n>3$. The computational complexity in both time and space is
analyzed in section 4. The implementation features about primary
variables, transformation of the objective function, and convexity and
consistency constraints are treated in section 5. The design of the experiment
for testing the solution is presented in section 6. Finally, section 7 is
dedicated to discussing the main aspects presented in this work as well as
possible future works.

\section{Background}

Let $\mathcal{B}=\left\{  0,1\right\}  $ and $f:\mathcal{B}^{n}\rightarrow
\mathbb{R}$ be a quadratic objective function defined as $f\left(  x\right)
=x^{T}Qx+b^{T}x$ with $Q=Q^{T}\in\mathbb{R}^{n\times n}$, $diag\left(
Q\right)  =\left(  0,\ldots,0\right)  $ and $b\in\mathbb{R}^{n}$. The UBQP
problem is defined as follows:

\begin{description}
\item[\textbf{UBQP}:] $\min_{x\in\mathcal{B}}f\left(  x\right)  $.
\end{description}

The objective function $f$ is usually called a quadratic pseudo-boolean
function, i.e. multilinear polynomials in binary unknowns. These functions
represent a class of energy functions that are widely useful and have had very
striking success in computer vision (see \cite{Felzenszwalb2011} for a recent survey).

This problem can naturally be extended to the solid hypercube $\mathcal{H}%
_{n}=\left[  0,1\right]  ^{n}$ spanned by $\mathcal{B}^{n}$. The extension of
the pseudo-Boolean function $f:\mathcal{B}^{n}\rightarrow\mathbb{R}$ is a
function $f^{pol}:\mathcal{H}_{n}\rightarrow\mathbb{R}$ that coincides with
$f$ at the vertices of $\mathcal{H}_{n}$. Rosenberg discovered an attractive
feature regarding the multilinear polynomial extension $f^{pol}$,
\cite{RO_1972__6_2_95_0}: the minimum of $f^{pol}$ is always attained at a
vertex of $\mathcal{H}_{n}$, and hence, that this minimum coincides with the
minimum of $f$. From this, our optimization problem is reduced to the
following relaxed quadratic problem:

\begin{description}
\item[\textbf{(P}$_{n}$\textbf{):}] $\min_{x\in\mathcal{H}_{n}}f\left(  x\right)  $.
\end{description}

\section{Main Result}

In this section, we prove that Problem (P) can be reduced to a Linear
Programming Problem.

\subsection{A Simple Case}

We begin with the simple case of minimization of a quadratic form $f\left(
x\right)  $ in the cube $\mathcal{H}_{3}$. Here the minimization problem is stated as follows:

\begin{description}
\item[\textbf{(P}$_{3}$\textbf{):}] $\min_{x\in\mathcal{H}_{3}}f\left(
x\right)  $.
\end{description}

Associated with the cube $\mathcal{H}_{3}$ we have a map $\phi:\mathcal{H}%
_{3}\rightarrow\left[  0,2\right]  ^{3}\times\left[  0,\frac{1}{2}\right]
^{3}$ defined as%
\[
\phi\left(  x_{1},x_{2},x_{3}\right)  =\left(
\begin{array}
[c]{c}%
\frac{x_{1}+2x_{1}x_{2}+x_{2}}{2}\\
\frac{x_{1}+2x_{1}x_{3}+x_{3}}{2}\\
\frac{x_{2}+2x_{2}x_{3}+x_{3}}{2}\\
\frac{x_{1}-2x_{1}x_{2}+x_{2}}{2}\\
\frac{x_{1}-2x_{1}x_{3}+x_{3}}{2}\\
\frac{x_{2}-2x_{2}x_{3}+x_{3}}{2}%
\end{array}
\right).
\]

An important fact is that the cube $\mathcal{H}_{3}$ can be expressed as a
convex hull of a finite set of vertices $V=\left\{  0,1\right\}  ^{3}$. For
simplicity, we enumerate the vertices in $V$ as $p_{1},p_{2},\ldots,p_{8}$ so
that $\mathcal{H}_{3}$ can be written as convex combinations of those vertices,
i.e. $\mathcal{H}_{3}=conv\left(  V\right)  $, where%
\[
conv\left(  V\right)  =\left\{  \sum_{i=1}^{8}\alpha_{i}p_{i}:a_{i}%
\geq0\text{, }\sum_{i=1}^{8}\alpha_{i}=1\right\}.
\]

The map $\phi$ is a composition of the maps $\alpha:\mathcal{H}_{3}%
\rightarrow\left[  0,1\right]  ^{6}$ and $\beta:\left[  0,1\right]
^{6}\rightarrow\left[  0,2\right]  ^{3}\times\left[  0,\frac{1}{2}\right]
^{3}$ defined as%
\begin{equation}
\alpha\left(  x\right)  =\left(  x_{1},x_{1}x_{2},x_{1}x_{3},x_{2},x_{2}%
x_{3},x_{3}\right),  \label{EQ7}%
\end{equation}%
\begin{equation}
\beta\left(  y\right)  =E_{3}y\text{ for every }y\in\left[  0,1\right]
^{6},\label{EQ6}%
\end{equation}
where $E_{3}$ is
\begin{equation}
E_{3}=\frac{1}{2}\left(
\begin{array}
[c]{cccccc}%
1 & 2 & 0 & 1 & 0 & 0\\
1 & 0 & 2 & 0 & 0 & 1\\
0 & 0 & 0 & 1 & 2 & 1\\
1 & -2 & 0 & 1 & 0 & 0\\
1 & 0 & -2 & 0 & 0 & 1\\
0 & 0 & 0 & 1 & -2 & 1
\end{array}
\right).  \label{EQ14}%
\end{equation}
More specifically $\phi=\beta\circ\alpha$.

\begin{figure*}
[ptb]
\begin{center}
\tikzset{every picture/.style={line width=0.75pt}} 

\begin{tikzpicture}[x=0.75pt,y=0.75pt,yscale=-1,xscale=1]

\draw    (309,101) -- (309,172.73) ;
\draw [shift={(309,174.73)}, rotate = 270] [color={rgb, 255:red, 0; green, 0; blue, 0 }  ][line width=0.75]    (10.93,-3.29) .. controls (6.95,-1.4) and (3.31,-0.3) .. (0,0) .. controls (3.31,0.3) and (6.95,1.4) .. (10.93,3.29)   ;
\draw    (285.33,175.33) -- (212.75,102.75) ;
\draw [shift={(211.33,101.33)}, rotate = 405] [color={rgb, 255:red, 0; green, 0; blue, 0 }  ][line width=0.75]    (10.93,-3.29) .. controls (6.95,-1.4) and (3.31,-0.3) .. (0,0) .. controls (3.31,0.3) and (6.95,1.4) .. (10.93,3.29)   ;
\draw    (288.33,82.37) -- (252.33,81.63) -- (219.33,81.63) ;
\draw [shift={(217.33,81.63)}, rotate = 360] [color={rgb, 255:red, 0; green, 0; blue, 0 }  ][line width=0.75]    (10.93,-3.29) .. controls (6.95,-1.4) and (3.31,-0.3) .. (0,0) .. controls (3.31,0.3) and (6.95,1.4) .. (10.93,3.29)   ;
\draw    (356,69) .. controls (357.67,70.67) and (357.67,72.33) .. (356,74) .. controls (354.33,75.67) and (354.33,77.33) .. (356,79) .. controls (357.67,80.67) and (357.67,82.33) .. (356,84) .. controls (354.33,85.67) and (354.33,87.33) .. (356,89) .. controls (357.67,90.67) and (357.67,92.33) .. (356,94) .. controls (354.33,95.67) and (354.33,97.33) .. (356,99) .. controls (357.67,100.67) and (357.67,102.33) .. (356,104) .. controls (354.33,105.67) and (354.33,107.33) .. (356,109) .. controls (357.67,110.67) and (357.67,112.33) .. (356,114) .. controls (354.33,115.67) and (354.33,117.33) .. (356,119) .. controls (357.67,120.67) and (357.67,122.33) .. (356,124) .. controls (354.33,125.67) and (354.33,127.33) .. (356,129) .. controls (357.67,130.67) and (357.67,132.33) .. (356,134) .. controls (354.33,135.67) and (354.33,137.33) .. (356,139) .. controls (357.67,140.67) and (357.67,142.33) .. (356,144) .. controls (354.33,145.67) and (354.33,147.33) .. (356,149) .. controls (357.67,150.67) and (357.67,152.33) .. (356,154) .. controls (354.33,155.67) and (354.33,157.33) .. (356,159) .. controls (357.67,160.67) and (357.67,162.33) .. (356,164) .. controls (354.33,165.67) and (354.33,167.33) .. (356,169) -- (356,170.73) -- (356,178.73) ;
\draw [shift={(356,180.73)}, rotate = 270] [color={rgb, 255:red, 0; green, 0; blue, 0 }  ][line width=0.75]    (10.93,-3.29) .. controls (6.95,-1.4) and (3.31,-0.3) .. (0,0) .. controls (3.31,0.3) and (6.95,1.4) .. (10.93,3.29)   ;
\draw    (252,192.73) .. controls (249.64,192.74) and (248.46,191.56) .. (248.46,189.2) .. controls (248.47,186.84) and (247.29,185.66) .. (244.93,185.66) .. controls (242.57,185.67) and (241.39,184.49) .. (241.39,182.13) .. controls (241.4,179.77) and (240.22,178.59) .. (237.86,178.59) .. controls (235.5,178.6) and (234.32,177.42) .. (234.32,175.06) .. controls (234.33,172.7) and (233.15,171.52) .. (230.79,171.52) .. controls (228.43,171.52) and (227.25,170.34) .. (227.25,167.98) .. controls (227.25,165.63) and (226.07,164.45) .. (223.72,164.45) .. controls (221.36,164.45) and (220.18,163.27) .. (220.18,160.91) .. controls (220.18,158.55) and (219,157.37) .. (216.64,157.38) .. controls (214.28,157.38) and (213.1,156.2) .. (213.11,153.84) .. controls (213.11,151.48) and (211.93,150.3) .. (209.57,150.31) .. controls (207.21,150.31) and (206.03,149.13) .. (206.04,146.77) .. controls (206.04,144.41) and (204.86,143.23) .. (202.5,143.24) .. controls (200.14,143.24) and (198.96,142.06) .. (198.97,139.7) .. controls (198.97,137.34) and (197.79,136.16) .. (195.43,136.16) .. controls (193.08,136.16) and (191.9,134.98) .. (191.9,132.63) .. controls (191.9,130.27) and (190.72,129.09) .. (188.36,129.09) .. controls (186,129.1) and (184.82,127.92) .. (184.82,125.56) .. controls (184.83,123.2) and (183.65,122.02) .. (181.29,122.02) .. controls (178.93,122.03) and (177.75,120.85) .. (177.75,118.49) .. controls (177.76,116.13) and (176.58,114.95) .. (174.22,114.95) .. controls (171.86,114.96) and (170.68,113.78) .. (170.68,111.42) -- (168.07,108.8) -- (162.41,103.15) ;
\draw [shift={(161,101.73)}, rotate = 405] [color={rgb, 255:red, 0; green, 0; blue, 0 }  ][line width=0.75]    (10.93,-3.29) .. controls (6.95,-1.4) and (3.31,-0.3) .. (0,0) .. controls (3.31,0.3) and (6.95,1.4) .. (10.93,3.29)   ;
\draw    (343,55.5) .. controls (341.33,57.17) and (339.67,57.17) .. (338,55.5) .. controls (336.33,53.83) and (334.67,53.83) .. (333,55.5) .. controls (331.33,57.17) and (329.67,57.17) .. (328,55.5) .. controls (326.33,53.83) and (324.67,53.83) .. (323,55.5) .. controls (321.33,57.17) and (319.67,57.17) .. (318,55.5) .. controls (316.33,53.83) and (314.67,53.83) .. (313,55.5) .. controls (311.33,57.17) and (309.67,57.17) .. (308,55.5) .. controls (306.33,53.83) and (304.67,53.83) .. (303,55.5) .. controls (301.33,57.17) and (299.67,57.17) .. (298,55.5) .. controls (296.33,53.83) and (294.67,53.83) .. (293,55.5) .. controls (291.33,57.17) and (289.67,57.17) .. (288,55.5) .. controls (286.33,53.83) and (284.67,53.83) .. (283,55.5) .. controls (281.33,57.17) and (279.67,57.17) .. (278,55.5) .. controls (276.33,53.83) and (274.67,53.83) .. (273,55.5) .. controls (271.33,57.17) and (269.67,57.17) .. (268,55.5) .. controls (266.33,53.83) and (264.67,53.83) .. (263,55.5) .. controls (261.33,57.17) and (259.67,57.17) .. (258,55.5) .. controls (256.33,53.83) and (254.67,53.83) .. (253,55.5) .. controls (251.33,57.17) and (249.67,57.17) .. (248,55.5) .. controls (246.33,53.83) and (244.67,53.83) .. (243,55.5) .. controls (241.33,57.17) and (239.67,57.17) .. (238,55.5) .. controls (236.33,53.83) and (234.67,53.83) .. (233,55.5) .. controls (231.33,57.17) and (229.67,57.17) .. (228,55.5) -- (224.1,55.5) -- (216.1,55.5) ;
\draw [shift={(214.1,55.5)}, rotate = 360] [color={rgb, 255:red, 0; green, 0; blue, 0 }  ][line width=0.75]    (10.93,-3.29) .. controls (6.95,-1.4) and (3.31,-0.3) .. (0,0) .. controls (3.31,0.3) and (6.95,1.4) .. (10.93,3.29)   ;
\draw   (377,43.33) -- (631.31,43.33) -- (631.31,193.33) -- (377,193.33) -- cycle ;

\draw (309,82) node    {$\mathcal{H}_{3}$};
\draw (309,188) node    {$\alpha (\mathcal{H}_{3})$};
\draw (198,82) node    {$\mathcal{C}_{n}$};
\draw (325,20) node    {$w =\phi ( x) =\left(\frac{x_{1} +2x_{1} x_{2} +x_{2}}{2} ,\frac{x_{1} +2x_{1} x_{3} +x_{3}}{2} ,\frac{x_{2} +2x_{2} x_{3} +x_{3}}{2} ,\frac{x_{1} -2x_{1} x_{2} +x_{2}}{2} ,\frac{x_{1} -2x_{1} x_{3} +x_{3}}{2} ,\frac{x_{2} -2x_{2} x_{3} +x_{3}}{2}\right)$};
\draw (183,46.5) node [anchor=north west][inner sep=0.75pt]   [align=left] {$\displaystyle w$};
\draw (353,46.5) node [anchor=north west][inner sep=0.75pt]   [align=left] {$\displaystyle x$};
\draw (348,199) node [anchor=north west][inner sep=0.75pt]   [align=left] {$\displaystyle \alpha ( x) =( x_{1} ,x_{1} x_{2} ,x_{1} x_{3} ,x_{2} ,x_{2} x_{3} \ ,x_{3}$)};
\draw (237,140) node    {$\beta $};
\draw (252,191.73) node [anchor=north west][inner sep=0.75pt]   [align=left] {$\displaystyle y$};
\draw (96,73) node [anchor=north west][inner sep=0.75pt]   [align=left] {$\displaystyle \beta ( y) =E_{3} y\ $};
\draw (324,134) node    {$\alpha $};
\draw (263,67) node    {$\phi $};
\draw (93,164) node    {$\frac{1}{2}\begin{pmatrix}
y_{1} +2y_{2} +y_{4}\\
y_{1} +2y_{3} +y_{6}\\
y_{4} +2y_{5} +y_{6}\\
y_{1} -2y_{2} +y_{4}\\
y_{1} -2y_{3} +y_{6}\\
y_{4} -2y_{5} +y_{6}
\end{pmatrix} =\beta ( y)$};
\draw (396.59,50.07) node [anchor=north west][inner sep=0.75pt]   [align=left] {$\displaystyle \beta ( \alpha ( x)) =\phi ( x) =\begin{pmatrix}
u_{12}\\
u_{13}\\
u_{23}\\
v_{12}\\
v_{13}\\
v_{23}
\end{pmatrix}$};

\end{tikzpicture}
\caption{Diagram for the maps $\phi$,
$\alpha$, and $\beta$.}%
\label{FIG2}%
\end{center}
\end{figure*}
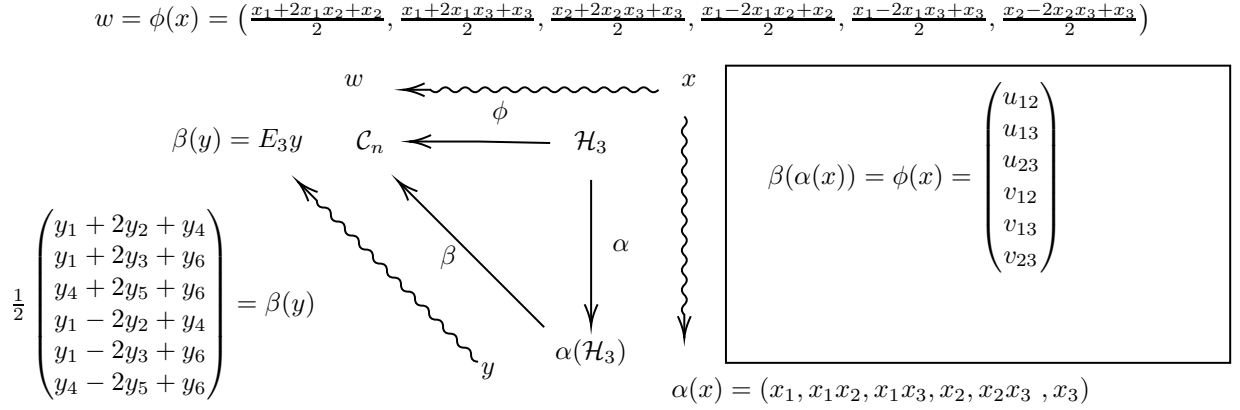

As a summary, the maps $\phi$, $\alpha$, and $\beta$ are represented in the
diagram of Figure \ref{FIG2}. The map $\phi$ is composition of $\alpha$ with $\beta$, i.e. $\phi=\beta
\circ\alpha$, where $\alpha$ can be built from $\mathcal{H}_{3}$ as a
selection of the Kronecker product $\tilde{x}\otimes\tilde{x}$ with $\tilde
{x}^{T}=\left(  1,x^{T}\right)  $ and $x\in\mathcal{H}_{3}$. The set
$\mathcal{\tilde{H}}_{3}=\left\{  \tilde{x}=\left(
\begin{array}
[c]{c}%
1\\
x
\end{array}
\right)  :x\in\mathcal{H}_{3}\right\}  $ is convex: this is trivial simply by
building a convex combination of two points $\tilde{x}$ and $\tilde{y}$ in
$\mathcal{\tilde{H}}_{3}$,%
\[
\lambda\tilde{x}+\left(  1-\lambda\right)  \tilde{y}=\left(
\begin{array}
[c]{c}%
1\\
\lambda x+\left(  1-\lambda\right)  y
\end{array}
\right)  \in\mathcal{\tilde{H}}_{3}\text{ with }\lambda\in\left[  0,1\right]
\text{.}%
\]
Since $\mathcal{H}_{3}=conv\left(  V\right)  $, it follows that $x=\sum
_{i=1}^{8}\lambda_{i}p_{i}$ where $\sum_{i=1}^{8}\lambda_{i}=1$, $\lambda
_{i}\geq0$, and $p_{i}\in V$. The set of vertices of $\mathcal{\tilde{H}}_{3}$
is
\[
\tilde{V}=\left\{  \tilde{p}_{i}=\left(
\begin{array}
[c]{c}%
1\\
p_{i}%
\end{array}
\right)  :i=1,\ldots,8\right\}  \text{.}%
\]
So $\mathcal{\tilde{H}}_{3}=conv\left(  \tilde{V}\right)  $ and $\tilde
{x}\otimes\tilde{x}$ is written as a convex combination:%
\[
\tilde{x}\otimes\tilde{x}=\sum_{i,j=1}^{8}\lambda_{i}\lambda_{j}\left(
\tilde{p}_{i}\otimes\tilde{p}_{j}\right)  \text{.}%
\]
There exists a matrix $S_{3}$ given by%
\[
S_{3}=\frac{1}{2}\left(
\begin{array}
[c]{cccccccccccccccc}%
0 & 1 & 0 & 0 & 1 & 0 & 0 & 0 & 0 & 0 & 0 & 0 & 0 & 0 & 0 & 0\\
0 & 0 & 0 & 0 & 0 & 0 & 1 & 0 & 0 & 1 & 0 & 0 & 0 & 0 & 0 & 0\\
0 & 0 & 0 & 0 & 0 & 0 & 0 & 1 & 0 & 0 & 0 & 0 & 0 & 1 & 0 & 0\\
0 & 0 & 1 & 0 & 0 & 0 & 0 & 0 & 1 & 0 & 0 & 0 & 0 & 0 & 0 & 0\\
0 & 0 & 0 & 0 & 0 & 0 & 0 & 0 & 0 & 0 & 0 & 1 & 0 & 0 & 1 & 0\\
0 & 0 & 0 & 1 & 0 & 0 & 0 & 0 & 0 & 0 & 0 & 0 & 1 & 0 & 0 & 0
\end{array}
\right)  \text{,}%
\]
such that $\alpha\left(  x\right)  =S_{3}\left(  \tilde{x}\otimes\tilde
{x}\right)  $.

\begin{figure*}
[hptb]
\begin{center}
\input{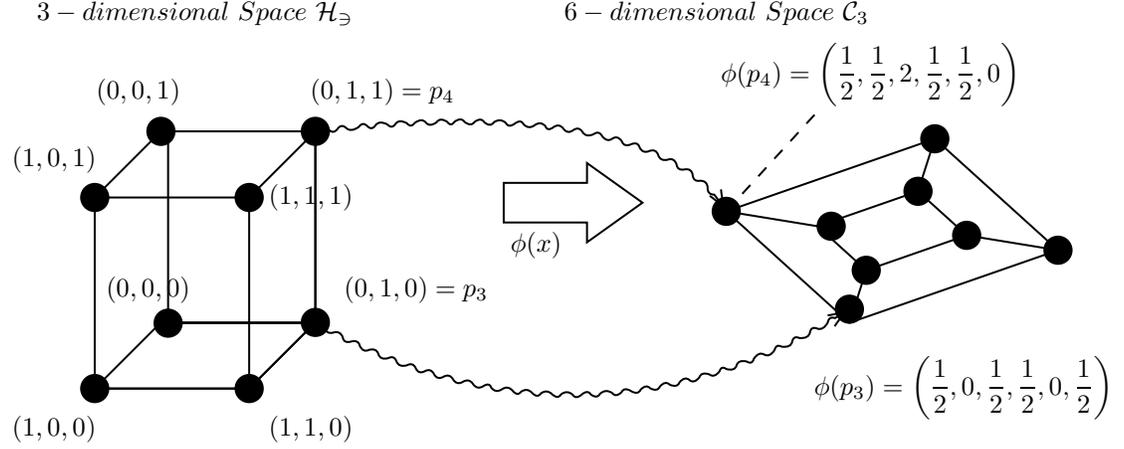}%
\caption{Map $\phi$ between the cube $\mathcal{H}_{3}$ and the 6-dimensional
convex-hull  $\mathcal{C}_{3}$.}%
\label{FIG1}%
\end{center}
\end{figure*}

From the map $\phi$, another convex hull is built
$\mathcal{C}_{3}=conv\left(  \phi\left(  V\right)  \right)  $. In Figure \ref{FIG1},
the transformation between $\mathcal{H}_{3}$ and $\mathcal{C}_{3}$ through the
map $\phi$ is illustrated.%

From $\phi\left(  x\right)  $ we can recover $x$ through the linear
transformation $x=L\phi\left(  x\right)  $, where%
\[
L=\frac{1}{2}\left(
\begin{array}
[c]{cccccc}%
1 & 1 & -1 & 1 & 1 & -1\\
1 & -1 & 1 & 1 & -1 & 1\\
-1 & 1 & 1 & -1 & 1 & 1
\end{array}
\right)  \text{.}%
\]
We have seen that the points of $\mathcal{C}_{3}$ are convex combinations of
the vertices $\phi\left(  p_{i}\right)  $, $i=1,\ldots,8$. 

For simplicity, we write such points as $w=B\lambda$ where matrix $B$ is given
by%
\begin{equation}
B=\left(
\begin{array}
[c]{cccccccc}%
0 & 0 & \frac{1}{2} & \frac{1}{2} & \frac{1}{2} & \frac{1}{2} & 2 & 2\\
0 & \frac{1}{2} & 0 & \frac{1}{2} & \frac{1}{2} & 2 & \frac{1}{2} & 2\\
0 & \frac{1}{2} & \frac{1}{2} & 2 & 0 & \frac{1}{2} & \frac{1}{2} & 2\\
0 & 0 & \frac{1}{2} & \frac{1}{2} & \frac{1}{2} & \frac{1}{2} & 0 & 0\\
0 & \frac{1}{2} & 0 & \frac{1}{2} & \frac{1}{2} & 0 & \frac{1}{2} & 0\\
0 & \frac{1}{2} & \frac{1}{2} & 0 & 0 & \frac{1}{2} & \frac{1}{2} & 0
\end{array}
\right)  \text{,}\label{EQ13}%
\end{equation}
and the vector $\lambda\in\left[  0,1\right]  ^{8}$ is such that $\sum
_{i=1}^{8}\lambda_{i}=1$.

Next we show with a counterexample that $\mathcal{C}_{3}\nsubseteq\phi\left(
\mathcal{H}_{3}\right)  $. For this, let $w$ be the point of $\mathcal{C}_{3}$
given by $w=B\left(  \frac{e_{2}+e_{3}}{2}\right)  $ where $e_{k}$ represents
the standard base vector in $\mathbb{R}^{8}$ with a '1' at the k-th component.
Suppose that there was an $x\in\mathcal{H}_{3}$ such that $\phi\left(
x\right)  =w$, then $x=L\phi\left(  x\right)  =Lw$, and therefore it should be
verified that $\phi\left(  Lw\right)  =w$. However, a contradiction would be
reached since%
\[
\phi\left(  Lw\right)  =\left(
\begin{array}
[c]{c}%
\frac{1}{8}\\
\frac{1}{8}\\
\frac{5}{16}\\
\frac{1}{8}\\
\frac{1}{8}\\
\frac{3}{16}%
\end{array}
\right)  \neq\left(
\begin{array}
[c]{c}%
\frac{1}{4}\\
\frac{1}{4}\\
\frac{1}{2}\\
\frac{1}{4}\\
\frac{1}{4}\\
\frac{1}{2}%
\end{array}
\right)  =w\text{,}%
\]
where $Lw=\left(  0,\frac{1}{4},\frac{1}{4}\right)  ^{T}$. This result reveals
that $\phi\left(  \mathcal{H}_{3}\right)  $ is not convex, because otherwise,
when $\mathcal{C}_{3}$ has its vertices in $\phi\left(  \mathcal{H}%
_{3}\right)  $ it would be that $\mathcal{C}_{3}\subseteq\phi\left(
\mathcal{H}_{3}\right)  $.

The image of $\mathcal{H}_{3}$ through $\phi$ is contained in $\mathcal{C}%
_{3}$.

\begin{lemma}
\label{LEMA2}$\phi\left(  \mathcal{H}_{3}\right)  \subseteq\mathcal{C}_{3}$.
\end{lemma}

\begin{proof}
Matrix $B$ has rank $6$. In fact submatrix $\tilde{B}$ formed by columns
$2,\ldots,7$ from $B$ has full rank. Let $x$ be an arbitrary point in
$\mathcal{H}_{3}$, we will prove that $\phi\left(  x\right)  \in
\mathcal{C}_{3}$. For this, we must find a vector $\lambda\in\left[
0,1\right]  ^{8}$ such that $\Sigma_{i=1}^{8}\lambda_{i}=1$ and%
\begin{equation}
B\lambda=\phi\left(  x\right)  \text{.}\label{EQQ1}%
\end{equation}
For simplicity we will use the notation $\tilde{\lambda}=\left(
\begin{array}
[c]{cccccc}%
\lambda_{2} & \lambda_{3} & \lambda_{4} & \lambda_{5} & \lambda_{6} &
\lambda_{7}%
\end{array}
\right)  ^{T}$. Then the vector $\tilde{\lambda}$ will depend on the point $x$
in $\mathcal{H}_{3}$ and on the free parameter $\lambda_{8}$:%
\[
\tilde{\lambda}=\left(  \tilde{B}^{T}\tilde{B}\right)  ^{-1}\tilde{B}%
^{T}\left(  \phi\left(  x\right)  -B_{8}\lambda_{8}\right)  =\left(
\begin{array}
[c]{c}%
\lambda_{8}+x_{3}-x_{1}x_{3}-x_{2}x_{3}\\
\lambda_{8}+x_{2}-x_{1}x_{2}-x_{2}x_{3}\\
x_{2}x_{3}-\lambda_{8}\\
\lambda_{8}+x_{1}-x_{1}x_{2}-x_{1}x_{3}\\
x_{1}x_{3}-\lambda_{8}\\
x_{1}x_{2}-\lambda_{8}%
\end{array}
\right)  \text{.}%
\]
Furthermore $\lambda_{1}$, $\tilde{\lambda}$, and $\lambda_{8}$ should satisfy
the convexity constraints:%
\begin{equation}
0\leq x_{i}x_{j}-\lambda_{8}\leq1\text{, }i,j\in\left\{  1,2,3\right\}  \text{
and }i\neq j,\label{EQQ2}%
\end{equation}%
\begin{equation}
0\leq\lambda_{8}+x_{i}\left(  x_{j}+x_{k}-1\right)  \leq1\text{, }%
i,j,k\in\left\{  1,2,3\right\}  \text{, }i\neq j\text{, }i\neq k\text{, }j\neq
k\text{,}\label{EQQ3}%
\end{equation}%
\begin{equation}
0\leq\lambda_{1},\lambda_{8}\leq1\text{,}\label{EQQ4}%
\end{equation}%
\begin{equation}
\lambda_{1}+\lambda_{8}+x_{1}+x_{2}+x_{3}-x_{1}x_{2}-x_{1}x_{3}-x_{2}%
x_{3}=1\text{.}\label{EQQ5}%
\end{equation}
\newline The identity $\left(  \ref{EQQ5}\right)  $ is precisely the convexity
constraint  $\Sigma_{i=1}^{8}\lambda_{i}=1$. Let us define the following
quantities:%
\begin{align*}
M_{1}\left(  x\right)    & =\min\left\{  1-x_{1}\left(  1-x_{2}-x_{3}\right)
,1-x_{2}\left(  1-x_{1}-x_{3}\right)  ,1-x_{3}\left(  1-x_{1}-x_{2}\right)
\right\}  \text{,}\\
M_{2}\left(  x\right)    & =\min\left\{  x_{1}x_{2},x_{1}x_{3},x_{2}%
x_{3}\right\}  \text{,}\\
m_{1}\left(  x\right)    & =\max\left\{  -x_{1}\left(  1-x_{2}-x_{3}\right)
,-x_{2}\left(  1-x_{1}-x_{3}\right)  ,-x_{3}\left(  1-x_{1}-x_{2}\right)
\right\}  \text{,}\\
m_{2}\left(  x\right)    & =\max\left\{  x_{1}x_{2}-1,x_{1}x_{3}-1,x_{2}%
x_{3}-1\right\}  \text{.}%
\end{align*}
It can be verified that $m_{1}\left(  x\right)  \leq1$, $m_{2}\left(
x\right)  \leq0$ and that $M_{1}\left(  x\right)  ,M_{2}\left(  x\right)
\geq0$: multilinear polynomials reach their extremes at the vertices of
$\mathcal{H}_{3}$, so%
\begin{align*}
\min_{x\in\mathcal{H}_{3}}1-x_{i}\left(  1-x_{j}-x_{k}\right)    &
=0\text{,}\\
\min_{x\in\mathcal{H}_{3}}x_{i}x_{j}  & =0\text{,}\\
\max_{x\in\mathcal{H}_{3}}-x_{i}\left(  1-x_{j}-x_{k}\right)    & =1\text{,}\\
\max_{x\in\mathcal{H}_{3}}x_{i}x_{j}-1  & =0\text{.}%
\end{align*}
\newline Then,%
\[
\max\left\{  m_{1}\left(  x\right)  ,m_{2}\left(  x\right)  \right\}
\leq\lambda_{8}\leq\min\left\{  M_{1}\left(  x\right)  ,M_{2}\left(  x\right)
\right\}  \text{.}%
\]
On the other hand,%
\[
m_{3}\left(  x\right)  =-\left(  x_{1}+x_{2}+x_{3}-x_{1}x_{2}-x_{1}x_{3}%
-x_{2}x_{3}\right)  \leq\lambda_{8}\leq1-\left(  x_{1}+x_{2}+x_{3}-x_{1}%
x_{2}-x_{1}x_{3}-x_{2}x_{3}\right)  =M_{3}\left(  x\right).
\]
\newline It is immediate to verify that $\min_{x\in\mathcal{H}_{3}}%
M_{3}\left(  x\right)  =\max_{x\in\mathcal{H}_{3}}m_{3}\left(  x\right)  =0$,
which implies that $m_{3}\left(  x\right)  \leq0$ and that $M_{3}\left(
x\right)  \geq0$ in the hypercube $\mathcal{H}_{3}$. Now it is enough to prove
that:\newline(i) $-x_{i}\left(  1-x_{j}-x_{k}\right)  \leq1-\left(
x_{1}+x_{2}+x_{3}-x_{1}x_{2}-x_{1}x_{3}-x_{2}x_{3}\right)  $: This boils down
to simply analyzing the case $\left(  i,j,k\right)  =\left(  1,2,3\right)  $
since the polynomial that appears on the right side is symmetric (for the
cases $\left(  i,j,k\right)  =\left(  2,1,3\right)  $ and $\left(
i,j,k\right)  =\left(  3,1,2\right)  $ the proof would be identical). We will
analyze the sign of the multilinear polynomial,%
\[
p\left(  x_{1},x_{2},x_{3}\right)  =-x_{1}\left(  1-x_{2}-x_{3}\right)
-\left(  1-\left(  x_{1}+x_{2}+x_{3}-x_{1}x_{2}-x_{1}x_{3}-x_{2}x_{3}\right)
\right)  \text{.}%
\]
The maximum of $p$ in the hypercube $\mathcal{H}_{3}$ is reached at one of its
vertices with value $0$. With this, we have shown that%
\[
m_{1}\left(  x\right)  \leq1-\left(  x_{1}+x_{2}+x_{3}-x_{1}x_{2}-x_{1}%
x_{3}-x_{2}x_{3}\right)  \text{.}%
\]
\newline(ii) $x_{i}x_{j}-1\leq1-\left(  x_{1}+x_{2}+x_{3}-x_{1}x_{2}%
-x_{1}x_{3}-x_{2}x_{3}\right)  $. By the same argument as in (i) we reduce the
problem to the pair $\left(  i,j\right)  =\left(  1,2\right)  $ and simply
analyze the multilinear polynomial,%
\[
p\left(  x_{1},x_{2},x_{3}\right)  =x_{1}x_{2}-1-\left(  1-\left(  x_{1}%
+x_{2}+x_{3}-x_{1}x_{2}-x_{1}x_{3}-x_{2}x_{3}\right)  \right)  \text{,}%
\]
whose maximum in $\mathcal{H}_{3}$ is $0$. In this way,%
\[
m_{2}\left(  x\right)  \leq1-\left(  x_{1}+x_{2}+x_{3}-x_{1}x_{2}-x_{1}%
x_{3}-x_{2}x_{3}\right)  \text{.}%
\]
(iii) $-\left(  x_{1}+x_{2}+x_{3}-x_{1}x_{2}-x_{1}x_{3}-x_{2}x_{3}\right)
\leq1-x_{i}\left(  1-x_{j}-x_{k}\right)  $. Again we reduce it to $\left(
i,j,k\right)  =\left(  1,2,3\right)  $, and we construct the multilinear
polynomial:%
\[
p\left(  x_{1},x_{2},x_{3}\right)  =-\left(  x_{1}+x_{2}+x_{3}-x_{1}%
x_{2}-x_{1}x_{3}-x_{2}x_{3}\right)  -\left(  1-x_{1}\left(  1-x_{2}%
-x_{3}\right)  \right)  \text{.}%
\]
It can be verified that $\min_{x\in\mathcal{H}_{3}}p\left(  x_{1},x_{2}%
,x_{3}\right)  =-2$, so
\[
-\left(  x_{1}+x_{2}+x_{3}-x_{1}x_{2}-x_{1}x_{3}-x_{2}x_{3}\right)  \leq
M_{1}\text{.}%
\]
\newline(iv) $-\left(  x_{1}+x_{2}+x_{3}-x_{1}x_{2}-x_{1}x_{3}-x_{2}%
x_{3}\right)  \leq x_{i}x_{j}$. We reduce it to $\left(  i,j\right)  =\left(
1,2\right)  $ and to the polynomial%
\[
p\left(  x_{1},x_{2},x_{3}\right)  =-\left(  x_{1}+x_{2}+x_{3}-x_{1}%
x_{2}-x_{1}x_{3}-x_{2}x_{3}\right)  -x_{1}x_{2}\text{,}%
\]
such that again $\min_{x\in\mathcal{H}_{3}}p\left(  x_{1},x_{2},x_{3}\right)
=-2$. Thus,
\[
-\left(  x_{1}+x_{2}+x_{3}-x_{1}x_{2}-x_{1}x_{3}-x_{2}x_{3}\right)  \leq
M_{2}\text{.}%
\]
\newline With this, we have proven that%
\begin{align*}
\max\left\{  m_{1}\left(  x\right)  ,m_{2}\left(  x\right)  \right\}    &
\leq\lambda_{8}\leq M_{3}\left(  x\right)  \text{,}\\
m_{3}\left(  x\right)    & \leq\lambda_{8}\leq\min\left\{  M_{1}\left(
x\right)  ,M_{2}\left(  x\right)  \right\}  \text{,}%
\end{align*}
and that
\[
\max\left\{  m_{1}\left(  x\right)  ,m_{2}\left(  x\right)  ,m_{3}\left(
x\right)  \right\}  \leq\lambda_{8}\leq\min\left\{  M_{1}\left(  x\right)
,M_{2}\left(  x\right)  ,M_{3}\left(  x\right)  \right\}  =M\left(  x\right)
\text{.}%
\]
Since $m_{2}\left(  x\right)  ,m_{3}\left(  x\right)  \leq0$, we should simply
take $\lambda_{8}\in\left[  \max\left\{  0,m_{1}\left(  x\right)  \right\}
,\min\left\{  M\left(  x\right)  ,1\right\}  \right]  $ and
\[
\lambda_{1}=1-\left(  \lambda_{8}+x_{1}+x_{2}+x_{3}-x_{1}x_{2}-x_{1}%
x_{3}-x_{2}x_{3}\right)  \text{.}%
\]
\newline
\end{proof}

Since the previous proof is constructive, in the following examples we will
show how $\lambda\in\left[  0,1\right]  ^{8}$ can be constructed from
$x\in\mathcal{H}_{3}$. We must emphasize that since $\lambda_{8}$ moves in a
permitted interval $\left[  \max\left\{  0,m_{1}\left(  x\right)  \right\}
,\min\left\{  M\left(  x\right)  ,1\right\}  \right]  $, in general, there can
be many solutions.

\begin{example}
Let $x=\left(
\begin{array}
[c]{ccc}%
1 & \frac{1}{2} & \frac{1}{2}%
\end{array}
\right)  ^{T}$. According to the notation in Lemma \ref{LEMA2}:%
\begin{align*}
M_{1}\left(  x\right)    & =\min\left\{  1-x_{1}\left(  1-x_{2}-x_{3}\right)
,1-x_{2}\left(  1-x_{1}-x_{3}\right)  ,1-x_{3}\left(  1-x_{1}-x_{2}\right)
\right\}  =\\
& =\min\left\{  1,\frac{5}{4},\frac{5}{4}\right\}  =1\text{,}\\
M_{2}\left(  x\right)    & =\min\left\{  x_{1}x_{2},x_{1}x_{3},x_{2}%
x_{3}\right\}  =\min\left\{  \frac{1}{2},\frac{1}{2},\frac{1}{4}\right\}
=\frac{1}{4}\text{,}\\
m_{1}\left(  x\right)    & =\max\left\{  -x_{1}\left(  1-x_{2}-x_{3}\right)
,-x_{2}\left(  1-x_{1}-x_{3}\right)  ,-x_{3}\left(  1-x_{1}-x_{2}\right)
\right\}  =\\
& =\max\left\{  0,\frac{1}{4},\frac{1}{4}\right\}  =\frac{1}{4}\text{,}\\
M_{3}\left(  x\right)    & =1-\left(  x_{1}+x_{2}+x_{3}-x_{1}x_{2}-x_{1}%
x_{3}-x_{2}x_{3}\right)  =\frac{1}{4}\text{,}%
\end{align*}
and as a result $M\left(  x\right)  =\min\left\{  M_{1}\left(  x\right)
,M_{2}\left(  x\right)  ,M_{3}\left(  x\right)  \right\}  =\frac{1}{4}$, and
$\lambda_{8}=\frac{1}{4}$, $\lambda_{1}=0$. With this selection, we will have
\[
\lambda=\left(
\begin{array}
[c]{c}%
\lambda_{1}\\
\lambda_{8}+x_{3}-x_{1}x_{3}-x_{2}x_{3}\\
\lambda_{8}+x_{2}-x_{1}x_{2}-x_{2}x_{3}\\
x_{2}x_{3}-\lambda_{8}\\
\lambda_{8}+x_{1}-x_{1}x_{2}-x_{1}x_{3}\\
x_{1}x_{3}-\lambda_{8}\\
x_{1}x_{2}-\lambda_{8}\\
\lambda_{8}%
\end{array}
\right)  =\left(
\begin{array}
[c]{c}%
0\\
0\\
0\\
0\\
\frac{1}{4}\\
\frac{1}{4}\\
\frac{1}{4}\\
\frac{1}{4}%
\end{array}
\right)  \text{.}%
\]
For this $\lambda$ we verify that $B\lambda=\phi\left(  x\right)  $:%
\[
B\lambda=\left(
\begin{array}
[c]{cccccccc}%
0 & 0 & \frac{1}{2} & \frac{1}{2} & \frac{1}{2} & \frac{1}{2} & 2 & 2\\
0 & \frac{1}{2} & 0 & \frac{1}{2} & \frac{1}{2} & 2 & \frac{1}{2} & 2\\
0 & \frac{1}{2} & \frac{1}{2} & 2 & 0 & \frac{1}{2} & \frac{1}{2} & 2\\
0 & 0 & \frac{1}{2} & \frac{1}{2} & \frac{1}{2} & \frac{1}{2} & 0 & 0\\
0 & \frac{1}{2} & 0 & \frac{1}{2} & \frac{1}{2} & 0 & \frac{1}{2} & 0\\
0 & \frac{1}{2} & \frac{1}{2} & 0 & 0 & \frac{1}{2} & \frac{1}{2} & 0
\end{array}
\right)  \left(
\begin{array}
[c]{c}%
0\\
0\\
0\\
0\\
\frac{1}{4}\\
\frac{1}{4}\\
\frac{1}{4}\\
\frac{1}{4}%
\end{array}
\right)  =\left(
\begin{array}
[c]{c}%
\frac{5}{4}\\
\frac{5}{4}\\
\frac{3}{4}\\
\frac{1}{4}\\
\frac{1}{4}\\
\frac{1}{4}%
\end{array}
\right)  =\phi\left(  x\right)  \text{.}%
\]

\end{example}

\begin{example}
Let $\left(  x_{1},x_{2},x_{3}\right)  =\left(  0,\frac{1}{4},\frac{1}%
{4}\right)  $:%
\begin{align*}
M_{1}\left(  x\right)    & =\min\left\{  1-x_{1}\left(  1-x_{2}-x_{3}\right)
,1-x_{2}\left(  1-x_{1}-x_{3}\right)  ,1-x_{3}\left(  1-x_{1}-x_{2}\right)
\right\}  =\\
& =\min\left\{  1,\frac{13}{16},\frac{13}{16}\right\}  =\frac{13}{16}%
\text{,}\\
M_{2}\left(  x\right)    & =\min\left\{  x_{1}x_{2},x_{1}x_{3},x_{2}%
x_{3}\right\}  =\min\left\{  0,0,\frac{1}{16}\right\}  =0\text{,}\\
m_{1}\left(  x\right)    & =\max\left\{  -x_{1}\left(  1-x_{2}-x_{3}\right)
,-x_{2}\left(  1-x_{1}-x_{3}\right)  ,-x_{3}\left(  1-x_{1}-x_{2}\right)
\right\}  =\\
& =\max\left\{  0,-\frac{3}{16},-\frac{3}{16}\right\}  =0\text{,}\\
M_{3}\left(  x\right)    & =1-\left(  x_{1}+x_{2}+x_{3}-x_{1}x_{2}-x_{1}%
x_{3}-x_{2}x_{3}\right)  =\frac{9}{16}\text{,}%
\end{align*}
and as a result $M\left(  x\right)  =\min\left\{  M_{1}\left(  x\right)
,M_{2}\left(  x\right)  ,M_{3}\left(  x\right)  \right\}  =0$, and
$\lambda_{8}=0$, $\lambda_{1}=\frac{9}{16}$. We this selection we obtain the
vector $\lambda$:%
\[
\lambda=\left(
\begin{array}
[c]{c}%
\lambda_{1}\\
\lambda_{8}+x_{3}-x_{1}x_{3}-x_{2}x_{3}\\
\lambda_{8}+x_{2}-x_{1}x_{2}-x_{2}x_{3}\\
x_{2}x_{3}-\lambda_{8}\\
\lambda_{8}+x_{1}-x_{1}x_{2}-x_{1}x_{3}\\
x_{1}x_{3}-\lambda_{8}\\
x_{1}x_{2}-\lambda_{8}\\
\lambda_{8}%
\end{array}
\right)  =\left(
\begin{array}
[c]{c}%
\frac{9}{16}\\
0\\
0\\
0\\
\frac{1}{4}\\
\frac{1}{4}\\
\frac{1}{4}\\
\frac{1}{4}%
\end{array}
\right)  \text{.}%
\]
Now for this $\lambda$, we check that $B\lambda=\phi\left(  x\right)  $:%
\[
B\lambda=\left(
\begin{array}
[c]{cccccccc}%
0 & 0 & \frac{1}{2} & \frac{1}{2} & \frac{1}{2} & \frac{1}{2} & 2 & 2\\
0 & \frac{1}{2} & 0 & \frac{1}{2} & \frac{1}{2} & 2 & \frac{1}{2} & 2\\
0 & \frac{1}{2} & \frac{1}{2} & 2 & 0 & \frac{1}{2} & \frac{1}{2} & 2\\
0 & 0 & \frac{1}{2} & \frac{1}{2} & \frac{1}{2} & \frac{1}{2} & 0 & 0\\
0 & \frac{1}{2} & 0 & \frac{1}{2} & \frac{1}{2} & 0 & \frac{1}{2} & 0\\
0 & \frac{1}{2} & \frac{1}{2} & 0 & 0 & \frac{1}{2} & \frac{1}{2} & 0
\end{array}
\right)  \left(
\begin{array}
[c]{c}%
\frac{9}{16}\\
\frac{3}{16}\\
\frac{3}{16}\\
\frac{1}{16}\\
0\\
0\\
0\\
0
\end{array}
\right)  =\left(
\begin{array}
[c]{c}%
\frac{1}{8}\\
\frac{1}{8}\\
\frac{5}{16}\\
\frac{1}{8}\\
\frac{1}{8}\\
\frac{3}{16}%
\end{array}
\right)  =\phi\left(  x\right)  \text{.}%
\]

\end{example}

\begin{example}
For the point $x=\left(
\begin{array}
[c]{ccc}%
1 & \frac{1}{2} & 0
\end{array}
\right)  ^{T}$ we compute the quantities $M_{i}\left(  x\right)  $ and
$m_{i}\left(  x\right)  $:\newline%
\begin{align*}
M_{1}\left(  x\right)    & =\min\left\{  1-x_{1}\left(  1-x_{2}-x_{3}\right)
,1-x_{2}\left(  1-x_{1}-x_{3}\right)  ,1-x_{3}\left(  1-x_{1}-x_{2}\right)
\right\}  =\\
& =\min\left\{  \frac{1}{2},1,1\right\}  =\frac{1}{2}\text{,}\\
M_{2}\left(  x\right)    & =\min\left\{  x_{1}x_{2},x_{1}x_{3},x_{2}%
x_{3}\right\}  =\min\left\{  \frac{1}{2},0,0\right\}  =0\text{,}\\
m_{1}\left(  x\right)    & =\max\left\{  -x_{1}\left(  1-x_{2}-x_{3}\right)
,-x_{2}\left(  1-x_{1}-x_{3}\right)  ,-x_{3}\left(  1-x_{1}-x_{2}\right)
\right\}  =\\
& =\max\left\{  -\frac{1}{2},0,0\right\}  =0\text{,}\\
M_{3}\left(  x\right)    & =1-\left(  x_{1}+x_{2}+x_{3}-x_{1}x_{2}-x_{1}%
x_{3}-x_{2}x_{3}\right)  =0\text{,}%
\end{align*}
which implies that $M\left(  x\right)  =\min\left\{  M_{1}\left(  x\right)
,M_{2}\left(  x\right)  ,M_{3}\left(  x\right)  \right\}  =0$, and
$\lambda_{8}=0$, $\lambda_{1}=0$. Using this selection we have that%
\[
\lambda=\left(
\begin{array}
[c]{c}%
\lambda_{1}\\
\lambda_{8}+x_{3}-x_{1}x_{3}-x_{2}x_{3}\\
\lambda_{8}+x_{2}-x_{1}x_{2}-x_{2}x_{3}\\
x_{2}x_{3}-\lambda_{8}\\
\lambda_{8}+x_{1}-x_{1}x_{2}-x_{1}x_{3}\\
x_{1}x_{3}-\lambda_{8}\\
x_{1}x_{2}-\lambda_{8}\\
\lambda_{8}%
\end{array}
\right)  =\left(
\begin{array}
[c]{c}%
0\\
0\\
0\\
0\\
\frac{1}{2}\\
0\\
\frac{1}{2}\\
0
\end{array}
\right)  \text{.}%
\]
For this $\lambda$ we verify that $B\lambda=\phi\left(  x\right)  $:%
\[
B\lambda=\left(
\begin{array}
[c]{cccccccc}%
0 & 0 & \frac{1}{2} & \frac{1}{2} & \frac{1}{2} & \frac{1}{2} & 2 & 2\\
0 & \frac{1}{2} & 0 & \frac{1}{2} & \frac{1}{2} & 2 & \frac{1}{2} & 2\\
0 & \frac{1}{2} & \frac{1}{2} & 2 & 0 & \frac{1}{2} & \frac{1}{2} & 2\\
0 & 0 & \frac{1}{2} & \frac{1}{2} & \frac{1}{2} & \frac{1}{2} & 0 & 0\\
0 & \frac{1}{2} & 0 & \frac{1}{2} & \frac{1}{2} & 0 & \frac{1}{2} & 0\\
0 & \frac{1}{2} & \frac{1}{2} & 0 & 0 & \frac{1}{2} & \frac{1}{2} & 0
\end{array}
\right)  \left(
\begin{array}
[c]{c}%
0\\
0\\
0\\
0\\
\frac{1}{2}\\
0\\
\frac{1}{2}\\
0
\end{array}
\right)  =\left(
\begin{array}
[c]{c}%
\frac{5}{4}\\
\frac{1}{2}\\
\frac{1}{4}\\
\frac{1}{4}\\
\frac{1}{2}\\
\frac{1}{4}%
\end{array}
\right)  =\phi\left(  x\right).
\]

\end{example}

The following example shows that the associated vector $\lambda$ for a
$\phi\left(  x\right)  $ is not unique.

\begin{example}
\label{EJEMPLO1}
Let $x=\left(
\begin{array}
[c]{ccc}%
\frac{1}{4} & \frac{1}{4} & \frac{1}{4}%
\end{array}
\right)  ^{T}$:%
\begin{align*}
M_{1}\left(  x\right)    & =\min\left\{  1-x_{1}\left(  1-x_{2}-x_{3}\right)
,1-x_{2}\left(  1-x_{1}-x_{3}\right)  ,1-x_{3}\left(  1-x_{1}-x_{2}\right)
\right\}  =\\
& =\min\left\{  \frac{7}{8},\frac{7}{8},\frac{7}{8}\right\}  =\frac{7}%
{8}\text{,}\\
M_{2}\left(  x\right)    & =\min\left\{  x_{1}x_{2},x_{1}x_{3},x_{2}%
x_{3}\right\}  =\min\left\{  \frac{1}{16},\frac{1}{16},\frac{1}{16}\right\}
=\frac{1}{16}\text{,}\\
m_{1}\left(  x\right)    & =\max\left\{  -x_{1}\left(  1-x_{2}-x_{3}\right)
,-x_{2}\left(  1-x_{1}-x_{3}\right)  ,-x_{3}\left(  1-x_{1}-x_{2}\right)
\right\}  =\\
& =\max\left\{  -\frac{1}{8},-\frac{1}{8},-\frac{1}{8}\right\}  =-\frac{1}%
{8}\text{,}\\
M_{3}\left(  x\right)    & =1-\left(  x_{1}+x_{2}+x_{3}-x_{1}x_{2}-x_{1}%
x_{3}-x_{2}x_{3}\right)  =\frac{7}{16}\text{,}%
\end{align*}
and as a consequence, $M\left(  x\right)  =\min\left\{  M_{1}\left(  x\right)
,M_{2}\left(  x\right)  ,M_{3}\left(  x\right)  \right\}  =\frac{1}{16}$, and
$\lambda_{8}\in\left[  0,\frac{1}{16}\right]  $. \newline If we select
$\lambda_{8}=\frac{1}{32}$ then $\lambda_{1}=\frac{13}{32}$ and we will have
the following vector $\lambda$:%
\[
\lambda=\left(
\begin{array}
[c]{c}%
\lambda_{1}\\
\lambda_{8}+x_{3}-x_{1}x_{3}-x_{2}x_{3}\\
\lambda_{8}+x_{2}-x_{1}x_{2}-x_{2}x_{3}\\
x_{2}x_{3}-\lambda_{8}\\
\lambda_{8}+x_{1}-x_{1}x_{2}-x_{1}x_{3}\\
x_{1}x_{3}-\lambda_{8}\\
x_{1}x_{2}-\lambda_{8}\\
\lambda_{8}%
\end{array}
\right)  =\left(
\begin{array}
[c]{c}%
\frac{13}{32}\\
\frac{5}{32}\\
\frac{5}{32}\\
\frac{1}{32}\\
\frac{5}{32}\\
\frac{1}{32}\\
\frac{1}{32}\\
\frac{1}{32}%
\end{array}
\right)  \text{.}%
\]
For this $\lambda$ we check that $B\lambda=\phi\left(  x\right)  $:%
\[
B\lambda=\left(
\begin{array}
[c]{cccccccc}%
0 & 0 & \frac{1}{2} & \frac{1}{2} & \frac{1}{2} & \frac{1}{2} & 2 & 2\\
0 & \frac{1}{2} & 0 & \frac{1}{2} & \frac{1}{2} & 2 & \frac{1}{2} & 2\\
0 & \frac{1}{2} & \frac{1}{2} & 2 & 0 & \frac{1}{2} & \frac{1}{2} & 2\\
0 & 0 & \frac{1}{2} & \frac{1}{2} & \frac{1}{2} & \frac{1}{2} & 0 & 0\\
0 & \frac{1}{2} & 0 & \frac{1}{2} & \frac{1}{2} & 0 & \frac{1}{2} & 0\\
0 & \frac{1}{2} & \frac{1}{2} & 0 & 0 & \frac{1}{2} & \frac{1}{2} & 0
\end{array}
\right)  \left(
\begin{array}
[c]{c}%
\frac{13}{32}\\
\frac{5}{32}\\
\frac{5}{32}\\
\frac{1}{32}\\
\frac{5}{32}\\
\frac{1}{32}\\
\frac{1}{32}\\
\frac{1}{32}%
\end{array}
\right)  =\left(
\begin{array}
[c]{c}%
\frac{5}{16}\\
\frac{5}{16}\\
\frac{5}{16}\\
\frac{3}{16}\\
\frac{3}{16}\\
\frac{3}{16}%
\end{array}
\right)  =\phi\left(  x\right)  \text{.}%
\]
On the other hand, with the selection $\lambda_{8}=\frac{1}{16}$, we would
have $\lambda_{1}=\frac{3}{8}$, and therefore%
\[
\lambda=\left(
\begin{array}
[c]{c}%
\frac{3}{8}\\
\frac{3}{16}\\
\frac{3}{16}\\
0\\
\frac{3}{16}\\
0\\
0\\
\frac{1}{16}%
\end{array}
\right)  \text{.}%
\]
For this $\lambda$ we check again that $B\lambda=\phi\left(  x\right)  $:%
\[
B\lambda=\left(
\begin{array}
[c]{cccccccc}%
0 & 0 & \frac{1}{2} & \frac{1}{2} & \frac{1}{2} & \frac{1}{2} & 2 & 2\\
0 & \frac{1}{2} & 0 & \frac{1}{2} & \frac{1}{2} & 2 & \frac{1}{2} & 2\\
0 & \frac{1}{2} & \frac{1}{2} & 2 & 0 & \frac{1}{2} & \frac{1}{2} & 2\\
0 & 0 & \frac{1}{2} & \frac{1}{2} & \frac{1}{2} & \frac{1}{2} & 0 & 0\\
0 & \frac{1}{2} & 0 & \frac{1}{2} & \frac{1}{2} & 0 & \frac{1}{2} & 0\\
0 & \frac{1}{2} & \frac{1}{2} & 0 & 0 & \frac{1}{2} & \frac{1}{2} & 0
\end{array}
\right)  \left(
\begin{array}
[c]{c}%
\frac{3}{8}\\
\frac{3}{16}\\
\frac{3}{16}\\
0\\
\frac{3}{16}\\
0\\
0\\
\frac{1}{16}%
\end{array}
\right)  =\left(
\begin{array}
[c]{c}%
\frac{5}{16}\\
\frac{5}{16}\\
\frac{5}{16}\\
\frac{3}{16}\\
\frac{3}{16}\\
\frac{3}{16}%
\end{array}
\right)  =\phi\left(  x\right)  \text{.}%
\]

\end{example}

The elements in $\mathcal{C}_{3}$ are called primary variables and are denoted
as $w$. We will make a distinction between primary variables $u\in\left[
0,2\right]  ^{3}$ and primary variables $v\in\left[  0,\frac{1}{2}\right]
^{3}$, which make up the vector $w$ as $w=\left(  u^{T},v^{T}\right)  ^{T}$.

\begin{definition}
[Primary variables]For the triplet of variables $\left(  x_{1},x_{2}%
,x_{3}\right)  \in\mathcal{H}_{3}$ we define the primary variables for $1\leq
i<j\leq3$:%
\begin{equation}
u_{ij}=\frac{x_{i}+2x_{i}x_{j}+x_{j}}{2},\label{EQ1}%
\end{equation}%
\begin{equation}
v_{ij}=\frac{x_{i}-2x_{i}x_{j}+x_{j}}{2},\label{EQ2}%
\end{equation}
where $u_{ij}\in\left[  0,2\right]  $ and $v_{ij}\in\left[  0,\frac{1}%
{2}\right]  $.
\end{definition}

The vector of primary variables $w^{T}=\left(  u^{T},v^{T}\right)  $ is such
that $u^{T}=\left(  u_{12},u_{13},u_{23}\right)  $ is given by $\left(
\ref{EQ1}\right)  $ and $v^{T}=\left(  v_{12},v_{13},v_{23}\right)  $ by
$\left(  \ref{EQ2}\right)  $. The advantage of defining these variables is
that they satisfy the following relationships:

(i) Cross-products in the objective function (corresponding to the
off-diagonal entries in $Q$):%
\begin{equation}
x_{i}x_{j}=\frac{u_{ij}-v_{ij}}{2}\text{ for }1\leq i<j\leq3.\label{EQ5}%
\end{equation}
(ii) Single variables (corresponding to the vector $b$):%
\begin{equation}
x_{i}=\frac{u_{ij}+v_{ij}+u_{ik}+v_{ik}-u_{jk}-v_{jk}}{2}\text{ for }1\leq
i<j\leq3.\label{EQ3}%
\end{equation}

We have seen in Lemma \ref{LEMA2} that given an $x\in\mathcal{H}_{3}$
there always exists a $w\in\mathcal{C}_{3}$ such that $\phi\left(  x\right)
=w$.

We define the linear transformation $\varphi:$ $\mathcal{C}_{3}\rightarrow
\mathbb{R}^{6}$ given by $\varphi\left(  w\right)  =T_{3}w$, with%
\[
T_{3}=\frac{1}{2}\left(
\begin{array}
[c]{cccccc}%
1 & 1 & -1 & 1 & 1 & -1\\
1 & 0 & 0 & -1 & 0 & 0\\
0 & 1 & 0 & 0 & -1 & 0\\
1 & -1 & 1 & 1 & -1 & 1\\
0 & 0 & 1 & 0 & 0 & -1\\
-1 & 1 & 1 & -1 & 1 & 1
\end{array}
\right).
\]

We claim that $\varphi\circ\beta=id$; this can be checked by inspection:%
\begin{equation}
T_{3}E_{3}=I_{6},\label{EQ103}%
\end{equation}
where $E_{3}$ was given in $\left(  \ref{EQ14}\right)  $.

\begin{remark}
According to the definition of $\alpha$ there exists a vector $c\in
\mathbb{R}^{6}$ such that $f\left(  x\right)  =c^{T}\alpha\left(  x\right)  $.
Additionally, always there exists a vector $\tilde{c}\in\mathbb{R}^{6}$ such
that $f\left(  x\right)  =\tilde{c}^{T}\phi\left(  x\right)  $, where%
\begin{equation}
\tilde{c}^{T}=c^{T}T_{3}.\label{EQ9}%
\end{equation}
This is due to
\[
\tilde{c}^{T}E_{3}\alpha\left(  x\right)  =c^{T}T_{3}E_{3}\alpha\left(
x\right),
\]
and we know from $\left(  \ref{EQ103}\right)  $ that $T_{3}E_{3}=I$.
\end{remark}

Let $\tilde{f}\left(  w\right)  =\tilde{c}^{T}w$, we define the optimization problem:

\begin{description}
\item[\textbf{(P}$_{3}^{\prime}$\textbf{):}] $\min_{w\in\mathcal{C}_{3}}%
\tilde{f}\left(  w\right)  $.
\end{description}

The following theorem states that the minimum of $f$ over $\mathcal{H}_{3}$ is
the minimum of $\tilde{f}$ over $\mathcal{C}_{3}$.

\begin{theorem}
\label{TEO1}Let $\,f\left(  x^{\ast}\right)  $ be the minimum of the problem
(P$_{3}$), and $\tilde{f}\left(  w^{\ast}\right)  $ the minimum of the problem
(P$_{3}^{\prime}$'). Then $f\left(  x^{\ast}\right)  =\tilde{f}\left(
w^{\ast}\right)  $.
\end{theorem}

\begin{proof}
In virtue of Lemma \ref{LEMA2}, $\phi\left(  \mathcal{H}_{3}\right)
\subseteq\mathcal{C}_{3}$. This means that there exists a point $w\in
\mathcal{C}_{3}$ such that $\phi\left(  x^{\ast}\right)  =w$. The minimum of
$\tilde{f}$ over $\mathcal{C}_{3}$ is attained at $w^{\ast}\in\mathcal{C}_{3}%
$, so that
\begin{equation}
\tilde{f}\left(  w^{\ast}\right)  \leq\tilde{f}\left(  w\right)
\text{.}\label{EQ100}%
\end{equation}
The connection between $f$ and $\tilde{f}$ yields:%
\begin{gather}
\tilde{f}\left(  \phi\left(  x^{\ast}\right)  \right)  =\tilde{c}^{T}%
\phi\left(  x^{\ast}\right)  =\left(  c^{T}T_{3}\right)  \left(  E_{3}%
\alpha\left(  x^{\ast}\right)  \right)  \label{EQ101}\\
=c^{T}\alpha\left(  x^{\ast}\right)  =f\left(  x^{\ast}\right)  \text{.}%
\nonumber
\end{gather}
\newline According to $\left(  \ref{EQ100}\right)  $ and $\left(
\ref{EQ101}\right)  $ it follows that%
\[
\tilde{f}\left(  w^{\ast}\right)  \leq\tilde{f}\left(  w\right)  =f\left(
x^{\ast}\right)  \text{.}%
\]
On the other hand, we know that there exists a vertex $y$ in $\mathcal{H}_{3}$
such that $\phi\left(  y\right)  =w^{\ast}$, so%
\begin{equation}
f\left(  y\right)  =c^{T}\alpha\left(  y\right)  =\left(  c^{T}T_{3}\right)
\left(  E_{3}\alpha\left(  y\right)  \right)  =\tilde{c}^{T}\phi\left(
y\right)  =\tilde{f}\left(  \phi\left(  y\right)  \right)  \text{.}%
\label{EQ102}%
\end{equation}
\newline Since the minimum of $f$ over $\mathcal{H}_{3}$ is attained at
$x^{\ast}\in\mathcal{H}_{3}$, and accounting for $\left(  \ref{EQ102}\right)
$, we have that%
\[
f\left(  y\right)  =\tilde{f}\left(  \phi\left(  y\right)  \right)  =f\left(
w^{\ast}\right)  \geq f\left(  x^{\ast}\right)  \text{.}%
\]
Henceforth, $f\left(  x^{\ast}\right)  =\tilde{f}\left(  w^{\ast}\right)  $.
\end{proof}

In the previous theorem the condition $\phi\left(  \mathcal{H}_{3}\right)
\subseteq\mathcal{C}_{3}$ could be eliminated since $\phi\left(  x^{\ast
}\right)  $ is a vertex of $\mathcal{C}_{3}$.

The problem (P$_{3}^{\prime}$): can be written as a linear programming problem:

\begin{description}
\item[\textbf{(LP}$_{3}$\textbf{):}] $\left\{
\begin{array}
[c]{l}%
\min\tilde{c}^{T}w\\
\text{such that }\left\{
\begin{array}
[c]{c}%
B\mathbf{\lambda}-w=0\\
u^{T}\mathbf{\lambda=}1\\
w\geq0\text{, }\mathbf{\lambda}\geq0
\end{array}
\right.
\end{array}
\right.  $
\end{description}

where $\mathbf{\lambda}^{T}=\left(  \lambda_{1},\ldots,\lambda_{8}\right)  $,
$B=\left(  \phi\left(  p_{1}\right)  ,\ldots,\phi\left(  p_{8}\right)
\right)  $, and $u$ is an all-ones vector of appropriate dimension. Throughout
this work the variables $\lambda$ are called secondary variables.

\begin{example}
Let $f\left(  x\right)  =x^{T}Qx+b^{T}x$ with%
\begin{align*}
Q  &  =\left(
\begin{array}
[c]{ccc}%
0 & -10 & -20\\
-10 & 0 & -10\\
-20 & -10 & 0
\end{array}
\right), \\
b  &  =\left(
\begin{array}
[c]{c}%
-2\\
-2\\
-26
\end{array}
\right).
\end{align*}
This objective function is written as $f\left(  x\right)  =c^{T}\alpha\left(
x\right)  $, where $c^{T}=\left(  -2,-20,-40,-2,-20,-26\right)  $. The problem
(LP$_{3}$) has an objective function $\tilde{f}\left(  w\right)  =\tilde
{c}^{T}w$ with $\tilde{c}^{T}=c^{T}T_{3}=\left(  1,-33,-23,21,7,-3\right)  $.

The matrix $B$ in the constraints was given in \ref{EQ13}. %
The minimum of (P') is $-110$ and is attained at $w^{\ast}=\left(
2,2,2,0,0,0\right)  ^{T}$.
\end{example}

\subsection{The General Case}

In this subsection, we generalize the simple case to the n-dimensional
hypercube $\mathcal{H}_{n}$.

For each triplet $\left(  i,j,k\right)  $ with $1\leq i<j<k\leq n$ we define
the convex-hull $\mathcal{H}_{3}^{\left(  i,j,k\right)  }$ in the variables
$x_{i}$, $x_{j}$, and $x_{k}$. Associated with this convex-hull we will have a
map $\phi_{i,j,k}:\mathcal{H}_{3}^{\left(  i,j,k\right)  }\rightarrow\left[
0,2\right]  ^{3}\times\left[  0,\frac{1}{2}\right]  ^{3}$ defined as%
\[
\phi_{i,j,k}\left(  x_{i},x_{j},x_{k}\right)  =\left(
\begin{array}
[c]{c}%
\frac{x_{i}+2x_{i}x_{j}+x_{j}}{2}\\
\frac{x_{i}+2x_{i}x_{k}+x_{k}}{2}\\
\frac{x_{j}+2x_{j}x_{k}+x_{k}}{2}\\
\frac{x_{i}-2x_{i}x_{j}+x_{j}}{2}\\
\frac{x_{i}-2x_{i}x_{k}+x_{k}}{2}\\
\frac{x_{j}-2x_{j}x_{k}+x_{k}}{2}%
\end{array}
\right).
\]

From the set of vertices $V^{\left(  i,j,k\right)  }$ of $\mathcal{H}%
_{3}^{\left(  i,j,k\right)  }$ the convex-Hull $\mathcal{C}_{3}^{\left(
i,j,k\right)  }=conv\left(  \phi_{i,j,k}\left(  V^{\left(  i,j,k\right)
}\right)  \right)  $) is defined. Recall that although the set $\phi
_{i,j,k}\left(  \mathcal{H}_{3}^{\left(  i,j,k\right)  }\right)  $ is not
convex, the image of $\mathcal{H}_{3}^{\left(  i,j,k\right)  }$ through
$\phi_{i,j,k}$ is contained in $\mathcal{C}_{3}^{\left(  i,j,k\right)  }$.

Similarly as done in $\left(  \ref{EQ1}\right)  $ and $\left(
\ref{EQ2}\right)  $ we define primary variables $u_{ij}$ and $v_{ij}$ for
$1\leq i<j\leq n$. For the sake of clarity we adopt the notation $w_{i,j,k}$
to refer to the vector $\left(  u_{ij},u_{ik},u_{jk}\,v_{ij},v_{ik}%
,v_{jk}\right)  $ and $x_{i,j,k}$ for the vector $\left(  x_{i},x_{j}%
,x_{k}\right)  $. With this notation the elements in $\mathcal{C}_{3}^{\left(
i,j,k\right)  }$ are $w_{i,j,k}\in\left[  0,2\right]  ^{3}\times\left[
0,\frac{1}{2}\right]  ^{3}$. Also, we will make a distinction between primary
variables $u_{i,j,k}\in\left[  0,2\right]  ^{3}$ and primary variables
$v_{i,j,k}\in\left[  0,\frac{1}{2}\right]  ^{3}$, which make up the vector
$w_{i,j,k}$ as $w_{i,j,k}=\left(  u_{i,j,k}^{T},v_{i,j,k}^{T}\right)  ^{T}$.
We have seen that given an $x_{i,j,k}\in\mathcal{H}_{3}^{\left(  i,j,k\right)
}$ there always exists a $w\in\mathcal{C}_{3}^{\left(  i,j,k\right)  }$ such
that $\phi_{i,j,k}\left(  x_{i,j,k}\right)  =w_{i,j,k}$. This result should be
borne in mind because it is key in the main theorem of this subsection.

From the convex-hulls $\mathcal{C}_{3}^{\left(  i,j,k\right)  }$ we create the
set $\mathcal{C}_{n}\subseteq\left[  0,2\right]  ^{n}\times\left[  0,\frac
{1}{2}\right]  ^{n}$ by introducing consistency constraints: For it, we
introduce the functions%
\[
g_{i,j,k}\left(  w\right)  =\frac{u_{i,j}+v_{i,j}+u_{i,k}+v_{i,k}%
-u_{j,k}-v_{j,k}}{2},%
\]
where $1\leq i<j<k\leq n$. This function can be compactly written as%
\[
g_{i,j,k}\left(  w\right)  =r^{T}w_{ijk}\text{ with }r^{T}=\frac{1}{2}\left(
1,1,1,1,-1,-1\right).
\]
Given a point $w\in\left[  0,2\right]  ^{n}\times\left[  0,\frac{1}{2}\right]
^{n}$ we define the following consistency constraints for $1\leq j<k\leq n$:

\textbf{(C1):} $g_{1,2,3}\left(  w\right)  =g_{1,j,k}\left(  w\right)  $ with
$\left(  j,k\right)  \neq\left(  2,3\right)  $.

\textbf{(C2):} $g_{2,1,3}\left(  w\right)  =g_{2,j,k}\left(  w\right)  $ with
$\left(  j,k\right)  \neq\left(  1,3\right)  $.

\textbf{(C3):} $g_{i,1,2}\left(  w\right)  =g_{i,j,k}\left(  w\right)  $ with
$\left(  j,k\right)  \neq\left(  1,2\right)  $ and $i\geq3$.

The set $\mathcal{C}_{n}$ is then defined as%
\[
\begin{split}
\mathcal{C}_{n}=\bigg \{  w\in\left[  0,2\right]  ^{n}\times\left[  0,\frac
{1}{2}\right]  ^{n}:w_{i,j,k}\in\mathcal{C}_{3}^{\left(  i,j,k\right)  }
\\
\text{
and }w\text{ satisfies \textbf{(C1)}, \textbf{(C2)}, and \textbf{(C3)}%
} \bigg\}.
\end{split}
\]
\bigskip

In this case $u$ and $v$ are $\binom{n}{2}$-dimensional vectors, and then
$w=\left(  u^{T},v^{T}\right)  ^{T}\in\mathbb{R}^{n\left(  n-1\right)  }$. Let
$w\in\mathcal{C}_{n}$, and triplets $\left(  i,j,k\right)  $, $\left(
i,j,l\right)  $, with $l\neq k$, there exist one-to-one maps $\phi_{i,j,k}$
and $\phi_{i,j,l}$ (as shown in the simple case) such that%
\[
w_{i,j,k}=\phi_{i,j,k}\left(  x_{i,j,k}\right)  \text{, }x_{i,j,k}%
\in\mathcal{H}_{3}^{\left(  i,j,k\right)  },%
\]
and
\[
w_{i,j,l}=\phi_{i,j,l}\left(  y_{i,j,k}\right)  \text{, }x_{i,j,l}%
\in\mathcal{H}_{3}^{\left(  i,j,l\right)  }.%
\]
Consistency means that we expect that
\begin{align*}
x_{i} &  =g_{i,j,k}\left(  w\right)  =g_{i,j,l}\left(  w\right)  =y_{i},\\
x_{j} &  =g_{j,i,k}\left(  w\right)  =g_{j,i.l}\left(  w\right)  =y_{j}.%
\end{align*}

Analogously to the simple case we have that $\phi\left(  \mathcal{H}%
_{n}\right)  \subset\mathcal{C}_{n}$. This result is argued in the following lemma:

\begin{lemma}
\label{LEMA3}Given an arbitrary point $x\in\mathcal{H}_{n}$, there exists a
$w\in\mathcal{C}_{n}$ such that $\phi\left(  x\right)  =w$.
\end{lemma}

\begin{proof}
For a triplet $\left(  i_{1},j_{1},k_{1}\right)  $ there exists a point
$\left(  w_{i_{1}},w_{j_{1}},w_{k_{1}}\right)  \in\mathcal{C}_{3}$ and a map
$\phi_{i_{1},j_{1},k_{1}}:\mathcal{H}_{3}\rightarrow\mathcal{C}_{3}$ such that
$\phi_{i_{1},j_{1},k_{1}}\left(  x_{i_{1}},x_{j_{1}},x_{k_{1}}\right)
=w_{i_{1},j_{1},k_{1}}$ (this was shown above in Lemma \ref{LEMA2} for the
simple case). Let $\left(  i_{2},j_{2},k_{2}\right)  $ be another triplet with
$\phi_{i_{2},j_{2},k_{2}}\left(  x_{i_{2}},x_{j_{2}},x_{k_{2}}\right)
=w_{i_{2},j_{2},k_{2}}$, such that $\left\{  i_{1},j_{1},k_{1}\right\}
\cap\left\{  i_{2},j_{2},k_{2}\right\}  \neq\varnothing$. Without loss of
generality let us assume that $i_{1}=i_{2}$ (otherwise we can make a
permutation of the indices to get that configuration) then from the
consistency constraints we have $x_{i_{1}}=y_{i_{1}}$:%
\[
x_{i_{1}}=g_{i_{1},j_{1},k_{1}}\left(  \phi_{i_{1},j_{1},k_{1}}\left(
x_{i_{1}},x_{j_{1}},x_{k_{1}}\right)  \right)  =g_{i_{1},j_{2},k_{2}}\left(
\phi_{i_{1},j_{2},k_{2}}\left(  y_{i_{1}},y_{j_{2}},y_{k_{2}}\right)  \right)
=y_{i_{1}}%
\]
Extending this idea to all pair of triplets, we conclude that for every
$x\in\mathcal{H}_{n}$, there exists a vector $w\in\mathcal{C}_{n}$ such that
$\phi\left(  x\right)  =w$.

\end{proof}

\begin{example}
\label{EXAMPLE2}For $n=4$, the consistency constraints are%
\begin{gather*}
\text{Consistency for }x_{1}\text{:}\\
\left\{
\begin{array}
[c]{c}%
\left(  u_{12}+v_{12}+u_{13}+v_{13}-u_{23}-v_{23}\right)  \\
-\left(  u_{12}+v_{12}+u_{14}+v_{14}-u_{24}-v_{24}\right)  =0\\
\left(  u_{12}+v_{12}+u_{13}+v_{13}-u_{23}-v_{23}\right)  \\
-\left(  u_{13}+v_{13}+u_{14}+v_{14}-u_{34}-v_{34}\right)  =0
\end{array}
\right.
\end{gather*}%
\begin{gather*}
\text{Consistency for }x_{2}\text{:}\\
\left\{
\begin{array}
[c]{c}%
\left(  u_{12}+v_{12}+u_{23}+v_{23}-u_{13}-v_{13}\right)  \\
-\left(  u_{12}+v_{12}+u_{24}+v_{24}-u_{14}-v_{14}\right)  =0\\
\left(  u_{12}+v_{12}+u_{23}+v_{23}-u_{13}-v_{13}\right)  \\
-\left(  u_{23}+v_{23}+u_{24}+v_{24}-u_{34}-v_{34}\right)  =0
\end{array}
\right.
\end{gather*}%
\begin{gather*}
\text{Consistency for }x_{3}\text{:}\\
\left\{
\begin{array}
[c]{c}%
\left(  u_{13}+v_{13}+u_{23}+v_{23}-u_{12}-v_{12}\right)  \\
-\left(  u_{13}+v_{13}+u_{34}+v_{34}-u_{14}-v_{14}\right)  =0\\
\left(  u_{13}+v_{13}+u_{23}+v_{23}-u_{12}-v_{12}\right)  \\
-\left(  u_{23}+v_{23}+u_{34}+v_{34}-u_{24}-v_{24}\right)  =0
\end{array}
\right.
\end{gather*}%
\begin{gather*}
\text{Consistency for }x_{4}\text{:}\\
\left\{
\begin{array}
[c]{c}%
\left(  u_{14}+v_{14}+u_{24}+v_{24}-u_{12}-v_{12}\right)  \\
-\left(  u_{14}+v_{14}+u_{34}+v_{34}-u_{13}-v_{13}\right)  =0\\
\left(  u_{14}+v_{14}+u_{24}+v_{24}-u_{12}-v_{12}\right)  \\
-\left(  u_{24}+v_{24}+u_{34}+v_{34}-u_{23}-v_{23}\right)  =0
\end{array}
\right.
\end{gather*}

\begin{figure*}
[ptb]
\begin{center}
\tikzset{every picture/.style={line width=0.75pt}} 

\begin{tikzpicture}[x=0.75pt,y=0.75pt,yscale=-1,xscale=1]

\draw   (45.5,24) .. controls (45.5,19.58) and (49.08,16) .. (53.5,16) -- (107.5,16) .. controls (111.92,16) and (115.5,19.58) .. (115.5,24) -- (115.5,48) .. controls (115.5,52.42) and (111.92,56) .. (107.5,56) -- (53.5,56) .. controls (49.08,56) and (45.5,52.42) .. (45.5,48) -- cycle ;
\draw   (45.5,81.33) .. controls (45.5,76.91) and (49.08,73.33) .. (53.5,73.33) -- (107.5,73.33) .. controls (111.92,73.33) and (115.5,76.91) .. (115.5,81.33) -- (115.5,105.33) .. controls (115.5,109.75) and (111.92,113.33) .. (107.5,113.33) -- (53.5,113.33) .. controls (49.08,113.33) and (45.5,109.75) .. (45.5,105.33) -- cycle ;
\draw   (45.5,138.66) .. controls (45.5,134.24) and (49.08,130.66) .. (53.5,130.66) -- (107.5,130.66) .. controls (111.92,130.66) and (115.5,134.24) .. (115.5,138.66) -- (115.5,162.66) .. controls (115.5,167.08) and (111.92,170.66) .. (107.5,170.66) -- (53.5,170.66) .. controls (49.08,170.66) and (45.5,167.08) .. (45.5,162.66) -- cycle ;
\draw   (45.5,193.5) .. controls (45.5,189.08) and (49.08,185.5) .. (53.5,185.5) -- (107.5,185.5) .. controls (111.92,185.5) and (115.5,189.08) .. (115.5,193.5) -- (115.5,217.5) .. controls (115.5,221.92) and (111.92,225.5) .. (107.5,225.5) -- (53.5,225.5) .. controls (49.08,225.5) and (45.5,221.92) .. (45.5,217.5) -- cycle ;
\draw    (117,36) -- (241,36) ;
\draw [shift={(243,36)}, rotate = 180] [color={rgb, 255:red, 0; green, 0; blue, 0 }  ][line width=0.75]    (10.93,-3.29) .. controls (6.95,-1.4) and (3.31,-0.3) .. (0,0) .. controls (3.31,0.3) and (6.95,1.4) .. (10.93,3.29)   ;
\draw    (117,93.33) -- (241,93.33) ;
\draw [shift={(243,93.33)}, rotate = 180] [color={rgb, 255:red, 0; green, 0; blue, 0 }  ][line width=0.75]    (10.93,-3.29) .. controls (6.95,-1.4) and (3.31,-0.3) .. (0,0) .. controls (3.31,0.3) and (6.95,1.4) .. (10.93,3.29)   ;
\draw    (117,150.66) -- (241,150.66) ;
\draw [shift={(243,150.66)}, rotate = 180] [color={rgb, 255:red, 0; green, 0; blue, 0 }  ][line width=0.75]    (10.93,-3.29) .. controls (6.95,-1.4) and (3.31,-0.3) .. (0,0) .. controls (3.31,0.3) and (6.95,1.4) .. (10.93,3.29)   ;
\draw    (118,211) -- (242,211) ;
\draw [shift={(244,211)}, rotate = 180] [color={rgb, 255:red, 0; green, 0; blue, 0 }  ][line width=0.75]    (10.93,-3.29) .. controls (6.95,-1.4) and (3.31,-0.3) .. (0,0) .. controls (3.31,0.3) and (6.95,1.4) .. (10.93,3.29)   ;
\draw  [dash pattern={on 4.5pt off 4.5pt}]  (332,36) -- (391.83,36) -- (408.83,36) ;
\draw [shift={(410.83,36)}, rotate = 180] [color={rgb, 255:red, 0; green, 0; blue, 0 }  ][line width=0.75]    (10.93,-3.29) .. controls (6.95,-1.4) and (3.31,-0.3) .. (0,0) .. controls (3.31,0.3) and (6.95,1.4) .. (10.93,3.29)   ;
\draw  [dash pattern={on 4.5pt off 4.5pt}]  (333,93.33) -- (392.83,93.33) -- (409.83,93.33) ;
\draw [shift={(411.83,93.33)}, rotate = 180] [color={rgb, 255:red, 0; green, 0; blue, 0 }  ][line width=0.75]    (10.93,-3.29) .. controls (6.95,-1.4) and (3.31,-0.3) .. (0,0) .. controls (3.31,0.3) and (6.95,1.4) .. (10.93,3.29)   ;
\draw  [dash pattern={on 4.5pt off 4.5pt}]  (334,150.66) -- (393.83,150.66) -- (410.83,150.66) ;
\draw [shift={(412.83,150.66)}, rotate = 180] [color={rgb, 255:red, 0; green, 0; blue, 0 }  ][line width=0.75]    (10.93,-3.29) .. controls (6.95,-1.4) and (3.31,-0.3) .. (0,0) .. controls (3.31,0.3) and (6.95,1.4) .. (10.93,3.29)   ;
\draw  [dash pattern={on 4.5pt off 4.5pt}]  (332,205.5) -- (391.83,205.5) -- (408.83,205.5) ;
\draw [shift={(410.83,205.5)}, rotate = 180] [color={rgb, 255:red, 0; green, 0; blue, 0 }  ][line width=0.75]    (10.93,-3.29) .. controls (6.95,-1.4) and (3.31,-0.3) .. (0,0) .. controls (3.31,0.3) and (6.95,1.4) .. (10.93,3.29)   ;
\draw  [dash pattern={on 4.5pt off 4.5pt}]  (332,36) -- (410.21,92.16) ;
\draw [shift={(411.83,93.33)}, rotate = 215.68] [color={rgb, 255:red, 0; green, 0; blue, 0 }  ][line width=0.75]    (10.93,-3.29) .. controls (6.95,-1.4) and (3.31,-0.3) .. (0,0) .. controls (3.31,0.3) and (6.95,1.4) .. (10.93,3.29)   ;
\draw  [dash pattern={on 4.5pt off 4.5pt}]  (332,36) -- (411.68,149.03) ;
\draw [shift={(412.83,150.66)}, rotate = 234.82] [color={rgb, 255:red, 0; green, 0; blue, 0 }  ][line width=0.75]    (10.93,-3.29) .. controls (6.95,-1.4) and (3.31,-0.3) .. (0,0) .. controls (3.31,0.3) and (6.95,1.4) .. (10.93,3.29)   ;
\draw  [dash pattern={on 4.5pt off 4.5pt}]  (333,93.33) -- (409.22,37.19) ;
\draw [shift={(410.83,36)}, rotate = 503.63] [color={rgb, 255:red, 0; green, 0; blue, 0 }  ][line width=0.75]    (10.93,-3.29) .. controls (6.95,-1.4) and (3.31,-0.3) .. (0,0) .. controls (3.31,0.3) and (6.95,1.4) .. (10.93,3.29)   ;
\draw  [dash pattern={on 4.5pt off 4.5pt}]  (333,93.33) -- (409.69,203.86) ;
\draw [shift={(410.83,205.5)}, rotate = 235.24] [color={rgb, 255:red, 0; green, 0; blue, 0 }  ][line width=0.75]    (10.93,-3.29) .. controls (6.95,-1.4) and (3.31,-0.3) .. (0,0) .. controls (3.31,0.3) and (6.95,1.4) .. (10.93,3.29)   ;
\draw  [dash pattern={on 4.5pt off 4.5pt}]  (334,150.66) -- (409.72,37.66) ;
\draw [shift={(410.83,36)}, rotate = 483.83] [color={rgb, 255:red, 0; green, 0; blue, 0 }  ][line width=0.75]    (10.93,-3.29) .. controls (6.95,-1.4) and (3.31,-0.3) .. (0,0) .. controls (3.31,0.3) and (6.95,1.4) .. (10.93,3.29)   ;
\draw  [dash pattern={on 4.5pt off 4.5pt}]  (334,150.66) -- (409.21,204.34) ;
\draw [shift={(410.83,205.5)}, rotate = 215.52] [color={rgb, 255:red, 0; green, 0; blue, 0 }  ][line width=0.75]    (10.93,-3.29) .. controls (6.95,-1.4) and (3.31,-0.3) .. (0,0) .. controls (3.31,0.3) and (6.95,1.4) .. (10.93,3.29)   ;
\draw  [dash pattern={on 4.5pt off 4.5pt}]  (332,205.5) -- (410.67,94.96) ;
\draw [shift={(411.83,93.33)}, rotate = 485.44] [color={rgb, 255:red, 0; green, 0; blue, 0 }  ][line width=0.75]    (10.93,-3.29) .. controls (6.95,-1.4) and (3.31,-0.3) .. (0,0) .. controls (3.31,0.3) and (6.95,1.4) .. (10.93,3.29)   ;
\draw  [dash pattern={on 4.5pt off 4.5pt}]  (332,205.5) -- (411.18,151.78) ;
\draw [shift={(412.83,150.66)}, rotate = 505.85] [color={rgb, 255:red, 0; green, 0; blue, 0 }  ][line width=0.75]    (10.93,-3.29) .. controls (6.95,-1.4) and (3.31,-0.3) .. (0,0) .. controls (3.31,0.3) and (6.95,1.4) .. (10.93,3.29)   ;

\draw (80.5,36) node    {$\mathcal{C}^{( 1,2,3)}_{3}$};
\draw (80.5,93.33) node    {$\mathcal{C}^{( 1,2,4)}_{3}$};
\draw (80.5,150.66) node    {$\mathcal{C}^{( 1,3,4)}_{3}$};
\draw (80.5,205.5) node    {$\mathcal{C}^{( 2,3,4)}_{3}$};
\draw (251,25) node [anchor=north west][inner sep=0.75pt]   [align=left] {$\displaystyle ( x_{1} ,x_{2} ,x_{3}$)};
\draw (143.5,9) node [anchor=north west][inner sep=0.75pt]   [align=left] {$\displaystyle Cw_{123}$};
\draw (250,82.33) node [anchor=north west][inner sep=0.75pt]   [align=left] {$\displaystyle ( x_{1} ,x_{2} ,x_{4}$)};
\draw (252,139.66) node [anchor=north west][inner sep=0.75pt]   [align=left] {$\displaystyle ( x_{1} ,x_{3} ,x_{4}$)};
\draw (251,194.5) node [anchor=north west][inner sep=0.75pt]   [align=left] {$\displaystyle ( x_{2} ,x_{3} ,x_{4}$)};
\draw (413,25) node [anchor=north west][inner sep=0.75pt]   [align=left] {$\displaystyle x_{1} =r^{T} w_{123} =r^{T} w_{124} =r^{T} w_{134} \ $};
\draw (143.5,67) node [anchor=north west][inner sep=0.75pt]   [align=left] {$\displaystyle Cw_{124}$};
\draw (143.5,124) node [anchor=north west][inner sep=0.75pt]   [align=left] {$\displaystyle Cw_{134}$};
\draw (143.5,188.5) node [anchor=north west][inner sep=0.75pt]   [align=left] {$\displaystyle Cw_{234}$};
\draw (433,3) node [anchor=north west][inner sep=0.75pt]   [align=left] {Consistency Constraints};
\draw (414,82.33) node [anchor=north west][inner sep=0.75pt]   [align=left] {$\displaystyle x_{2} =r^{T} w_{213} =r^{T} w_{214} =r^{T} w_{234} \ $};
\draw (414,139.66) node [anchor=north west][inner sep=0.75pt]   [align=left] {$\displaystyle x_{3} =r^{T} w_{312} =r^{T} w_{314} =r^{T} w_{324} \ $};
\draw (414,194.5) node [anchor=north west][inner sep=0.75pt]   [align=left] {$\displaystyle x_{4} =r^{T} w_{412} =r^{T} w_{413} =r^{T} w_{423} \ $};

\end{tikzpicture}
\caption{Consistency constraints for $n=4$}%
\label{FIG3}%
\end{center}
\end{figure*}
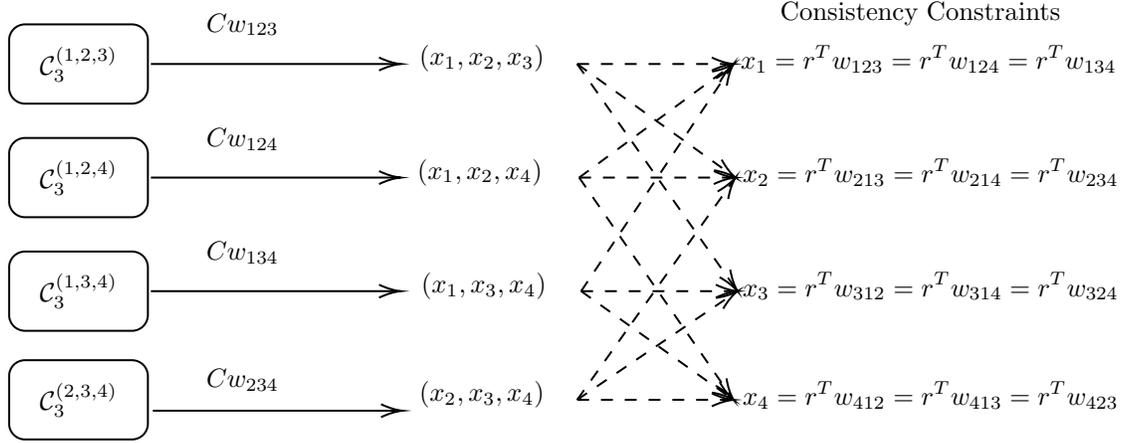

In Figure \ref{FIG3}, we graphically show the consistency constraints for the
variables $x_{1}$, $x_{2}$, $x_{3}$ and $x_{4}$, generated from the
convex-hulls $\mathcal{C}_{3}^{\left(  1,2,3\right)  }\,$, $\mathcal{C}%
_{3}^{\left(  1,2,4\right)  }$, $\mathcal{C}_{3}^{\left(  1,3,4\right)  }$ and
$\mathcal{C}_{3}^{\left(  2,3,4\right)  }$ via the linear transformation
$C=\left(  M,M\right)  $ with%
\[
M=\frac{1}{4}\left(
\begin{array}
[c]{ccc}%
1 & 1 & -1\\
1 & -1 & 1\\
-1 & 1 & 1
\end{array}
\right).
\]
Note that
\[
C\phi_{i,j,k}\left(  x_{i,j,k}\right)  =\left(
\begin{array}
[c]{c}%
x_{i}\\
x_{j}\\
x_{k}%
\end{array}
\right).
\]

\end{example}

We will see that the minimization problem of $f\left(  x\right)  $ in
$\mathcal{H}_{n}$ is equivalent to the minimization problem of a linear
objective function $\tilde{f}\left(  w\right)  $ in $\mathcal{C}_{n}$, where
$\mathcal{C}_{n}$ can be expressed as constraints on equality given by the
convexity constraints and consistency constraints, in addition to the natural
constraints for the vectors $\lambda^{\left(  i,j,k\right)  }$ (secondary
variables) in the definition of the convex-hull $\mathcal{C}_{3}^{\left(
i,j,k\right)  }$:%
\[
\begin{split}
\mathcal{C}_{3}^{\left(  i,j,k\right)  }=\bigg \{  \sum_{l=1}^{8}\lambda
_{l}^{\left(  i,j,k\right)  }\phi\left(  p_{l}\right)  :\sum_{l=1}^{8}%
\lambda_{l}^{\left(  i,j,k\right)  }=1,
\\
\text{ }p_{l}\in V^{\left(
i,j,k\right)  }\text{, }\lambda_{l}^{\left(  i,j,k\right)  }\geq0\text{ for
}l=1,\ldots,8 \bigg \}.
\end{split}
\]
In particular, to define $\tilde{f}\left(  w\right)  $ it is necessary to
introduce a transformation $T_{n}$ such that%
\[
\tilde{f}\left(  w\right)  =\tilde{c}^{T}w=c^{T}T_{n}w.
\]
This transformation is obtained according to the following relationships:%
\begin{align*}
x_{i} &  =g_{i,j,k}\left(  w\right),  \\
x_{i}x_{j} &  =\frac{u_{i,j}-v_{i,j}}{2}.%
\end{align*}
More specifically, if $\alpha:\mathcal{H}_{n}\rightarrow\alpha\left(
\mathcal{H}_{n}\right)  $ is the map%
\begin{align*}
\alpha\left(  x\right)    & =(x_{1},x_{1}x_{2},\cdots,x_{1}x_{n},x_{2}%
,x_{2}x_{3},\cdots,\\
& x_{2}x_{n},\cdots,x_{n-1},x_{n-1}x_{n},x_{n}).
\end{align*}
and $\beta:\alpha\left(  \mathcal{H}_{n}\right)  \rightarrow\mathcal{C}_{n}$
is a linear map $\beta\left(  \alpha\left(  x\right)  \right)  =E_{n}%
\alpha\left(  x\right)  =w$ such that for a triplet $\left(  i,j,k\right)  $
we verify that%
\begin{align*}
u_{ij}  & =\frac{x_{i}+2x_{i}x_{j}+x_{j}}{2}\text{, }u_{ik}=\frac{x_{i}%
+2x_{i}x_{k}+x_{k}}{2},\\
u_{jk}  & =\frac{x_{j}+2x_{j}x_{k}+x_{k}}{2}\text{, }v_{ij}=\frac{x_{i}%
-2x_{i}x_{j}+x_{j}}{2},\\
v_{ik}  & =\frac{x_{i}-2x_{i}x_{k}+x_{k}}{2},v_{jk}=\frac{x_{j}-2x_{j}%
x_{k}+x_{k}}{2}.%
\end{align*}
Here $E_{n}$ is not a square matrix but a rectangular $n\left(  n-1\right)
\times\frac{n\left(  n+1\right)  }{2}$. The matrices $E_{n}$ and $T_{n}$ are
connected by the relation:%
\[
T_{n}E_{n}=I,
\]
where $I$ is the identity matrix of dimension $\frac{n\left(  n+1\right)  }%
{2}\times\frac{n\left(  n+1\right)  }{2}$. Obviolusly the transformation
$T_{n}$ satisfies the relation%
\[
T_{n}\beta\left(  \alpha\left(  x\right)  \right)  =\alpha\left(  x\right).
\]

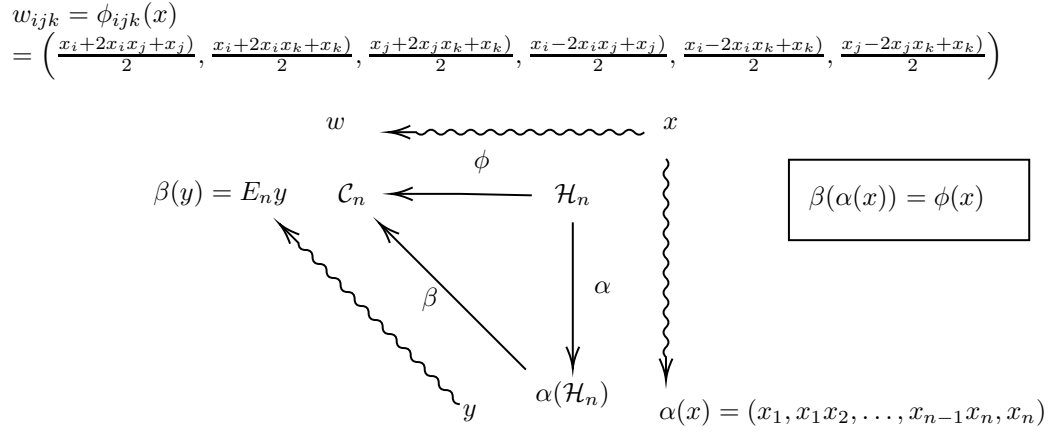
\begin{figure*}
[ptb]
\begin{center}
\tikzset{every picture/.style={line width=0.75pt}} 

\begin{tikzpicture}[x=0.75pt,y=0.75pt,yscale=-1,xscale=1]

\draw    (309,122) -- (309,193.73) ;
\draw [shift={(309,195.73)}, rotate = 270] [color={rgb, 255:red, 0; green, 0; blue, 0 }  ][line width=0.75]    (10.93,-3.29) .. controls (6.95,-1.4) and (3.31,-0.3) .. (0,0) .. controls (3.31,0.3) and (6.95,1.4) .. (10.93,3.29)   ;
\draw    (285.33,196.33) -- (212.75,123.75) ;
\draw [shift={(211.33,122.33)}, rotate = 405] [color={rgb, 255:red, 0; green, 0; blue, 0 }  ][line width=0.75]    (10.93,-3.29) .. controls (6.95,-1.4) and (3.31,-0.3) .. (0,0) .. controls (3.31,0.3) and (6.95,1.4) .. (10.93,3.29)   ;
\draw    (288.33,108.37) -- (252.33,107.63) -- (219.33,107.63) ;
\draw [shift={(217.33,107.63)}, rotate = 360] [color={rgb, 255:red, 0; green, 0; blue, 0 }  ][line width=0.75]    (10.93,-3.29) .. controls (6.95,-1.4) and (3.31,-0.3) .. (0,0) .. controls (3.31,0.3) and (6.95,1.4) .. (10.93,3.29)   ;
\draw    (356,90) .. controls (357.67,91.67) and (357.67,93.33) .. (356,95) .. controls (354.33,96.67) and (354.33,98.33) .. (356,100) .. controls (357.67,101.67) and (357.67,103.33) .. (356,105) .. controls (354.33,106.67) and (354.33,108.33) .. (356,110) .. controls (357.67,111.67) and (357.67,113.33) .. (356,115) .. controls (354.33,116.67) and (354.33,118.33) .. (356,120) .. controls (357.67,121.67) and (357.67,123.33) .. (356,125) .. controls (354.33,126.67) and (354.33,128.33) .. (356,130) .. controls (357.67,131.67) and (357.67,133.33) .. (356,135) .. controls (354.33,136.67) and (354.33,138.33) .. (356,140) .. controls (357.67,141.67) and (357.67,143.33) .. (356,145) .. controls (354.33,146.67) and (354.33,148.33) .. (356,150) .. controls (357.67,151.67) and (357.67,153.33) .. (356,155) .. controls (354.33,156.67) and (354.33,158.33) .. (356,160) .. controls (357.67,161.67) and (357.67,163.33) .. (356,165) .. controls (354.33,166.67) and (354.33,168.33) .. (356,170) .. controls (357.67,171.67) and (357.67,173.33) .. (356,175) .. controls (354.33,176.67) and (354.33,178.33) .. (356,180) .. controls (357.67,181.67) and (357.67,183.33) .. (356,185) .. controls (354.33,186.67) and (354.33,188.33) .. (356,190) -- (356,191.73) -- (356,199.73) ;
\draw [shift={(356,201.73)}, rotate = 270] [color={rgb, 255:red, 0; green, 0; blue, 0 }  ][line width=0.75]    (10.93,-3.29) .. controls (6.95,-1.4) and (3.31,-0.3) .. (0,0) .. controls (3.31,0.3) and (6.95,1.4) .. (10.93,3.29)   ;
\draw    (252,213.73) .. controls (249.64,213.74) and (248.46,212.56) .. (248.46,210.2) .. controls (248.47,207.84) and (247.29,206.66) .. (244.93,206.66) .. controls (242.57,206.67) and (241.39,205.49) .. (241.39,203.13) .. controls (241.4,200.77) and (240.22,199.59) .. (237.86,199.59) .. controls (235.5,199.6) and (234.32,198.42) .. (234.32,196.06) .. controls (234.33,193.7) and (233.15,192.52) .. (230.79,192.52) .. controls (228.43,192.52) and (227.25,191.34) .. (227.25,188.98) .. controls (227.25,186.63) and (226.07,185.45) .. (223.72,185.45) .. controls (221.36,185.45) and (220.18,184.27) .. (220.18,181.91) .. controls (220.18,179.55) and (219,178.37) .. (216.64,178.38) .. controls (214.28,178.38) and (213.1,177.2) .. (213.11,174.84) .. controls (213.11,172.48) and (211.93,171.3) .. (209.57,171.31) .. controls (207.21,171.31) and (206.03,170.13) .. (206.04,167.77) .. controls (206.04,165.41) and (204.86,164.23) .. (202.5,164.24) .. controls (200.14,164.24) and (198.96,163.06) .. (198.97,160.7) .. controls (198.97,158.34) and (197.79,157.16) .. (195.43,157.16) .. controls (193.08,157.16) and (191.9,155.98) .. (191.9,153.63) .. controls (191.9,151.27) and (190.72,150.09) .. (188.36,150.09) .. controls (186,150.1) and (184.82,148.92) .. (184.82,146.56) .. controls (184.83,144.2) and (183.65,143.02) .. (181.29,143.02) .. controls (178.93,143.03) and (177.75,141.85) .. (177.75,139.49) .. controls (177.76,137.13) and (176.58,135.95) .. (174.22,135.95) .. controls (171.86,135.96) and (170.68,134.78) .. (170.68,132.42) -- (168.07,129.8) -- (162.41,124.15) ;
\draw [shift={(161,122.73)}, rotate = 405] [color={rgb, 255:red, 0; green, 0; blue, 0 }  ][line width=0.75]    (10.93,-3.29) .. controls (6.95,-1.4) and (3.31,-0.3) .. (0,0) .. controls (3.31,0.3) and (6.95,1.4) .. (10.93,3.29)   ;
\draw    (345,76.5) .. controls (343.33,78.17) and (341.67,78.17) .. (340,76.5) .. controls (338.33,74.83) and (336.67,74.83) .. (335,76.5) .. controls (333.33,78.17) and (331.67,78.17) .. (330,76.5) .. controls (328.33,74.83) and (326.67,74.83) .. (325,76.5) .. controls (323.33,78.17) and (321.67,78.17) .. (320,76.5) .. controls (318.33,74.83) and (316.67,74.83) .. (315,76.5) .. controls (313.33,78.17) and (311.67,78.17) .. (310,76.5) .. controls (308.33,74.83) and (306.67,74.83) .. (305,76.5) .. controls (303.33,78.17) and (301.67,78.17) .. (300,76.5) .. controls (298.33,74.83) and (296.67,74.83) .. (295,76.5) .. controls (293.33,78.17) and (291.67,78.17) .. (290,76.5) .. controls (288.33,74.83) and (286.67,74.83) .. (285,76.5) .. controls (283.33,78.17) and (281.67,78.17) .. (280,76.5) .. controls (278.33,74.83) and (276.67,74.83) .. (275,76.5) .. controls (273.33,78.17) and (271.67,78.17) .. (270,76.5) .. controls (268.33,74.83) and (266.67,74.83) .. (265,76.5) .. controls (263.33,78.17) and (261.67,78.17) .. (260,76.5) .. controls (258.33,74.83) and (256.67,74.83) .. (255,76.5) .. controls (253.33,78.17) and (251.67,78.17) .. (250,76.5) .. controls (248.33,74.83) and (246.67,74.83) .. (245,76.5) .. controls (243.33,78.17) and (241.67,78.17) .. (240,76.5) .. controls (238.33,74.83) and (236.67,74.83) .. (235,76.5) .. controls (233.33,78.17) and (231.67,78.17) .. (230,76.5) -- (226.1,76.5) -- (218.1,76.5) ;
\draw [shift={(216.1,76.5)}, rotate = 360] [color={rgb, 255:red, 0; green, 0; blue, 0 }  ][line width=0.75]    (10.93,-3.29) .. controls (6.95,-1.4) and (3.31,-0.3) .. (0,0) .. controls (3.31,0.3) and (6.95,1.4) .. (10.93,3.29)   ;
\draw   (418,90.4) -- (539.2,90.4) -- (539.2,131) -- (418,131) -- cycle ;

\draw (309,108) node    {$\mathcal{H}_{n}$};
\draw (309,209) node    {$\alpha (\mathcal{H}_{n})$};
\draw (198,108) node    {$\mathcal{C}_{n}$};
\draw (276,30) node    {$ \begin{array}{l}
w_{ijk} =\phi _{ijk}( x)\\
=\left(\frac{x_{i} +2x_{i} x_{j} +x_{j})}{2} ,\frac{x_{i} +2x_{i} x_{k} +x_{k})}{2} ,\frac{x_{j} +2x_{j} x_{k} +x_{k})}{2} ,\frac{x_{i} -2x_{i} x_{j} +x_{j})}{2} ,\frac{x_{i} -2x_{i} x_{k} +x_{k})}{2} ,\frac{x_{j} -2x_{j} x_{k} +x_{k})}{2}\right)
\end{array}$};
\draw (183,67.5) node [anchor=north west][inner sep=0.75pt]   [align=left] {$\displaystyle w$};
\draw (353,67.5) node [anchor=north west][inner sep=0.75pt]   [align=left] {$\displaystyle x$};
\draw (351,210) node [anchor=north west][inner sep=0.75pt]   [align=left] {$\displaystyle \alpha ( x) =( x_{1} ,x_{1} x_{2} ,\dotsc ,x_{n-1} x_{n} ,x_{n}$)};
\draw (237,161) node    {$\beta $};
\draw (252,212.73) node [anchor=north west][inner sep=0.75pt]   [align=left] {$\displaystyle y$};
\draw (96,99) node [anchor=north west][inner sep=0.75pt]   [align=left] {$\displaystyle \beta ( y) =E_{n} y\ $};
\draw (324,155) node    {$\alpha $};
\draw (263,91) node    {$\phi $};
\draw (426.6,101.7) node [anchor=north west][inner sep=0.75pt]   [align=left] {$\displaystyle \beta ( \alpha ( x)) =\phi ( x)$};

\end{tikzpicture}
\caption{Diagram for the maps $\phi$,
$\alpha$, and $\beta$ in the general case.}%
\label{FIG4}%
\end{center}
\end{figure*}

Figure \ref{FIG4} illustrates the maps $\phi$, $\alpha$, and $\beta$, in a
similar way as in the simple case ($n=3$).

\begin{example}
\label{EXAMPLE3}For $n=4$:%
\[
T_{4}=\frac{1}{2}\left(
\begin{array}
[c]{cccccccccccc}%
1 & 1 & 0 & -1 & 0 & 0 & 1 & 1 & 0 & -1 & 0 & 0\\
1 & 0 & 0 & 0 & 0 & 0 & -1 & 0 & 0 & 0 & 0 & 0\\
0 & 1 & 0 & 0 & 0 & 0 & 0 & -1 & 0 & 0 & 0 & 0\\
0 & 0 & 1 & 0 & 0 & 0 & 0 & 0 & -1 & 0 & 0 & 0\\
1 & -1 & 0 & 1 & 0 & 0 & 1 & -1 & 0 & 1 & 0 & 0\\
0 & 0 & 0 & 1 & 0 & 0 & 0 & 0 & 0 & -1 & 0 & 0\\
0 & 0 & 0 & 0 & 1 & 0 & 0 & 0 & 0 & 0 & -1 & 0\\
-1 & 1 & 0 & 1 & 0 & 0 & -1 & 1 & 0 & 1 & 0 & 0\\
0 & 0 & 0 & 0 & 0 & 1 & 0 & 0 & 0 & 0 & 0 & -1\\
-1 & 0 & 1 & 0 & 1 & 0 & -1 & 0 & 1 & 0 & 1 & 0
\end{array}
\right).
\]
The matrix $E_{4}$ is%
\[
\frac{1}{2}\left(
\begin{array}
[c]{cccccccccc}%
1 & 2 & 0 & 0 & 1 & 0 & 0 & 0 & 0 & 0\\
1 & 0 & 2 & 0 & 0 & 0 & 0 & 1 & 0 & 0\\
1 & 0 & 0 & 2 & 0 & 0 & 0 & 0 & 0 & 1\\
0 & 0 & 0 & 0 & 1 & 2 & 0 & 1 & 0 & 0\\
0 & 0 & 0 & 0 & 1 & 0 & 2 & 0 & 0 & 1\\
0 & 0 & 0 & 0 & 0 & 0 & 0 & 1 & 2 & 1\\
1 & -2 & 0 & 0 & 1 & 0 & 0 & 0 & 0 & 0\\
1 & 0 & -2 & 0 & 0 & 0 & 0 & 1 & 0 & 0\\
1 & 0 & 0 & -2 & 0 & 0 & 0 & 0 & 0 & 1\\
0 & 0 & 0 & 0 & 1 & -2 & 0 & 1 & 0 & 0\\
0 & 0 & 0 & 0 & 1 & 0 & -2 & 0 & 0 & 1\\
0 & 0 & 0 & 0 & 0 & 0 & 0 & 1 & -2 & 1
\end{array}
\right).
\]
It can be checked that $T_{4}E_{4}=I_{10}$.
\end{example}

Now we formulate the minimization problem in $\mathcal{C}_{n}$:

\begin{description}
\item[\textbf{(P}$_{n}^{\prime}$\textbf{):}] $\min_{w\in\mathcal{C}_{n}}%
\tilde{f}\left(  w\right)  $,
\end{description}

where the objective function $\tilde{f}:\mathcal{C}_{n}\rightarrow\mathbb{R}$
is defined as $\tilde{f}\left(  x\right)  =\tilde{c}^{T}w$ with $\tilde{c}%
^{T}=c^{T}T_{n}$. The following Theorem is an extension of Theorem \ref{TEO1}
and proves that the minimum of the problem ($P_{n}$) is the same as that of
($P_{n}^{\prime}$).

\begin{theorem}
Let $w^{\ast}\in\mathcal{C}_{n}$ be the point that minimizes the function
$\tilde{f}$, and $x^{\ast}\in\mathcal{H}_{n}$ the point where $f$ reaches the
minimum (which we know to be a vertex of $\mathcal{H}_{n}$, then $\tilde
{f}\left(  w^{\ast}\right)  =f\left(  x^{\ast}\right)  $).
\end{theorem}

\begin{proof}
From Lemma \ref{LEMA3}, we know that $\phi\left(  \mathcal{H}_{n}\right)
\subseteq\mathcal{C}_{n}$, so there is a point $w\in\mathcal{C}_{n}$ such that
$\phi\left(  x^{\ast}\right)  =w$. The minimum of $\tilde{f}$ over
$\mathcal{C}_{n}$ is attained at $w^{\ast}\in\mathcal{C}_{n}$, so that
\begin{equation}
\tilde{f}\left(  w^{\ast}\right)  \leq\tilde{f}\left(  w\right)
\text{.}\label{EQ104}%
\end{equation}
The connection between $f$ and $\tilde{f}$ yields:%
\begin{gather}
\tilde{f}\left(  \phi\left(  x^{\ast}\right)  \right)  =\tilde{c}^{T}%
\phi\left(  x^{\ast}\right)  =\left(  c^{T}T_{n}\right)  \left(  E_{n}%
\alpha\left(  x^{\ast}\right)  \right)  \label{EQ105}\\
=c^{T}\alpha\left(  x^{\ast}\right)  =f\left(  x^{\ast}\right)  \text{.}%
\nonumber
\end{gather}
\newline According to $\left(  \ref{EQ104}\right)  $ and $\left(
\ref{EQ105}\right)  $ it follows that%
\[
\tilde{f}\left(  w^{\ast}\right)  \leq\tilde{f}\left(  w\right)  =f\left(
x^{\ast}\right)  \text{.}%
\]
On the other hand, we know that there exists a vertex $y$ in $\mathcal{H}_{n}$
such that $\phi\left(  y\right)  =w^{\ast}$, so%
\begin{equation}
f\left(  y\right)  =c^{T}\alpha\left(  y\right)  =\left(  c^{T}T_{n}\right)
\left(  E_{n}\alpha\left(  y\right)  \right)  =\tilde{c}^{T}\phi\left(
y\right)  =\tilde{f}\left(  \phi\left(  y\right)  \right)  \text{.}%
\label{EQ106}%
\end{equation}
\newline Since the minimum of $f$ over $\mathcal{H}_{n}$ is attained at
$x^{\ast}\in\mathcal{H}_{n}$, and accounting for $\left(  \ref{EQ106}\right)
$, we have that%
\[
f\left(  y\right)  =\tilde{f}\left(  \phi\left(  y\right)  \right)  =f\left(
w^{\ast}\right)  \geq f\left(  x^{\ast}\right)  \text{.}%
\]
Henceforth, $f\left(  x^{\ast}\right)  =\tilde{f}\left(  w^{\ast}\right)  $.
\end{proof}

As for the simple case $n=3$, in the general case, the condition $\phi\left(
\mathcal{H}_{n}\right)  \subseteq\mathcal{C}_{n}$ can be relaxed. In the
appendix, another proof is made where this relaxation is performed and the
correspondence between the vertices of $\mathcal{H}_{n}$ and the vertices of
$\mathcal{C}_{n}$ is simply considered. This requires expressing the problem
in secondary variables along with box constraints.

With the ideas presented above, the optimization problem (\textbf{P}%
$_{n}^{\prime}$) is transformed into the following linear optimization problem:

\begin{description}
\item[\textbf{(LP}$_{n}$\textbf{):}] $\left\{
\begin{array}
[c]{l}%
\min\tilde{c}^{T}w\\
\text{such that for }1\leq i<j<k\leq n\text{:}\\
\left\{
\begin{array}
[c]{c}%
B\mathbf{\lambda}^{\left(  i,j,k\right)  }-w_{i,j,k}=0\text{ }\\
u^{T}\mathbf{\lambda}^{\left(  i,j,k\right)  }\mathbf{=}1\\
w^{\left(  i,j,k\right)  }\geq0\text{, }\mathbf{\lambda}^{\left(
i,j,k\right)  }\geq0
\end{array}
\right.  \\
\text{and for }1\leq j<k\leq n\text{:}\\
\left\{
\begin{array}
[c]{c}%
r^{T}\left(  w_{123}-w_{1jk}\right)  =0\text{ with }\left(  j,k\right)
\neq\left(  2,3\right)  \\
r^{T}\left(  w_{213}-w_{2jk}\right)  =0\text{ with }\left(  j,k\right)
\neq\left(  1,3\right)  \\
r^{T}\left(  w_{i,1,2}-w_{i,j,k}\right)  =0\text{ with }\left(  j,k\right)
\neq(1,2)\text{, }i\geq3
\end{array}
\right.
\end{array}
\right.  $
\end{description}

where $\mathbf{\lambda}^{\left(  i,j,k\right)  }\in\mathbb{R}^{8}$, $B=\left(
\phi\left(  p_{1}\right)  ,\ldots,\phi\left(  p_{8}\right)  \right)  $ is the
matrix in $\left(  \ref{EQ13}\right)  $ , $r^{T}=\frac{1}{2}\left(
1,1,1,1,-1,-1\right)  $ and $u$ is an all-ones vector of appropriate dimension.

Calling $N=\binom{n}{3}$, $N_{1}=\binom{n}{2}$ , and $N_{2}=n\left(
\binom{n-1}{2}-1\right)  $, the problem \textbf{(LP}$_{n}$\textbf{)} can be
written in matrix form:%
\begin{equation}
\text{\textbf{(LP}}_{n}\text{\textbf{)}}\mathbf{:}\left\{
\begin{array}
[c]{l}%
\min\tilde{c}^{T}w\\
\text{such that :}\\
\left(
\begin{array}
[c]{cc}%
A_{11} & A_{12}\\
A_{21} & A_{22}\\
A_{31} & A_{32}%
\end{array}
\right)  \left(
\begin{array}
[c]{c}%
\lambda\\
w
\end{array}
\right)  =\left(
\begin{array}
[c]{c}%
0\\
0\\
u
\end{array}
\right) \\
w\geq0\text{, }\lambda\geq0
\end{array}
\right.  \label{EQ18}%
\end{equation}
where $A_{11}\in\mathbb{R}^{6N\times8N}$, $A_{12}\in\mathbb{R}^{6N\times
2N_{1}}$, $A_{22}\in\mathbb{R}^{N_{2}\times2N_{1}}$, $A_{31}\in\mathbb{R}%
^{N\times8N}$, $u$ is an all-ones vector of dimension $N$, and $A_{21}$ and
$A_{32}$ are zero matrices of dimensions $N_{2}\times8N$ and $N\times2N_{1}$
respectively. The primary variables $w$ and the secondary variables have
appropriate dimensions. The equivalence of problems (\textbf{P}$_{n}$),
(\textbf{P'}$_{n}$), and (\textbf{LP}$_{n}$) is described in Figure \ref{FIG5},
where
\[
\tilde{A}_{11}=\left(
\begin{array}
[c]{c}%
A_{11}\\
A_{21}\\
A_{31}%
\end{array}
\right)  \text{, }\tilde{A}_{22}=\left(
\begin{array}
[c]{c}%
A_{12}\\
A_{22}\\
A_{32}%
\end{array}
\right)  \text{, }\tilde{b}=\left(
\begin{array}
[c]{c}%
0\\
0\\
u
\end{array}
\right).
\]

\begin{figure*}
[hptb]
\begin{center}
\tikzset{every picture/.style={line width=0.75pt}} 

\begin{tikzpicture}[x=0.75pt,y=0.75pt,yscale=-1,xscale=1]

\draw   (11.55,36.67) .. controls (11.55,26.06) and (20.15,17.47) .. (30.75,17.47) -- (277.92,17.47) .. controls (288.52,17.47) and (297.12,26.06) .. (297.12,36.67) -- (297.12,94.27) .. controls (297.12,104.87) and (288.52,113.47) .. (277.92,113.47) -- (30.75,113.47) .. controls (20.15,113.47) and (11.55,104.87) .. (11.55,94.27) -- cycle ;

\draw   (298,59.74) -- (315.5,39.74) -- (315.5,49.74) -- (350.5,49.74) -- (350.5,39.74) -- (368,59.74) -- (350.5,79.74) -- (350.5,69.74) -- (315.5,69.74) -- (315.5,79.74) -- cycle ;
\draw   (165,176.3) .. controls (165,160.13) and (178.11,147.01) .. (194.29,147.01) -- (428.44,147.01) .. controls (444.62,147.01) and (457.73,160.13) .. (457.73,176.3) -- (457.73,264.18) .. controls (457.73,280.35) and (444.62,293.47) .. (428.44,293.47) -- (194.29,293.47) .. controls (178.11,293.47) and (165,280.35) .. (165,264.18) -- cycle ;

\draw   (378,27.5) .. controls (378,15.64) and (387.62,6.01) .. (399.49,6.01) -- (543.24,6.01) .. controls (555.11,6.01) and (564.73,15.64) .. (564.73,27.5) -- (564.73,91.98) .. controls (564.73,103.85) and (555.11,113.47) .. (543.24,113.47) -- (399.49,113.47) .. controls (387.62,113.47) and (378,103.85) .. (378,91.98) -- cycle ;

\draw   (110.58,120.18) -- (137.04,117.69) -- (130.17,124.95) -- (155.58,149.01) -- (162.46,141.75) -- (161.42,168.3) -- (134.96,170.8) -- (141.83,163.53) -- (116.42,139.47) -- (109.54,146.73) -- cycle ;
\draw   (511.42,120.18) -- (484.96,117.69) -- (491.83,124.95) -- (466.42,149.01) -- (459.54,141.75) -- (460.58,168.3) -- (487.04,170.8) -- (480.17,163.53) -- (505.58,139.47) -- (512.46,146.73) -- cycle ;

\draw (18.83,34.47) node [anchor=north west][inner sep=0.75pt]   [align=left] {\begin{tabular}
[c]{l}%
$\text{Problem }\left(  P_{n}\right)  \text{:}$\\
$\text{minimize }f\left(  x\right)  =x^{T} Qx+b^{T} x=c^{T}\alpha\left(  x\right)  $\\
$\text{subject to }x\in\mathcal{H}_{n}$%
\end{tabular}};
\draw (394.87,24.74) node [anchor=north west][inner sep=0.75pt]   [align=left] {\begin{tabular}
[c]{l}%
$\text{Problem }\left(  P_{n}^{\prime}\right)  \text{:}$\\
$\text{minimize }\tilde{f}\left(  w\right)  =\tilde{c}^{T}w$\\
$\text{subject to }w\in\mathcal{C}_{n}$%
\end{tabular}};
\draw (201.37,164.74) node [anchor=north west][inner sep=0.75pt]   [align=left]
{\begin{tabular}
[c]{l}%
$\text{Problem }\left(  LP_{n}\right)  \text{:}$\\
$\text{minimize }\tilde{f}\left(  w\right)  =\tilde{c}^{T}w$\\
$\text{subject to }\left(
\begin{array}
[c]{cc}%
\tilde{A}_{11} & \tilde{A}_{12}%
\end{array}
\right)  \left(
\begin{array}
[c]{c}%
\lambda\\
w
\end{array}
\right)  =\tilde{b}$\\
and $w\geq0$, $\lambda\geq0$%
\end{tabular}};

\end{tikzpicture}
\caption{Equivalence of problems (\textbf{P}$_{n}$),
(\textbf{P'}$_{n}$) and (\textbf{LP}$_{n}$)}%
\label{FIG5}%
\end{center}
\end{figure*}
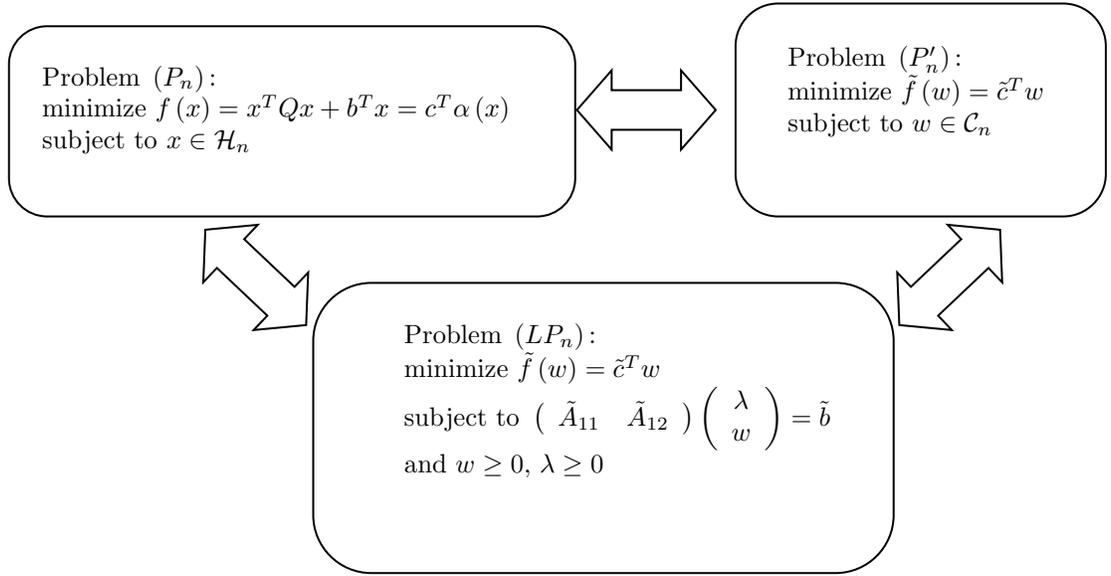

\begin{example}
\label{EXAMPLE1}Let $Q\in\mathbb{Z}^{4\times4}$ be the symmetric matrix:%
\[
Q=\left(
\begin{array}
[c]{cccc}%
0 & -30 & 6 & -22\\
-30 & 0 & 15 & -2\\
6 & 15 & 0 & -5\\
-22 & -2 & -5 & 0
\end{array}
\right),
\]
and $b\in\mathbb{Z}^{4}$ the vector%
\[
b=\left(
\begin{array}
[c]{c}%
-8\\
-22\\
0\\
-32
\end{array}
\right).
\]

The consistency constraints were given in Example \ref{EXAMPLE2}. The optimal
value of $f$ is $-170$ and
\begin{align*}
\lambda^{\ast}  &  =e_{7}+e_{16}+e_{22}+e_{30},\\
w^{\ast}  &  =\left(  2,\frac{1}{2},2,\frac{1}{2},2,\frac{1}{2},0,\frac{1}%
{2},0,\frac{1}{2},0,\frac{1}{2}\right)  ^{T},%
\end{align*}
where $e_{k}\in\mathbb{R}^{32}$ is the k-th vector of the standard basis (with
an entry '1' at the position $k$ and '0' for the rest of positions).

The optimal point $x^{\ast}$ can be recovered from $w$ by applying $E_{3}$ to
$w_{ijk}$. To illustrate this point we compute $x^{\ast}$ for the previous example.

\end{example}
\begin{example}
From example \ref{EXAMPLE1}$:$
\begin{align*}
w_{123}^{\ast}  &  =\left(  2,\frac{1}{2},\frac{1}{2},0,\frac{1}{2},\frac
{1}{2}\right)  ^{T},\\
w_{234}^{\ast}  &  =\left(  \frac{1}{2},2,\frac{1}{2},\frac{1}{2},0,\frac
{1}{2}\right)  ^{T}.%
\end{align*}
We can recover the point $x^{\ast}$ as follows:%
\[
E_{3}^{-1}w_{123}^{\ast}=\left(  1,1,0,1,0,0\right)  ^{T}=\left(  x_{1}%
,x_{1}x_{2},x_{1}x_{3},x_{2},x_{2}x_{3},x_{3}\right)  ^{T},%
\]
and this means $x_{1}^{\ast}=x_{2}^{\ast}=1$, $x_{3}^{\ast}=0$. \ We proceed
in a similar way for $w_{234}^{\ast}$:%
\[
E_{3}^{-1}w_{234}^{\ast}=\left(  1,0,1,0,0,1\right)  ^{T}=\left(  x_{2}%
,x_{2}x_{3},x_{2}x_{4},x_{3},x_{3}x_{4},x_{4}\right)  ^{T},%
\]
which implies that $x_{2}^{\ast}=1$, $x_{3}^{\ast}=0$, $x_{4}^{\ast}=1$. The
optimal point is $x^{\ast}=\left(  1,1,0,1\right)  $.
\end{example}

\section{Computational Complexity}

The problem (\textbf{LP}$_{n}$) requires an amount of memory given by the
dimension of $A$ in $\left(  \ref{EQ18}\right)  $, that is $\left(
7N+N_{2}\right)  \left(  8N+2N_{1}\right)  $. Henceforth, the space complexity
is of order $O\left(  n^{6}\right)  $. Since both $B$ (involved in the
definition of the convex-hull $\mathcal{C}_{3}^{\left(  i,j,k\right)  }$) and
$E_{3}$ (corresponding to the constraint $\sum_{l=1}^{8}\mathbf{\lambda}%
_{l}^{\left(  i,j,k\right)  }=1$ in the convex-hull $\mathcal{C}_{3}^{\left(
i,j,k\right)  }$) are constant, the time complexity is also $O\left(
n^{6}\right)  $ (assuming that the storing of an entry in a matrix is
$O\left(  1\right)  $). Note that the objective function $\tilde{f}\left(
w\right)  =\tilde{c}^{T}w$ requires $2N_{1}$ multiplications and $2N_{1}-1$
sums, resulting in a time complexity of order $O\left(  n^{2}\right)  $.

Once generated the matrix $A$ and the vector $\tilde{c}$, the problem
(\textbf{LP}$_{n}$) can be solved in polynomial-time via interior-point
methods (the reader is referred to \cite{Gonzaga1995} for more details): (i)
ellipsoid method due to Khachiyan ($O\left(  m^{4}L\right)  $ where $L$
denotes the number of bits in a binary representation of $A$ and $m$ is the
space dimension, (ii) projective algorithm of Karmarkar ($O\left(  m^{\frac
{7}{2}}L\right)  $), (iii) Gonzaga algorithm and Vaidya's 87 algorithm, with
the complexity of $O\left(  m^{3}L\right)  $ operations (these two algorithms
were simultaneously developed in 1987), or (iv) Vaidya's 89 algorithm
($O\left(  m^{\frac{5}{2}}\right)  $), among others. Henceforth, the problem
(P$_{n}$) is solved in polynomial time. The number of variables in (PL$_{n}$)
is $8N+2N_{1}$ so that the problem can be solved with Vaidya's 89 algorithm in
$O\left(  n^{\frac{15}{2}}\right)  $.

\section{Implementation Aspects}

This section addresses the implementation aspects of the equivalent linear
program \textbf{(LP}$_{n}$\textbf{)}. We will develop the different algorithms
both for the transformation of the objective function in the form $\tilde
{f}\left(  w\right)  =\tilde{c}^{T}w$ through the linear transformation
$T_{n}$ and for the equality constraints (convexity constraints and
consistency constraints). The convexity and consistency constraints allow us
to define the feasible set $\mathcal{C}_{n}$ as follows:%
\begin{gather*}
\mathcal{C}_{n}=\bigg \{w\in\left[  0,2\right]  ^{n}\times\left[  0,\frac{1}%
{2}\right]  ^{n}:\left(
\begin{array}
[c]{cc}%
\tilde{A}_{11} & \tilde{A}_{12}%
\end{array}
\right)  \left(
\begin{array}
[c]{c}%
\lambda\\
w
\end{array}
\right)  =\tilde{b},\\
\lambda\geq0 \bigg \},
\end{gather*}
where%
\[
\tilde{A}_{11}=\left(
\begin{array}
[c]{c}%
A_{11}\\
A_{21}\\
A_{31}%
\end{array}
\right)  \text{, }\tilde{A}_{12}=\left(
\begin{array}
[c]{c}%
A_{12}\\
A_{22}\\
A_{32}%
\end{array}
\right)  \text{, }\tilde{b}=\left(
\begin{array}
[c]{c}%
0\\
0\\
u
\end{array}
\right).
\]
We will start with the indexing of the primary variables since this point is
key for a correct implementation of the method.

\subsection{Definition of Primary Variables}

As indicated above, the primary variables are stored in a vector
$w\in\mathbb{R}^{n\left(  n-1\right)  }$ as a stacking of the $\binom{n}{2}%
$-dimensional vectors $u$, and $v$:%
\[
w=\left(
\begin{array}
[c]{c}%
u\\
v
\end{array}
\right).
\]

In turn, each element in $u$ and $v$ is defined by a pair of indices
$\left(  i,j\right)  $ such that $1\leq i<j\leq n$. The element $u_{ij}$ (idem
for $v_{ij}$), is stored in $u$ (in $v$) in the position $\left(  i-1\right)
\frac{2n-i}{2}+\left(  j-i\right)  $ when $i<j$ (in the position
$N_{1}+\left(  i-1\right)  \frac{2n-i}{2}+\left(  j-i\right)  $ for $v_{ij}$).
This can be analyzed simply according to the scheme below where
different subvectors are represented according to the first index:%

\begin{gather*}[h]
\underset{\text{Length=}\sum_{k=1}^{i-1}\left(  n-k\right)  }{\underbrace
{\left.
\begin{array}
[c]{c}%
u_{1,2}\\
u_{1,3}\\
\vdots\\
u_{1,n}%
\end{array}
\right\}  \text{Length}=n-1,\left.
\begin{array}
[c]{c}%
u_{2,3}\\
u_{2,4}\\
\vdots\\
u_{2,n}%
\end{array}
\right\}  \text{Length}=n-2,\cdots}},\\
\left.
\begin{array}
[c]{c}%
u_{i,i+1}\\
u_{i,i+2}\\
\vdots\\
u_{i,n}%
\end{array}
\right\}  \text{Length}=n-i,\cdots,\left.
\begin{array}
[c]{c}%
u_{i+1,i+2}\\
u_{i+2,i+3}\\
\vdots\\
u_{i+1,n}%
\end{array}
\right\}
\begin{array}
[c]{c}%
\text{Length}=\\
n-\left(  i+1\right)
\end{array}
\\
\cdots,\left.
\begin{array}
[c]{c}%
u_{n-2,n-1}\\
u_{n-2,n}%
\end{array}
\right\}  \text{Length}=2,\left.
\begin{array}
[c]{c}%
u_{n-1,n}%
\end{array}
\right\}  \text{Length}=1
\end{gather*}

The i-th subvector, that is $\left(  u_{i,i+1},u_{i,i+2},\cdots,u_{i,n}%
\right)  $ has above $i-1$ subvectors with lengths from $n-1$ to $n-\left(
i-1\right)  $. Therefore, the number of entries above the i-th subvector is%
\[
\sum_{k=1}^{i-1}\left(  n-k\right)  =n\left(  i-1\right)  -\frac{i\left(
i-1\right)  }{2}=\left(  i-1\right)  \frac{2n-i}{2}.%
\]

Finally, within the subvector $i$, the first element is $u_{i,i+1}$, so the
input $u_{i,j}$ is in position $\left(  j-i\right)  $ within this subvector:%
\[
\underset{\left(  \left.  \left.
\begin{array}
[c]{c}%
u_{i,j+1}\\
\vdots\\
u_{i,n}%
\end{array}
\right.  \right.  \right)  }{\left.  \left(
\begin{array}
[c]{c}%
u_{i,i+1}\\
u_{i,i+2}\\
\vdots\\
u_{i,j}%
\end{array}
\right)  \right\}  }\left(  j-i\right)  \text{ positions}%
\]
Similarly we can position ourselves in vector $v$: when $i\geq j$, the
position of the element $u_{i,j}$ can be calculated immediately by swapping
the roles of $i$ and $j$. Based on the previous information, we define the
positioning index function:%
\[
\iota\left(  i,j\right)  =\left\{
\begin{array}
[c]{cc}%
\left(  i-1\right)  \frac{2n-i}{2}+\left(  j-i\right)   & \text{, when }i<j\\
\left(  j-1\right)  \frac{\left(  2n-j\right)  }{2}+\left(  i-j\right)   &
\text{, when }i\geq j
\end{array}
\right.
\]
Throughout this presentation, we will exemplify the main ideas and algorithms
for a problem of dimension $n=4$.

\begin{example}
For $n=4$, in Figures \ref{TAB:Primaryvariablesu} and
\ref{TAB:Primaryvariablesv} the correspondence between primary variables and
the indices generated by the $\iota$ function has been represented.

\begin{figure*}[ht]
\begin{subfigure}{0.49\linewidth}
\centering
$\overset{\text{Primary variables }u}{\overbrace{%
\begin{tabular}
[c]{cccccc}%
$\iota\left(  1,2\right)  $ & $\iota\left(  1,3\right)  $ & $\iota\left(
1,4\right)  $ & $\iota\left(  2,3\right)  $ & $\iota\left(  2,4\right)  $ &
$\iota\left(  3,4\right)  $\\
$1$ & $2$ & $3$ & $4$ & $5$ & $6$\\
$\downarrow$ & $\downarrow$ & $\downarrow$ & $\downarrow$ & $\downarrow$ &
$\downarrow$\\\hline
\multicolumn{1}{|c}{$u_{12}$} & \multicolumn{1}{|c}{$u_{13}$} &
\multicolumn{1}{|c}{$u_{14}$} & \multicolumn{1}{|c}{$u_{23}$} &
\multicolumn{1}{|c}{$u_{24}$} & \multicolumn{1}{|c|}{$u_{34}$}\\\hline
\end{tabular}
}}$ 
\caption{Primary variables $u$ in the vector $w$ 
}
\label{TAB:Primaryvariablesu}
\end{subfigure}\hfill

\begin{subfigure}{0.49\linewidth}
\centering
$\overset{\text{Primary variables }v}{\overbrace{%
\begin{tabular}
[c]{cccccc}%
$N_{1}+\iota\left(  1,2\right)  $ & $N_{1}+\iota\left(  1,3\right)  $ &
$N_{1}+\iota\left(  1,4\right)  $ & $N_{1}+\iota\left(  2,3\right)  $ &
$N_{1}+\iota\left(  2,4\right)  $ & $N_{1}+\iota\left(  3,4\right)  $\\
$6+1$ & $6+2$ & $6+3$ & $6+4$ & $6+5$ & $6+6$\\
$\downarrow$ & $\downarrow$ & $\downarrow$ & $\downarrow$ & $\downarrow$ &
$\downarrow$\\\hline
\multicolumn{1}{|c}{$v_{12}$} & \multicolumn{1}{|c}{$v_{13}$} &
\multicolumn{1}{|c}{$v_{14}$} & \multicolumn{1}{|c}{$v_{23}$} &
\multicolumn{1}{|c}{$v_{24}$} & \multicolumn{1}{|c|}{$v_{34}$}\\\hline
\end{tabular}
}}$ 
\caption{Primary variables $v$ in the vector $w$ 
}
\label{TAB:Primaryvariablesv}
\end{subfigure}\hfill

\caption{Positions of the primary variables in the vector $w$ according to the index $\iota$.}
\label{}
\end{figure*}

\end{example}

\subsection{Transformation of Objective Function}

In this subsection we will build the linear transformation $T_{n}$ that allows
us to transform the objective function $f:\mathcal{H}_{n}\rightarrow
\mathbb{R}$ into $\tilde{f}:\mathcal{C}_{n}\rightarrow\mathbb{R}$.
Specifically, $\tilde{f}\left(  w\right)  =\tilde{c}^{T}w$ where  $\tilde
{c}^{T}=c^{T}T_{n}$. The matrix $T_{n}$ is obtained by exploiting the
following relationships:%
\begin{align*}
x_{i} &  =\frac{u_{ij}+v_{ij}+u_{ik}+v_{ik}-u_{jk}-v_{jk}}{2},\\
x_{i}x_{j} &  =\frac{u_{ij}-v_{ij}}{2},%
\end{align*}
such that $T_{n}E_{n}\alpha\left(  x\right)  =\alpha\left(  x\right)  $,
remember that $w=\phi\left(  x\right)  =E_{n}\alpha\left(  x\right)  $. We
will illustrate the construction of $T_{n}$ with an example:

\begin{example}
For $n=4$ we have that $N_{1}=\binom{4}{2}=6$ ($=$ number of primary variables
$u=$ number of primary variables $v$). In Tables \ref{tab:MatrixTPrimaryu} and \ref{tab:MatrixTPrimaryv}, the $T_{4}$ matrix is represented indicating the correspondence of the rows with the single variables
$x_{1}$, $x_{2}$, $x_{3}$, and $x_{4}$, as well as the cross products
$x_{1}x_{2}$, $x_{1}x_{3}$, $x_{1}x_{4}$, $x_{2}x_{3}$, $x_{2}x_{4}$ and
$x_{3}x_{4}$. The vector $T_{4}w$ is a linear combination of the column
vectors in $T_{4}$ through the primary variables $u$ and $v$ (this has been
highlighted with each column indicating which primary variable it is
associated with). Furthermore, the primary variables that correspond to each other
are found within the vector $w$ according to the index $\iota\left(
i,j\right)  $. In reality, the matrix $T_{4}$ is used to express the variables
$x_{i}$ as well as the cross products $x_{i}x_{j}$ with $j>i$ as a function of
the primary variables. 
For clarity we will divide the matrix $T_{4}$ into a part $T_{4}^{u}$
corresponding to the primary variables $u$ and another $T_{4}^{v}$ to the
primary variables $v$ so that $T_{4}=\left(
\begin{array}
[c]{cc}%
T_{4}^{u} & T_{4}^{v}%
\end{array}
\right)  $. Tables  \ref{tab:MatrixTPrimaryu} and \ref{tab:MatrixTPrimaryv}
represent $T_{4}^{u}$ and $T_{4}^{v}$ respectively.

\begin{table}[htbp]
   \centering
   \caption{Part of the matrix $T_{4}$ corresponding to the primary variables $u$.}
   \label{tab:MatrixTPrimaryu}
   \begin{tabular}
[c]{ccccccc}
& $\iota\left(  1,2\right)  $ & $\iota\left(  1,3\right)  $ & $\iota\left(
1,4\right)  $ & $\iota\left(  2,3\right)  $ & $\iota\left(  2,4\right)  $ &
$\iota\left(  3,4\right)  $\\
& $1$ & $2$ & $3$ & $4$ & $5$ & $6$\\
& $\downarrow$ & $\downarrow$ & $\downarrow$ & $\downarrow$ & $\downarrow$ &
$\downarrow$\\\cline{2-7}%
$x_{1}\rightarrow$ & \multicolumn{1}{|c}{$1$} & \multicolumn{1}{|c}{$1$} &
\multicolumn{1}{|c}{$0$} & \multicolumn{1}{|c}{$-1$} & \multicolumn{1}{|c}{$0$%
} & \multicolumn{1}{|c|}{$0$}\\\cline{2-7}%
$x_{1}x_{2}\rightarrow$ & \multicolumn{1}{|c}{$1$} & \multicolumn{1}{|c}{$0$}
& \multicolumn{1}{|c}{$0$} & \multicolumn{1}{|c}{$0$} &
\multicolumn{1}{|c}{$0$} & \multicolumn{1}{|c|}{$0$}\\\cline{2-7}%
$x_{1}x_{3}\rightarrow$ & \multicolumn{1}{|c}{$0$} & \multicolumn{1}{|c}{$1$}
& \multicolumn{1}{|c}{$0$} & \multicolumn{1}{|c}{$0$} &
\multicolumn{1}{|c}{$0$} & \multicolumn{1}{|c|}{$0$}\\\cline{2-7}%
$x_{1}x_{4}\rightarrow$ & \multicolumn{1}{|c}{$0$} & \multicolumn{1}{|c}{$0$}
& \multicolumn{1}{|c}{$1$} & \multicolumn{1}{|c}{$0$} &
\multicolumn{1}{|c}{$0$} & \multicolumn{1}{|c|}{$0$}\\\cline{2-7}%
$x_{2}\rightarrow$ & \multicolumn{1}{|c}{$1$} & \multicolumn{1}{|c}{$-1$} &
\multicolumn{1}{|c}{$0$} & \multicolumn{1}{|c}{$1$} & \multicolumn{1}{|c}{$0$}
& \multicolumn{1}{|c|}{$0$}\\\cline{2-7}%
$x_{2}x_{3}\rightarrow$ & \multicolumn{1}{|c}{$0$} & \multicolumn{1}{|c}{$0$}
& \multicolumn{1}{|c}{$0$} & \multicolumn{1}{|c}{$1$} &
\multicolumn{1}{|c}{$0$} & \multicolumn{1}{|c|}{$0$}\\\cline{2-7}%
$x_{2}x_{4}\rightarrow$ & \multicolumn{1}{|c}{$0$} & \multicolumn{1}{|c}{$0$}
& \multicolumn{1}{|c}{$0$} & \multicolumn{1}{|c}{$0$} &
\multicolumn{1}{|c}{$1$} & \multicolumn{1}{|c|}{$0$}\\\cline{2-7}%
$x_{3}\rightarrow$ & \multicolumn{1}{|c}{$-1$} & \multicolumn{1}{|c}{$1$} &
\multicolumn{1}{|c}{$0$} & \multicolumn{1}{|c}{$1$} & \multicolumn{1}{|c}{$0$}
& \multicolumn{1}{|c|}{$0$}\\\cline{2-7}%
$x_{3}x_{4}\rightarrow$ & \multicolumn{1}{|c}{$0$} & \multicolumn{1}{|c}{$0$}
& \multicolumn{1}{|c}{$0$} & \multicolumn{1}{|c}{$0$} &
\multicolumn{1}{|c}{$0$} & \multicolumn{1}{|c|}{$1$}\\\cline{2-7}%
$x_{4}\rightarrow$ & \multicolumn{1}{|c}{$-1$} & \multicolumn{1}{|c}{$0$} &
\multicolumn{1}{|c}{$1$} & \multicolumn{1}{|c}{$0$} & \multicolumn{1}{|c}{$1$}
& \multicolumn{1}{|c|}{$0$}\\\cline{2-7}
& $\uparrow$ & $\uparrow$ & $\uparrow$ & $\uparrow$ & $\uparrow$ & $\uparrow
$\\
& $u_{12}$ & $u_{13}$ & $u_{14}$ & $u_{23}$ & $u_{24}$ & $u_{34}$%
\end{tabular}
\end{table}

\begin{table}[htbp]
   \centering
   \caption{Part of the matrix $T_{4}$ corresponding to the primary variables $v$.}
   \label{tab:MatrixTPrimaryv}
   \begin{tabular}
[c]{cccc}
& $N_{1}+\iota\left(  1,2\right)  $ & $N_{1}+\iota\left(  1,3\right)  $ &
$N_{1}+\iota\left(  1,4\right)  $\\
& $6+1$ & $6+2$ & $6+3$\\
& $\downarrow$ & $\downarrow$ & $\downarrow$\\\cline{2-4}%
$x_{1}\rightarrow$ & \multicolumn{1}{|c}{$\frac{1}{2}$} &
\multicolumn{1}{|c}{$\frac{1}{2}$} & \multicolumn{1}{|c|}{$0$}\\\cline{2-4}%
$x_{1}x_{2}\rightarrow$ & \multicolumn{1}{|c}{$-\frac{1}{2}$} &
\multicolumn{1}{|c}{$0$} & \multicolumn{1}{|c|}{$0$}\\\cline{2-4}%
$x_{1}x_{3}\rightarrow$ & \multicolumn{1}{|c}{$0$} &
\multicolumn{1}{|c}{$-\frac{1}{2}$} & \multicolumn{1}{|c|}{$0$}\\\cline{2-4}%
$x_{1}x_{4}\rightarrow$ & \multicolumn{1}{|c}{$0$} & \multicolumn{1}{|c}{$0$}
& \multicolumn{1}{|c|}{$-\frac{1}{2}$}\\\cline{2-4}%
$x_{2}\rightarrow$ & \multicolumn{1}{|c}{$\frac{1}{2}$} &
\multicolumn{1}{|c}{$-\frac{1}{2}$} & \multicolumn{1}{|c|}{$0$}\\\cline{2-4}%
$x_{2}x_{3}\rightarrow$ & \multicolumn{1}{|c}{$0$} & \multicolumn{1}{|c}{$0$}
& \multicolumn{1}{|c|}{$0$}\\\cline{2-4}%
$x_{2}x_{4}\rightarrow$ & \multicolumn{1}{|c}{$0$} & \multicolumn{1}{|c}{$0$}
& \multicolumn{1}{|c|}{$0$}\\\cline{2-4}%
$x_{3}\rightarrow$ & \multicolumn{1}{|c}{$-\frac{1}{2}$} &
\multicolumn{1}{|c}{$\frac{1}{2}$} & \multicolumn{1}{|c|}{$0$}\\\cline{2-4}%
$x_{3}x_{4}\rightarrow$ & \multicolumn{1}{|c}{$0$} & \multicolumn{1}{|c}{$0$}
& \multicolumn{1}{|c|}{$0$}\\\cline{2-4}%
$x_{4}\rightarrow$ & \multicolumn{1}{|c}{$-\frac{1}{2}$} &
\multicolumn{1}{|c}{$0$} & \multicolumn{1}{|c|}{$\frac{1}{2}$}\\\cline{2-4}
& $\uparrow$ & $\uparrow$ & $\uparrow$\\
& $v_{23}$ & $v_{24}$ & $v_{34}$%
\end{tabular}
\end{table}

\addtocounter{table}{-1}

\begin{table}[htbp]
   \centering
   \caption{Part of the matrix $T_{4}$ corresponding to the primary variables $v$ (\emph{Continuation}).}
   \begin{tabular}
[c]{cccc}
& $N_{1}+\iota\left(  2,3\right)  $ & $N_{1}+\iota\left(  2,4\right)  $ &
$N_{1}+\iota\left(  3,4\right)  $\\
& $6+4$ & $6+5$ & $6+6$\\
& $\downarrow$ & $\downarrow$ & $\downarrow$\\\cline{2-4}%
$x_{1}\rightarrow$ & \multicolumn{1}{|c}{$-\frac{1}{2}$} &
\multicolumn{1}{|c}{$0$} & \multicolumn{1}{|c|}{$0$}\\\cline{2-4}%
$x_{1}x_{2}\rightarrow$ & \multicolumn{1}{|c}{$0$} & \multicolumn{1}{|c}{$0$}
& \multicolumn{1}{|c|}{$0$}\\\cline{2-4}%
$x_{1}x_{3}\rightarrow$ & \multicolumn{1}{|c}{$0$} & \multicolumn{1}{|c}{$0$}
& \multicolumn{1}{|c|}{$0$}\\\cline{2-4}%
$x_{1}x_{4}\rightarrow$ & \multicolumn{1}{|c}{$0$} & \multicolumn{1}{|c}{$0$}
& \multicolumn{1}{|c|}{$0$}\\\cline{2-4}%
$x_{2}\rightarrow$ & \multicolumn{1}{|c}{$\frac{1}{2}$} &
\multicolumn{1}{|c}{$0$} & \multicolumn{1}{|c|}{$0$}\\\cline{2-4}%
$x_{2}x_{3}\rightarrow$ & \multicolumn{1}{|c}{$-\frac{1}{2}$} &
\multicolumn{1}{|c}{$0$} & \multicolumn{1}{|c|}{$0$}\\\cline{2-4}%
$x_{2}x_{4}\rightarrow$ & \multicolumn{1}{|c}{$0$} &
\multicolumn{1}{|c}{$-\frac{1}{2}$} & \multicolumn{1}{|c|}{$0$}\\\cline{2-4}%
$x_{3}\rightarrow$ & \multicolumn{1}{|c}{$\frac{1}{2}$} &
\multicolumn{1}{|c}{$0$} & \multicolumn{1}{|c|}{$0$}\\\cline{2-4}%
$x_{3}x_{4}\rightarrow$ & \multicolumn{1}{|c}{$0$} & \multicolumn{1}{|c}{$0$}
& \multicolumn{1}{|c|}{$-\frac{1}{2}$}\\\cline{2-4}%
$x_{4}\rightarrow$ & \multicolumn{1}{|c}{$0$} & \multicolumn{1}{|c}{$\frac
{1}{2}$} & \multicolumn{1}{|c|}{$0$}\\\cline{2-4}
& $\uparrow$ & $\uparrow$ & $\uparrow$\\
& $v_{23}$ & $v_{24}$ & $v_{34}$%
\end{tabular}

\end{table}

When $i=1$, the variable $x_{1}$
is expressed as a function of $u$ and $v$ as%
\[
x_{1}=\frac{u_{12}+u_{13}-u_{23}+v_{12}+v_{13}-v_{23}}{2},%
\]
and cross products as%
\begin{equation}
x_{1}x_{2}=\frac{u_{12}-v_{12}}{2}.\label{EQ132}%
\end{equation}
The variable $x_{1}$ corresponds to row $1$ in $T_{4}$, the primary variables
$u_{12}$, $u_{13}$ and $u_{23}$ are associated with columns $\iota\left(
1,2\right)  =1$, $\iota\left(  1,3\right)  =2$ and $\iota\left(  2,3\right)
=4$, while the primary variables $v_{12}$, $v_{13}$ and $v_{23}$ are
placed in columns $N_{1}+\iota\left(  1,2\right)  =7$, $N_{1}+\iota\left(
1,3\right)  =8$, and $N_{1}+\iota\left(  2,3\right)  =10$. The cross product
$x_{1}x_{2}$ corresponds to row $2$ in $T_{4}$. According to $\left(
\ref{EQ132}\right)  $ we must access columns $\iota\left(  1,2\right)  =1$
(corresponding to $u_{12}$) and $N_{1}+\iota\left(  1,2\right)  =7$
(corresponding to $v_{12}$). Following this process we would finish building
the matrix $T_{4}$ completely.
\end{example}

Algorithm \ref{Alg:MatrixT} calculates $T_{n}$ for an arbitrary $n$ according to these ideas.

\begin{algorithm}
\caption{Transformation of objective function}\label{Alg:MatrixT}
\begin{algorithmic}[3]
\Procedure{Transformed Objective Function}{$c$}\Comment{Computation of $\tilde{c}$ in $\tilde{f}\left(  w\right)  =\tilde{c}^{T}w$}

\State $T \gets 0_{\frac{n\left(  n+1\right)  }{2}\times2N_{1}}$
\State $r \gets 0 $

\For{$i\gets 1, n$}
    \State $S\leftarrow\left\{  1,2\ldots,n\right\}  $
    \State $r\gets r+1$
    \If{i=1} 
        \State $T_{r,8N+\iota\left(  1,2\right)  }\gets\frac{1}{2}$
        \State $T_{r,8N+\iota\left(  1,3\right)  }\gets\frac{1}{2}$
        \State $T_{r,8N+\iota\left(  2,3\right)  }\gets-\frac{1}{2}$
        \State $T_{r,8N+N_{1}+\iota\left(  1,2\right)  }\gets\frac{1}{2}$
        \State $T_{r,8N+N_{1}+\iota\left(  1,3\right)  }\gets\frac{1}{2}$
        \State $T_{r,8N+N_{1}+\iota\left(  2,3\right)  }\gets-\frac{1}{2}$
    \ElsIf{i=2}
        \State $T_{r,8N+\iota\left(  1,2\right)  }\gets\frac{1}{2}$
        \State $T_{r,8N+\iota\left(  2,3\right)  }\gets\frac{1}{2}$
        \State $T_{r,8N+\iota\left(  1,3\right)  }\gets-\frac{1}{2}$
        \State $T_{r,8N+N_{1}+\iota\left(  1,2\right)  }\gets\frac{1}{2}$
        \State $T_{r,8N+N_{1}+\iota\left(  2,3\right)  }\gets\frac{1}{2}$
        \State $T_{r,8N+N_{1}+\iota\left(  1,3\right)  }\gets-\frac{1}{2}$
    \Else 
        \State $T_{r,8N+\iota\left(  1,2\right)  }\gets-\frac{1}{2}$
        \State $T_{r,8N+\iota\left(  1,i\right)  }\gets\frac{1}{2}$
        \State $T_{r,8N+\iota\left(  2,i\right)  }\gets\frac{1}{2}$
        \State $T_{r,8N+N_{1}+\iota\left(  1,2\right)  }\gets-\frac{1}{2}$
        \State $T_{r,8N+N_{1}+\iota\left(  1,i\right)  }\gets\frac{1}{2}$
        \State $T_{r,8N+N_{1}+\iota\left(  2,i\right)  }\gets\frac{1}{2}$

    \EndIf
    
     \State Generate an ordered array $C$ of the elements of $\bar{S}$ taken two by two.
     
     \For{$j\gets i+1, n$}
        \State $r \gets r+1$
        \State $T_{r,8N+\iota\left(  i,j\right)  \gets \frac{1}{2}}$
        \State $T_{r,8N+N_{1}+\iota\left(  i,j\right)  }\gets-\frac{1}{2}$
     \EndFor

\EndFor

\State $\tilde{c}^{T} \gets c^{T}T$

\Return $\tilde{c}$

\EndProcedure
\end{algorithmic}
\end{algorithm}

\subsection{Equality Constraints}

At the beginning of the procedure, we will assume that $\tilde{A}$ is an empty
matrix that we will fill in. This matrix is of dimension $7N+N_{2}%
\times8N+2N_{1}$.

\subsubsection{Convexity Constraints}

For the secondary variables $\mathbf{\lambda}$ we have adopted the following
notation%
\[
\mathbf{\lambda}=\left(
\begin{array}
[c]{c}%
\mathbf{\lambda}^{\left(  1,2,3\right)  }\\
\mathbf{\lambda}^{\left(  1,2,4\right)  }\\
\vdots\\
\mathbf{\lambda}^{\left(  n-2,n-1,n\right)  }%
\end{array}
\right),
\]
where $\mathbf{\lambda}^{\left(  i,j,k\right)  }\in\left[  0,1\right]  ^{8}$
(this notation does not follow the definition of $w_{i,j,k}$, that is, for
each $\left(  i,j,k\right)  $ we have a vector $\mathbf{\lambda}^{\left(
i,j,k\right)  }$  which is exclusive to the convex hull $\mathcal{C}%
_{3}^{\left(  i,j,k\right)  }$). For each triplet $\left(  i,j,k\right)  $ we
generate the convex hull $\mathcal{C}_{3}^{\left(  i,j,k\right)  }$:%
\begin{equation}
\mathcal{C}_{3}^{\left(  i,j,k\right)  }=\left\{  w_{i,j,k}=B\mathbf{\lambda
}^{\left(  i,j,k\right)  }:u^{T}\mathbf{\lambda}^{\left(  i,j,k\right)
}\mathbf{=}1\text{, }\mathbf{\lambda}^{\left(  i,j,k\right)  }\geq0\right\},
\label{EQ129}%
\end{equation}
where $u=\left(  1,1,1,1,1,1,1,1\right)  ^{T}$ and%
\[
w_{i,j,k}=\left(
\begin{array}
[c]{c}%
u_{i,j,k}\\
v_{i,j,k}%
\end{array}
\right),
\]
with $u_{i,j,k}^{\left(  i,j,k\right)  },v^{\left(  i,j,k\right)  }%
\in\mathbb{R}^{3}$. According to $\left(  \ref{EQ129}\right)  $ the convexity
constraints for $\mathcal{C}_{3}^{\left(  i,j,k\right)  }$ are equal
constraints written as%
\begin{equation}
B\mathbf{\lambda}^{\left(  i,j,k\right)  }-w_{i,j,k}=0,\label{EQ130}%
\end{equation}%
\begin{equation}
u^{T}\mathbf{\lambda}^{\left(  i,j,k\right)  }=1,\label{EQ131}%
\end{equation}
along with the natural constraint $\mathbf{\lambda}^{\left(  i,j,k\right)
}\geq0$. Let us start by looking at the implementation of $\left(
\ref{EQ130}\right)  $.

As we see in $\left(  \ref{EQ130}\right)  $ we need the basic block $B=\left(
\phi\left(  p_{1}\right)  ,\ldots,\phi\left(  p_{8}\right)  \right)  $ that is
implemented in Algorithm \ref{Algoritmo4}.

\begin{algorithm}
\caption{Basic block $B$}
\label{Algoritmo4}

\begin{algorithmic}[1]
\Procedure{Basic Block}{}\Comment{The basic block $B$}
\State $B\gets 0_{6\times8}$
\State $i\gets 0$ \Comment{it is used as a column index of $B$}
\For{$x_{1}\gets 0, 1$}
    \For{$x_{2}\gets 0, 1$}
        \For{$x_{3} \gets 0, 1$}
            \State $i\gets i+1$
            \State $B_{1,i}\gets \frac{\left(  x_{1}+x_{2}\right)  ^{2}}{2}$
            \State $B_{2,i}\leftarrow\frac{\left(  x_{1}+x_{3}\right)  ^{2}}{2}$
            \State $B_{3,i}\leftarrow\frac{\left(  x_{2}+x_{3}\right)  ^{2}}{2}$
            \State $B_{4,i}\leftarrow\frac{\left(  x_{1}-x_{2}\right)  ^{2}}{2}$
            \State $B_{5,i}\leftarrow \frac{\left(  x_{1}-x_{3}\right)  ^{2}}{2}$
            \State $B_{6,i}\leftarrow\frac{\left(  x_{2}-x_{3}\right)  ^{2}}{2}$

    \EndFor
    \EndFor
\EndFor
\EndProcedure
\end{algorithmic}
\end{algorithm}

In the previous algorithm, squares appear for simplicity in writing since we
are evaluating binary variables. For simplicity of implementation, each
triplet $\left(  i,j,k\right)  $ is associated with a single index
$r\in\left\{  1,\ldots,N\right\}  $, where we remember that $N=\binom{n}{3}$;
This index $r$ will allow us to traverse rows inside the matrix $\tilde{A}$. We
write the matrix $A_{11}\in\mathbb{R}^{6N\times8N}$ as a stack of submatrices
$A_{11}^{(r)}$ of dimension $\mathbb{R}^{6\times8N}$ each one of them
transforming a vector $\mathbf{\lambda}^{\left(  i,j,k\right)  }$:%
\[
A_{11}=\left(
\begin{array}
[c]{c}%
A_{11}^{(1)}\\
A_{11}^{\left(  2\right)  }\\
\vdots\\
A_{11}^{\left(  N\right)  }%
\end{array}
\right).
\]
These sub-matrices $A_{11}^{\left(  r\right)  }$ are divided according to
$\left(  i,j,k\right)  $:%

\[
A_{11}^{(r)}=\left(  \underset{\text{Part corresponding to the secondary
variables }\mathbf{\lambda}}{\underbrace{%
\begin{array}
[c]{cccc}%
\mathbf{\lambda}^{\left(  1,2,3\right)  } & \mathbf{\lambda}^{\left(
1,2,4\right)  } & \cdots & \mathbf{\lambda}^{\left(  n-2,n-1,n\right)  }\\
\downarrow & \downarrow & \cdots & \downarrow\\
A_{11}^{\left(  r\right)  \left(  1\right)  } & A_{11}^{\left(  r\right)
\left(  2\right)  } & \cdots & A_{11}^{\left(  r\right)  \left(  N\right)  }%
\end{array}
}}\right),
\]
where $A_{11}^{\left(  r\right)  \left(  i\right)  }\in\mathbb{R}^{6\times8}$.
A part of the constraint in $\left(  \ref{EQ130}\right)  $ is written simply
as: $A_{11}^{\left(  r\right)  \left(  r\right)  }\leftarrow B$.

Similarly, the matrix $A_{12}\in\mathbb{R}^{6N\times2N_{1}}$ will be written as%
\[
A_{12}=\left(
\begin{array}
[c]{c}%
A_{12}^{(1)}\\
A_{12}^{\left(  2\right)  }\\
\vdots\\
A_{12}^{\left(  N\right)  }%
\end{array}
\right),
\]
where $A_{12}^{(r)}\in$ $\mathbb{R}^{6\times8}$ transforms a vector
$w_{i,j,k}$. In turn, each matrix $A_{12}^{\left(  r\right)  }$ will be
divided into $A_{12}^{\left(  r\right)  ^{\prime}}$ and $A_{12}^{\left(
r\right)  ^{\prime\prime}}$ corresponding to the primary variables $u$ and $v$
respectively, i.e. $A_{12}^{\left(  r\right)  }=\left(
\begin{array}
[c]{cc}%
A_{12}^{\left(  r\right)  ^{\prime}} & A_{12}^{\left(  r\right)
^{\prime\prime}}%
\end{array}
\right)  $. Given an index $r\in\left\{  1,\ldots,N\right\}  $ we can generate
$A_{12}^{\left(  r\right)  ^{\prime}}$ and $A_{12}^{\left(  r\right)
^{\prime\prime}}$ by accessing columns $\iota\left(  i,j\right)  $,
$\iota\left(  i,k\right)  $, and $\iota\left(  j,k\right)  $. For primary
variables $u$ we have to make the following assignments%

\begin{align*}
\left(  A_{12}^{\left(  r\right)  ^{\prime}}\right)  _{1,\iota\left(
i,j\right)  }  & \leftarrow-1,\\
\left(  A_{12}^{\left(  r\right)  ^{\prime}}\right)  _{2,\iota\left(
i,k\right)  }  & \leftarrow-1,\\
\left(  A_{12}^{\left(  r\right)  ^{\prime}}\right)  _{3,\iota\left(
j,k\right)  }  & \leftarrow-1,
\end{align*}

and for primary variables $v$:%

\begin{align*}
\left(  A_{12}^{\left(  r\right)  ^{\prime\prime}}\right)  _{4,\iota\left(
i,j\right)  }  & \leftarrow-1,\\
\left(  A_{12}^{\left(  r\right)  ^{\prime\prime}}\right)  _{5,\iota\left(
i,k\right)  }  & \leftarrow-1,\\
\left(  A_{12}^{\left(  r\right)  ^{\prime\prime}}\right)  _{6,\iota\left(
j,k\right)  }  & \leftarrow-1.
\end{align*}
With this, we achieve that%
\[
A_{12}^{\left(  r\right)  }w=-I_{6\times6}w_{i,j,k}=-w_{i,j,k}.%
\]
Finally, the constraint in $\left(  \ref{EQ131}\right)  $ can be entirely written through  $A_{31}\in\mathbb{R}^{N\times8N}$ following the same
idea as for $A_{11}$. Specifically, $A_{31}$ is a stack of rows $A_{31}%
^{\left(  r\right)  }$ that transform the vector $\mathbf{\lambda}^{\left(
i,j,k\right)  }$ into a scalar:%
\[
A_{31}=\left(
\begin{array}
[c]{c}%
A_{31}^{(1)}\\
A_{31}^{\left(  2\right)  }\\
\vdots\\
A_{31}^{\left(  N\right)  }%
\end{array}
\right).
\]

Again these sub-matrices row $A_{31}^{\left(  r\right)  }$ are divided
according to  $\left(  i,j,k\right)  $:%
\[
A_{31}^{(r)}=\left(  \underset{\text{Part corresponding to the secondary
variables }\mathbf{\lambda}}{\underbrace{%
\begin{array}
[c]{cccc}%
\mathbf{\lambda}^{\left(  1,2,3\right)  } & \mathbf{\lambda}^{\left(
1,2,4\right)  } & \cdots & \mathbf{\lambda}^{\left(  n-2,n-1,n\right)  }\\
\downarrow & \downarrow & \cdots & \downarrow\\
A_{31}^{\left(  r\right)  \left(  1\right)  } & A_{31}^{\left(  r\right)
\left(  2\right)  } & \cdots & A_{31}^{\left(  r\right)  \left(  N\right)  }%
\end{array}
}}\right),
\]
where $A_{31}^{\left(  r\right)  \left(  i\right)  }\in\mathbb{R}^{1\times8}$.
So the matrix $A_{31}$ from the constraint $\left(  \ref{EQ131}\right)  $ is
written simply as $A_{31}^{\left(  r\right)  \left(  r\right)  }%
\leftarrow\left(  1,1,1,1,1,1,1,1\right)  $ for $r=1,\ldots,N$.

\subsubsection{Consistency Constraints}

Consistency constraints express the cancellation of a linear combination of
primary variables $u$ and $v$. In particular the consistency constraints can
be written as%
\begin{equation}
g_{1,2,3}\left(  w\right)  -g_{1,j,k}\left(  w\right)  =0\text{ with }\left(
j,k\right)  \neq\left(  2,3\right),  \label{EQ124}%
\end{equation}%
\begin{equation}
g_{2,1,3}\left(  w\right)  -g_{2,j,k}\left(  w\right)  =0\text{ with }\left(
j,k\right)  \neq\left(  1,3\right),  \label{EQ125}%
\end{equation}

\begin{equation}
g_{i,1,2}\left(  w\right)  -g_{i,j,k}\left(  w\right)  =0\text{ with }\left(
j,k\right)  \neq(1,2),\text{ }i\geq3,\label{EQ126}%
\end{equation}
or in matrix form as $A_{22}w=0$. \ Remembering that the number of primary
variables is $2N_{1}=2\binom{n}{2}$, we have that this will be the number of
columns of $A_{22}$. The number of rows in $A_{22}$ matches the number of
consistency constraints that we saw to be $N_{2}=n\left(  \binom{n-1}%
{2}-1\right)  $.

From the definition of $g_{i,j,k}\left(  w\right)  $ we have that the
combination of variables $u$ is the same as that of variables $v$; remember
that%
\[
g_{i,j,k}\left(  w\right)  =\frac{u_{ij}+u_{ik}-u_{jk}}{2}+\frac{v_{ij}%
+v_{ik}-v_{jk}}{2}.%
\]
This idea reduces many calculations since you simply have to create a matrix
$M=\left(  M_{1},M_{1}\right)  $ such that $M_{1}\in\mathbb{R}^{N_{2}\times
N_{1}}$ and $A_{22}=M$. The matrix $M$ will be filled sequentially by rows.
Initially, it is reset to zero.

To generate the consistency constraints as defined in the equations $\left(
\ref{EQ124}\right)  $, $\left(  \ref{EQ125}\right)  $, and $\left(
\ref{EQ126}\right)  $, we must create in first a set of indices $S=\left\{
1,2,\ldots,n\right\}  $ that will be traversed sequentially. In the i-th step
we remove the element $i$ from the set $S$ which gives us a set $\bar
{S}=S\diagdown\left\{  i\right\}  $, and it is from this set we generate the
pairs $\left(  j,k\right)  $ as they appear in the constraints of consistency.
For this, we generate an ordered array $C\in\mathbb{N}^{2\binom{n-1}{2}}$
according to the lexicographic order and including all the combinations of two
elements taken from $\bar{S}$. Note that when $i=1$, the first combination of
this list is $\left(  j_{1},k_{1}\right)  =\left(  2,3\right)  $ and this
corresponds to the subscripts of $g_{1,2,3}\left(  w\right)  $ in equation
$\left(  \ref{EQ124}\right)  $. Similarly, for $i=2$, the first ordered pair
in $C$ is $\left(  j_{1},k_{1}\right)  =\left(  1,3\right)  $, corresponding
to $g_{2,1,3}\left(  w\right)  $ in the equation $\left(  \ref{EQ125}\right)
$, and when $i\geq2$, $\left(  j_{1},k_{1}\right)  =\left(  1,2\right)  $ is
the first ordered pair in $C$ associated with $g_{i,1,2}\left(  w\right)  $ in
$\left(  \ref{EQ126}\right)  $.

Now all that is left is to generate the rest of the indexes that appear in the
consistency constraints. For it, we go through $C$ from the second ordered pair
onwards. Let $\left(  j,k\right)  $ be this ordered pair, and suppose we are at
step $r$. We introduce the part of $g_{i,j_{1},k_{1}}\left(  w\right)  $
corresponding to the variables $u$:%
\begin{align*}
M_{r,\iota\left(  i,j_{1}\right)  } &  \leftarrow M_{r,\iota\left(
i,j_{1}\right)  }-1,\\
M_{r,\iota\left(  i,k_{1}\right)  } &  \leftarrow M_{r,\iota\left(
i,k_{1}\right)  }-1,\\
M_{r,\iota\left(  j_{1},k_{1}\right)  } &  \leftarrow M_{r,\iota\left(
j_{1},k_{1}\right)  }+1.
\end{align*}
Similarly, we do with the part of $g_{i,j,k}\left(  w\right)  $ that
corresponds to the variables $u$:%
\begin{align*}
M_{r,\iota\left(  i,j\right)  } &  \leftarrow M_{r,\iota\left(  i,j\right)
}-1,\\
M_{r,\iota\left(  i,k\right)  } &  \leftarrow M_{r,\iota\left(  i,k\right)
}-1,\\
M_{r,\iota\left(  j,k\right)  } &  \leftarrow M_{r,\iota\left(  j,k\right)
}+1.
\end{align*}
Once we have completed this task we have filled in the part of $A_{22}$
corresponding to $u$, that is to say, $A_{22}=\left(  M_{1},\cdot\right)  $. To
complete the part corresponding to v we simply replicate $M_{1}$ in $A_{22}$,
i.e. $A_{22}=\left(  M_{1},M_{1}\right)  $. The construction process of the
matrix $M$ is shown in Algorithm \ref{AlgConsistency}.

\begin{algorithm}
\caption{Create consistency constraints}
\begin{algorithmic}[2]
\State $M\gets 0_{N_{2}\times2N_{1}}$
\State $S\gets \left\{  1,2,\ldots,n\right\}  $ \Comment{Index Set}
\State $r\gets 0$ \Comment{Row index of $M$}

\For{$i\gets 1, n$}
    \State $\bar{S}\gets S\diagdown\left\{  i\right\}  $
    \State Generate an ordered array $C$ of the elements of $\bar{S}$ taken two by two.
    \State $\left(j_{1},k_{1}\right)  \gets C_{1}$ \Comment{first ordered pair in
$C$}
    \For{$l\gets 2, n$}
        \State $r\gets r+1$
        \State $\left(  j,k\right)  \leftarrow C_{l}$ \Comment{l-th ordered pair in $C$}
        \State $M_{r,\iota\left(  i,j_{1}\right)  }\gets M_{r,\iota\left(
i,j_{1}\right)  }-1$

        \State $M_{r,\iota\left(  i,k_{1}\right)  }\gets M_{r,\iota\left(
i,k_{1}\right)  }-1$

        \State $M_{r,\iota\left(  j_{1},k_{1}\right)  }\gets M_{r,\iota\left(
j_{1},k_{1}\right)  }+1$

        \State $M_{r,\iota\left(  i,j\right)  }\gets M_{r,\iota\left(  i,j\right)  }-1$

        \State $M_{r,\iota\left(  i,k\right)  }\gets M_{r,\iota\left(  i,k\right)  }-1$

        \State $M_{r,\iota\left(  j,k\right)  }\leftarrow M_{r,\iota\left(  j,k\right)  }+1$.
    
    \EndFor

\EndFor

\State $M_{1}\gets M\left[  1,\ldots,N_{2};1,\ldots N_{1}\right]  $
\State $M\left[  1,\ldots,N_{2};N_{1}+1,\ldots2N_{1}\right]  \gets M_{1}$

\end{algorithmic}

\label{AlgConsistency}
\end{algorithm}

In Algorithm \ref{AlgConsistency} the submatrix formed by the rows $\left\{
1,\ldots,N_{2}\right\}  $ and by the columns $\left\{  1,\ldots,N_{1}\right\}
$ is denoted by $M\left[  1,\ldots,N_{2};1,\ldots N_{1}\right]  $. Similarly
for $M\left[  1,\ldots,N_{2};N_{1}+1,\ldots2N_{1}\right]  $.

\begin{algorithm}
\caption{Create improved consistency constraints}
\begin{algorithmic}[2a]
\State $M\gets 0_{N_{2}\times2N_{1}}$
\State $S\gets \left\{  1,2,\ldots,n\right\}  $ \Comment{Index Set}
\State $r\gets 0$ \Comment{Row index of $M$}

\For{$i\gets 1, n$}
    \State $\bar{S}\gets S\diagdown\left\{  i\right\}  $
    \State $j_{1}\gets \bar{S}\left(  1\right)$
    \State $k_{1}\gets \bar{S}\left(  2\right)$
    \For{$l_{1}\gets 1, n-2$}
        \For{$l_{2}\gets l_{1}+1, n-1$}
            \If{$l_{1}>1$ or $l_{2}>2$}
                \State $r\gets r+1$
                \State $j\gets\bar {S}\left(  l_{1}\right)$
                \State $k\gets\bar{S}\left(  l_{2}\right)  $
                \State $\vdots$
                \State Here the same assignments as in Algorithm \ref{AlgConsistency}
                \State $\vdots$
            \EndIf
        \EndFor
    \EndFor
\EndFor

\State $M_{1}\gets M\left[  1,\ldots,N_{2};1,\ldots N_{1}\right]  $
\State $M\left[  1,\ldots,N_{2};N_{1}+1,\ldots2N_{1}\right]  \gets M_{1}$

\end{algorithmic}
\label{AlgConsistencyImproved}
\end{algorithm}

When $n$ is large, the generation of the array $C$ of elements of $\bar{S}$
taken two by two can take up a lot of memory, causing the system to blow up.
To avoid this problem, the combinations must be generated one by one.
Algorithm \ref{AlgConsistencyImproved} is a variant of Algorithm \ref{AlgConsistency} that
sequentially generates the combinations using the indices  $l_{1}\in\left[
1,\ldots,n-2\right]  $ and  $l_{2}\in\left[  l_{1}+1\ldots,n-1\right]  $. 

Finally, it was already commented in the previous sections that $A_{21}$ is a
null matrix since the secondary variables $\mathbf{\lambda}$ do not intervene
in the definition of the consistency constraints.

\subsection{Linear Optimizer}

Once the binary quadratic problem has been translated into a linear
optimization problem, it only remains to invoke a standard optimizer. Here it
should be noted, although it is a well-known fact, that the linear optimizer
works in polynomial time. There are many interior-point methods to solve the
linear programming problem, although they are all improvements to the
ellipsoid method due to Khachiyan. The objective of this technical note is not
to present an efficient method but to demonstrate that the procedure of
transformation to a linear program is successful. For this reason, it is
sufficient for the implementation to simply invoke a standard resolver. The
problem \textbf{(LP}$_{n}$\textbf{)} will be written compactly as%
\begin{align*}
\min\tilde{f}\left(  w\right)    & =\tilde{c}^{T}w\\
\text{subject to }\tilde{A}\left(
\begin{array}
[c]{c}%
\lambda\\
w
\end{array}
\right)    & =\tilde{b}\\
\lambda & \geq0
\end{align*}
The implementation was done in MATLAB and the \texttt{linprog}
function was invoked with parameters $\tilde{c}$, $\tilde{A}$, and $\tilde{b}%
$, and the constraint $\lambda\geq0$. Here we are not interested in the
implementation details of linprog as they are not relevant to ensure that the
proposed method works correctly, we are only interested in the result produced
by the linear optimizer.

For a future improvement in the implementation, we suggest the following stop
condition: if all the components $w_{i}$ are distant from the minimum
$w_{i}^{\ast}$ in less than $\frac{1}{4}$ the optimizer must stop. This is
because the primary variables u and v move in the domains $\left[  0,2\right]
^{n}$ and $\left[  0,\frac{1}{2}\right]  ^{n}$ respectively. The presented
formulation is matrix for reasons of clarity and simplicity, however this
involves having a somewhat sparsed structure with many zero entries.

\section{Experiment Description}

The implementation of the algorithm presented above was done in MATLAB (the
code can be found in the supplementary material). Experimentation was performed for
arbitrary dimensions of the problem (up to $n=30$, due to memory limitations
in MATLAB). Both matrix $Q$ and vector $b$ were chosen arbitrarily in a range of
values $\left[  lv,uv\right]  =\left[  -50,50\right]  $ (this range can be
freely modified by the experimenter). Also, the domain of values is allowed to
be the reals or the integers.

Regarding the success condition of the experiment, a comparison is made of the
optimal $\tilde{c}^{T}w^{\ast}$ obtained from \texttt{linprog} with that
produced by brute force, $f\left(  x^{\ast}\right)  $ (exploring all the
possible binary combinations for the original problem). For this, we set a
dimension $\epsilon>0$ and compare  $\tilde{c}^{T}w^{\ast}$ with  $f\left(
x^{\ast}\right)  $: if  $\left\vert \tilde{c}^{T}w^{\ast}-f\left(  x^{\ast
}\right)  \right\vert <\epsilon$ the result is successful.

The method has been tested in Matlab R2016a 64-bit under Windows 10 (64-bit) Intel\textregistered\ Core\texttrademark\ i5-4300U CPU@ 1.90 GHz. RAM 8.00GB. After five million experiments with variability in the dimension of the
problem and randomness in $Q$ and $b$, we did not record any unsatisfactory
results, which suggests that this method is correct.

\section{Discussion}

In this paper an algorithm has been developed to find the global minimum of a
quadratic binary function in polynomial time. Specifically, the computational
complexity of the algorithm when using the Vaidya linear optimizer is
$O\left(  n^{\frac{15}{2}}\right)  $. This bound is very conservative but it
is enough to prove that the problem is in class $P$. The reduction of the
complexity exponent, and therefore the speed of resolution of the method,
strongly depends on the advances that occur in linear optimizers as well as of
technological aspects such as parallel computing.

A great advantage of this algorithm is its modularity according to the
dimension: the definition of the consistency and convexity constraints through
$A$ and $b$ are fixed for all problems of a given dimension $n$. This means
that these matrices can be precalculated for different sizes of problems and
stored in a data file (physically it could be stored in a fast access ROM
memory). To test the algorithm for problems of dimension $n$, this
precalculated information is loaded into memory and the linear optimizer is
directly invoked.

The algorithm has been implemented in MATLAB and has been verified generating
more than five million matrices of arbitrary dimension up to $30$ with random
entries in the range $\left[  -50,50\right]  $. Actually, the algorithm is
designed for any dimension $n$, however, tests beyond $30$ can be done
considering the storage limits of matrices in MATLAB. In particular, matrix
$A$ of the problem \textbf{(LP}$_{n}$\textbf{)} has dimension  $\left(
8N+2N_{1}\right)  \times\left(  7N+N_{1}\right)  \in O\left(  n^{6}\right)  $
where $N=\binom{n}{3}$, $N_{1}=\binom{n}{2}$ y $N_{2}=n\left(  \binom{n-1}%
{2}-1\right)  $.. This matrix is memory intensive if its sparse condition is
not exploited. Therefore, the algorithm can be made to work more efficiently
by storing $A$ as a sparse matrix (matrix A is extremely sparse so the memory
demand is not as great as it might seem a priori from the dimension). If you
do not exploit this property, with MATLAB's default options, there is a memory
limit on matrices of about $2^{48}=256\cdot10^{12}$ entries. The size of
matrix $A$ (number of entries) according to dimension $n$ of the problem as
indicated in Figure \ref{FIGMemory} where a conservative memory limit of $10^{12}$ is
indicated. Thus, for this threshold, it is theoretically possible to calculate
dimension problems up to $n=93$. In practice, the situation is more dramatic
if sparse matrices are not handled. In particular, for dimension $n=30$,
MATLAB begins to give space reservation problems requiring $10.1$ GB of
storage. For all the reasons stated above it is recommended to exploit the
sparsity of the matrices.%

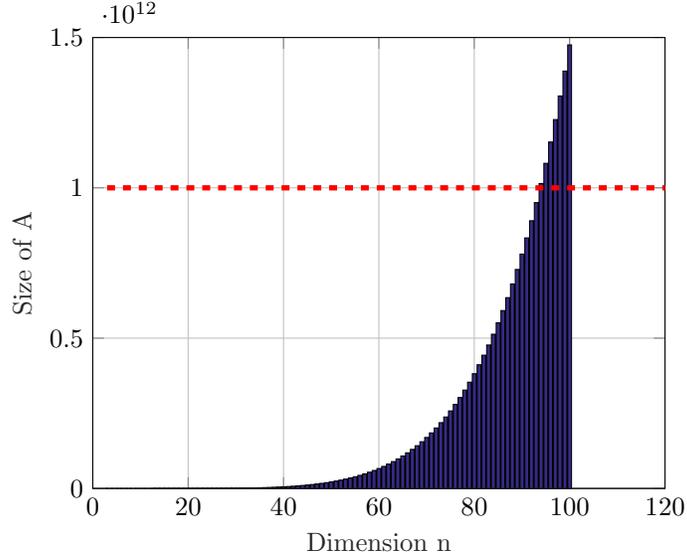
\begin{figure}
[hptb]
\begin{center}
%
%
\definecolor{mycolor1}{rgb}{0.20810,0.16630,0.52920}%
\begin{tikzpicture}

\begin{axis}[%
width=7.607cm,
height=6cm,
at={(0cm,0cm)},
scale only axis,
bar shift auto,
xmin=0,
xmax=120,
xlabel style={font=\color{white!15!black}},
xlabel={Dimension n},
ymin=0,
ymax=1500000000000,
ylabel style={font=\color{white!15!black}},
ylabel={Size of A},
axis background/.style={fill=white},
xmajorgrids,
ymajorgrids
]
\addplot[ybar, bar width=0.8, fill=mycolor1, draw=black, area legend] table[row sep=crcr] {%
3	101\\
4	1238\\
5	7010\\
6	26615\\
7	78911\\
8	197596\\
9	437508\\
10	882045\\
11	1651705\\
12	2913746\\
13	4892966\\
14	7883603\\
15	12262355\\
16	18502520\\
17	27189256\\
18	39035961\\
19	54901773\\
20	75810190\\
21	102968810\\
22	137790191\\
23	181913831\\
24	237229268\\
25	305900300\\
26	390390325\\
27	493488801\\
28	618338826\\
29	768465838\\
30	947807435\\
31	1160744315\\
32	1412132336\\
33	1707335696\\
34	2052261233\\
35	2453393845\\
36	2917833030\\
37	3453330546\\
38	4068329191\\
39	4772002703\\
40	5574296780\\
41	6485971220\\
42	7518643181\\
43	8684831561\\
44	9998002498\\
45	11472615990\\
46	13124173635\\
47	14969267491\\
48	17025630056\\
49	19312185368\\
50	21849101225\\
51	24657842525\\
52	27761225726\\
53	31183474426\\
54	34950276063\\
55	39088839735\\
56	43627955140\\
57	48598052636\\
58	54031264421\\
59	59961486833\\
60	66424443770\\
61	73457751230\\
62	81100982971\\
63	89395737291\\
64	98385704928\\
65	108116738080\\
66	118636920545\\
67	129996638981\\
68	142248655286\\
69	155448180098\\
70	169652947415\\
71	184923290335\\
72	201322217916\\
73	218915493156\\
74	237771712093\\
75	257962384025\\
76	279562012850\\
77	302648179526\\
78	327301625651\\
79	353606338163\\
80	381649635160\\
81	411522252840\\
82	443318433561\\
83	477136015021\\
84	513076520558\\
85	551245250570\\
86	591751375055\\
87	634708027271\\
88	680232398516\\
89	728445834028\\
90	779473930005\\
91	833446631745\\
92	890498332906\\
93	950767975886\\
94	1014399153323\\
95	1081540210715\\
96	1152344350160\\
97	1226969735216\\
98	1305579596881\\
99	1388342340693\\
100	1475431654950\\
};
\addplot[forget plot, color=white!15!black] table[row sep=crcr] {%
0	0\\
120	0\\
};
\addplot [color=red, dashed, line width=2.0pt, forget plot]
  table[row sep=crcr]{%
3	1000000000000\\
120	1000000000000\\
};
\end{axis}
\end{tikzpicture}%
\caption{Size of $A$ in terms of $n$. The memory threshold for non-sparse matrices is set  at $10^{12}$ and is represented by a dashed red line.}
\label{FIGMemory}%
\end{center}
\end{figure}

Although the problem has polynomial computational complexity in both time and
space when $n$ is large this takes time to compute. For example, $n=100$
implies a bound proportional to $10^{15}$. On a small scale, polynomial-time
algorithms generally do not perform well compared to metaheuristic methods.
Performance is found when handling large dimensions. For large-scale UBQP
problems, it is necessary to compile the code avoiding the use of interpreted
MATLAB code. On the other hand, the system should be parallelized as much as
possible using multi-core processor architectures and offloading computing in
graphics processing units (GPUs). Parallelization should avoid the dependency
problem. This requires that the counter variable r within the for-loops in
algorithms \ref{Alg:MatrixT}, \ref{AlgConsistency}, and
\ref{AlgConsistencyImproved} should be made explicit as a dependent function
of $i$ and $j$, $r\left(  i,j\right)  $, thus avoiding the increment
$r\leftarrow r+1$.

Finally, concerning the experiments, it should be noted that the method
solution has been compared with the brute force solution. For medium and large
dimensions this way of proceeding does not work. For example for $n=100$ the
brute force method requires $2^{100}$ iterations, approximately one followed
by thirty zeros. Therefore, in the future, it would be necessary to design
verification experiments that compared solutions by methods of different nature.

\section{Appendix}

We rewrite the problem \textbf{(LP}$_{n}$\textbf{)} in secondary variables as:%

\[
\text{\textbf{(LP}}_{n}^{\prime}\text{\textbf{)}: }\left\{
\begin{array}
[c]{c}%
\min d^{T}\lambda\\
\text{s.t. }F\lambda=g\\
\lambda\in\left[  0,1\right]  ^{8N}%
\end{array}
\right.
\]
where $F\in\mathbb{R}^{\left(  N_{2}+N\right)  \times8N_{1}}$. In the problem
\textbf{(LP'}$_{n}$\textbf{)} box restrictions of the form $0\leq\lambda
_{i}\leq1$ have been added. The upper bound on $\lambda$ is redundant since
for each $\mathcal{C}^{\left(  i,j,k\right)  }$ we have that%
\[
\sum_{l=1}^{8}\lambda_{l}^{\left(  i,j,k\right)  }=1\text{,}%
\]
and $\lambda_{l}^{\left(  i,j,k\right)  }\geq0$ for $l=1,\ldots,8$. This
implies that $\lambda_{l}^{\left(  i,j,k\right)  }\leq1$. However, this
redundancy is convenient to determine the vertices of the feasible region
$\mathcal{C}_{n}$.

\begin{lemma}
\label{LEMA10}Let $V$ be the set of vertices of the hypercube $\mathcal{H}%
_{n}$. For each $p\in V$, there is a $\lambda\in\left\{  0,1\right\}  ^{8N}$
that is in the feasible region of \textbf{(LP}$_{n}^{\prime}$\textbf{)}.
Conversely, for every $\lambda\in\left\{  0,1\right\}  ^{8N}$ in the feasible
region of \textbf{(LP}$_{n}^{\prime}$\textbf{)} there exists a vertex $p$ of
$\mathcal{H}_{n}$ such that $\phi\left(  p\right)  =w\in\mathcal{C}_{n}$.
\end{lemma}

\begin{proof}
For each triple $\left(  i,j,k\right)  $ we have that $\phi_{i,j,k}\left(
p_{i,j,k}\right)  $ is a vertex of $\mathcal{C}_{i,j,k}$. Consequently, there
exists a $\lambda^{\left(  i,j,k\right)  }\in\mathbb{R}^{8}$, with one and the
remainder zeros. These vertices have a preimage $p_{i,j,k}$ through
$\phi_{i,j,k}$ given by%
\[
p_{i,j,k}=\left(
\begin{array}
[c]{c}%
x_{i}\\
x_{j}\\
x_{k}%
\end{array}
\right)  \text{,}%
\]
with
\[
x_{i}=\frac{u_{ij}+v_{ij}+u_{ik}+v_{ik}-u_{jk}-v_{jk}}{2}\in\left\{
0,1\right\}  \text{.}%
\]
Since $x_{i}$ must be the same for each convex $\mathcal{C}_{3}^{\left(
i,j,k\right)  }$, the consistency restrictions must be satisfied at the point
$\phi\left(  p\right)  $. Hence, $\phi\left(  p\right)  \in\mathcal{C}_{n}$,
and there is a $\lambda\in\left\{  0,1\right\}  ^{8N}$ in the feasible region
of \textbf{(LP}$_{n}^{\prime}$\textbf{)}. \newline For the reciprocal, we have
that $\lambda^{\left(  i,j,k\right)  }$ is a vector of the standard base of
$\mathbb{R}^{8}$, that is, the vector of primary variables $w$ is such that
$w_{i,j,k}$ is a vertex of the convex $C_{3}^{\left(  i,j,k\right)  }$. For
this vertex there exists a vertex $p_{i,j,k}$ in $\mathcal{H}_{3}^{\left(
i,j,k\right)  }$ such that
\[
\phi_{i,j,k}\left(  p_{i,j,k}\right)  =w_{i,j,k}\text{.}%
\]
The consistency constraints guarantee that the coordinate $x_{i}$ at two
points $p_{i,j,k}$ and $p_{i,j^{\prime},k^{\prime}}$ is the same. Hence, there
exists a unique preimage $p$ for $w$ through $\phi$ such that $p\in V$.
\end{proof}

In the previous Lemma, despite the fact that for each $p\in V$, the vector
$\lambda\in\left\{  0,1\right\}  ^{8N}$ does not have to be unique, the
associated vector of primary variables is unique through the transformation
$w=B\lambda$ (see example \ref{EJEMPLO1} of Section 3.1). This shows that
there is a bijection between the vertices of $\mathcal{H}_{n}$ and those of
$\mathcal{C}_{n}$ through $\phi$. And therefore, the number of vertices of
$\mathcal{C}_{n}$ is just $2^{n}$.

\begin{theorem}
Let $w^{\ast}$ be the point in $\mathcal{C}_{n}$ where the minimum of the
problem ($P_{n}^{\prime}$) is reached and let $x^{\ast}$ be the point in
$\mathcal{H}_{n}$ for the minimum of ($P_{n}$), then $\tilde{f}\left(
w^{\ast}\right)  =f\left(  x^{\ast}\right)  $.
\end{theorem}

\begin{proof}
According to Lemma \ref{LEMA10}, the set of vertices of $\mathcal{C}_{n}$ is
$\left\{  \phi\left(  p\right)  :p\in V\right\}  $ (where $V$ is the set of
vertices of the hypercube $\mathcal{H}_{n}$). We know that the minimum of the
problem of ($P_{n}^{\prime}$) is reached at a vertex of $\mathcal{C}_{n}$, so
there will be a $p\in V$ such that $\phi\left(  p\right)  =w^{\ast}$. On the
other hand, the minimum of $f$ over $\mathcal{H}_{n}$ is achieved at $x^{\ast
}\in\mathcal{H}_{n}$, so that%
\begin{equation}
f\left(  x^{\ast}\right)  \leq f\left(  p\right)  \text{.}\label{EQQ104}%
\end{equation}
The connection between $f$ and $\tilde{f}$ leads to%
\begin{equation}
\tilde{f}\left(  \phi\left(  x^{\ast}\right)  \right)  =\tilde{c}^{T}%
\phi\left(  x^{\ast}\right)  =\left(  c^{T}T_{n}\right)  \left(  E_{n}%
\alpha\left(  x^{\ast}\right)  \right)  =c^{T}\alpha\left(  x^{\ast}\right)
=f\left(  x^{\ast}\right)  \label{EQQ105}%
\end{equation}
where $\phi\left(  x^{\ast}\right)  $ is a vertex of $\mathcal{C}_{n}$ because
$x^{\ast}\in V$, and to%
\begin{equation}
f\left(  p\right)  =c^{T}\alpha\left(  p\right)  =\left(  c^{T}T_{n}\right)
\left(  E_{n}\alpha\left(  p\right)  \right)  =\tilde{c}^{T}\phi\left(
p\right)  =\tilde{f}\left(  \phi\left(  p\right)  \right)  \text{.}%
\label{EQQ106}%
\end{equation}
\newline According to $\left(  \ref{EQQ104}\right)  $ and $\left(
\ref{EQQ106}\right)  $,%
\[
f\left(  x^{\ast}\right)  \leq f\left(  p\right)  =\tilde{f}\left(  w^{\ast
}\right)  \text{.}%
\]
\newline Since the minimum of $\tilde{f}$ over $\mathcal{C}_{n}$ is attained
at $w^{\ast}\in\mathcal{C}_{n}$, and accounting for $\left(  \ref{EQQ105}%
\right)  $, we have that%
\[
f\left(  x^{\ast}\right)  =\tilde{f}\left(  \phi\left(  x^{\ast}\right)
\right)  \geq\tilde{f}\left(  w^{\ast}\right).
\]
Henceforth, $f\left(  x^{\ast}\right)  =\tilde{f}\left(  w^{\ast}\right)  $.
\end{proof}

\bibliographystyle{elsarticle-num}
\bibliography{bibliografia}

\end{document}